\documentclass[12pt,oneside,reqno]{scrartcl}
\usepackage[T1]{fontenc}

\usepackage{palatino}
\usepackage{microtype}

\usepackage{amsmath,amsthm,amssymb,exscale}
\usepackage{bbm} 
\usepackage{mathrsfs} 

\usepackage{todonotes}

\usepackage{etoolbox}

\usepackage{mathtools} 

\usepackage{paralist}

\usepackage{listings}

\lstnewenvironment{MacaulayInput}{\lstset{
	basicstyle=\upshape\ttfamily\small,
        breaklines=true,
        columns=flexible,
        breakautoindent=false,
        breakindent=0pt,
        xleftmargin=4ex,
        numbers=none,
        aboveskip=\smallskipamount,
        belowskip=\smallskipamount,
	frame=leftline,
	framerule=0.3mm,
}}{}

\numberwithin{equation}{section}
\numberwithin{figure}{section}

\theoremstyle{plain}
	\newtheorem{thm}{Theorem}
	\newtheorem{lem}[thm]{Lemma}
	\newtheorem{example}[thm]{Example}
	\newtheorem{cor}[thm]{Corollary}
	
	\newtheorem{quest}[thm]{Question}
	\newtheorem{prop}[thm]{Proposition}
\theoremstyle{definition}
	\newtheorem{defn}[thm]{Definition}
	\newtheorem{rem}[thm]{Remark}

\usepackage{graphicx}
\usepackage{xcolor}

\usepackage{tabularx,booktabs}

\usepackage[numbers,sort&compress]{natbib}

\definecolor{links}{rgb}{0,0.3,0}
\usepackage[colorlinks=true,citecolor=blue,urlcolor=links,breaklinks=true]{hyperref}

\newcommand{\U}{\mathcal{U}}
\newcommand{\F}{\mathcal{F}}
\newcommand{\G}{\mathcal{G}}
\newcommand{\BP}{\mathcal{P}}

\newcommand{\Dim}{d}
\newcommand{\Den}{\mathsf{D}}

\newcommand{\FeynEps}{\epsilon} 
\newcommand{\sdc}{\omega}

\newcommand{\uv}[1]{\mathsf{e}_{#1}}

\newcommand{\Loops}{L}
\newcommand{\AllMoms}{M}
\newcommand{\EdgeMom}{k}
\newcommand{\LoopMom}{\ell}
\newcommand{\ExtMom}{p}
\newcommand{\ExtMoms}{E}
\newcommand{\FI}{\mathcal{I}}
\newcommand{\FImel}{\widetilde{\mathcal{I}}} 
\newcommand{\Gram}{\mathsf{Gr}}

\newcommand{\Graph}[2][1.0]{\vcenter{\hbox{\includegraphics[scale=#1]{Graphs/#2}}}}

\newcommand{\bilabels}{\Theta}
\newcommand{\spMat}{\mathcal{A}}

\newcommand{\loopMat}{\Lambda}
\newcommand{\loopLin}{Q}
\newcommand{\loopConst}{J}

\newcommand{\bigo}[1]{\mathcal{O}\left( #1 \right)}

\newcommand{\set}[1]{\left\{ #1 \right\}}
\newcommand{\setexp}[2]{\left\{ #1 \colon #2 \right\}}
\newcommand{\restrict}[2]{\left.#1 \right|_{#2}}
\newcommand{\defas}{\mathrel{\mathop:}=}
\newcommand{\abs}[1]{\left\lvert #1 \right\rvert}
\newcommand{\norm}[1]{\left\| #1 \right\|}

\newcommand{\injects}{\lhook\joinrel\longrightarrow}
\newcommand{\surjects}{\relbar\joinrel\twoheadrightarrow}

\newcommand{\prOrth}[1]{\mathrm{pr}_{#1}^{\perp}}
\DeclareMathOperator{\lin}{lin}

\newcommand{\commu}[2]{\left[ #1, #2 \right]}

\renewcommand{\Re}{\operatorname{Re}}
\newcommand{\dd}[1][]{\mathrm{d}^{#1}} 
\newcommand{\R}{\mathbbm{R}}
\newcommand{\C}{\mathbbm{C}}
\newcommand{\N}{\mathbbm{N}}
\newcommand{\Z}{\mathbbm{Z}}
\newcommand{\Q}{\mathbbm{Q}}

\newcommand{\M}{\mathscr{M}}
\newcommand{\Nmod}{\mathscr{N}}
\newcommand{\Mfrak}{\mathfrak{M}}
\newcommand{\WeylT}[1]{\mathcal{D}^{#1}}

\newcommand{\Aff}{\mathbbm{A}}

\newcommand{\Lef}{\mathbbm{L}}
\newcommand{\PP}{\mathbbm{P}}
\newcommand{\Vanish}{\mathbbm{V}}

\newcommand{\Gm}[1][]{\mathbbm{G}_{m\ifstrequal{#1}{}{}{,#1}}}
\newcommand{\GmA}{\iota}

\newcommand{\GL}[2]{\mathsf{GL}_{#1}( #2 )}
\newcommand{\NewPol}[1]{\mathsf{NP}\left(#1\right)}
\DeclareMathOperator{\conv}{conv}
\DeclareMathOperator{\Vol}{Vol}

\newcommand{\OO}{\mathcal{O}}
\DeclareMathOperator{\DR}{DR}
\newcommand{\Koszul}[3][]{K^{#1}(#2;#3)}

\newcommand{\WS}[1]{\mathrm{WS}_{#1}}
\newcommand{\Sun}[1]{\mathrm{S}_{#1}}

\newcommand{\cupdot}{\mathbin{\dot{\cup}}}

\DeclareMathOperator{\sign}{sgn}

\DeclareMathOperator{\coldet}{col-det}

\newcommand{\PlusD}[1]{\hat{\mathbf{#1}}^{+}} 
\newcommand{\Minus}[1]{\mathbf{#1}^{-}} 
\newcommand{\nOp}[1]{\mathbf{n}_{#1}} 

\newcommand{\momIBP}[2]{\mathbf{o}_{#2}^{#1}}
\newcommand{\parIBP}[2]{\widetilde{\mathbf{O}}_{#2}^{#1}}
\newcommand{\shiftIBP}[2]{\mathbf{O}_{#2}^{#1}}

\newcommand{\hide}[1]{}

\newcommand{\isomorph}{\cong}
\newcommand{\qiso}{\simeq}
\DeclareMathOperator{\id}{id}

\newcommand{\Transpose}{\intercal}

\newcommand{\Mellin}[2][]{\mathcal{M}\ifstrequal{#1}{}{}{^{#1}}\!\left\{ #2 \right\}}

\newcommand{\Weyl}[1]{{A}^{#1}}
\newcommand{\Shifts}[1]{{S}^{#1}}

\newcommand{\iu}{\mathrm{{i}}}

\DeclareMathOperator{\Realteil}{Re}
\newcommand{\tp}{\otimes}

\newcommand{\NoM}[1]{\mathfrak{C}\ifstrequal{#1}{}{}{\left(#1\right)}}
\newcommand{\NoMtop}[1]{\widehat{\mathfrak{C}}\left(#1\right)}

\DeclareMathOperator{\Ann}{Ann}

\DeclareMathOperator{\Mom}{Mom} 

\newcommand{\Singular}{\href{https://www.singular.uni-kl.de/}{\textsc{Singular}}}
\newcommand{\JaxoDraw}{\texttt{\textup{JaxoDraw}}}
\newcommand{\Axodraw}{\texttt{\textup{Axodraw}}}
\newcommand{\Azurite}{\href{https://bitbucket.org/yzhphy/azurite/}{\textsc{Azurite}}}
\newcommand{\Mint}{\href{http://www.inp.nsk.su/~lee/programs/LiteRed/\#utils}{\texttt{Mint}}}
\newcommand{\Reduze}{\href{https://reduze.hepforge.org/}{\texttt{Reduze}}}
\newcommand{\Macaulay}{\href{http://www.macaulay2.com/}{\textit{Macaulay2}}}
\newcommand{\FeynArts}{\textsl{FeynArts}}

\newcommand*\samethanks{\footnotemark[\value{footnote}]}

\title{Feynman integral relations from parametric annihilators}
\author{%
	Thomas Bitoun\thanks{University of Oxford}
\and
	Christian Bogner\thanks{Humboldt-Universit\"{a}t zu Berlin}
\and
	Ren\'e~Pascal Klausen\samethanks
\and
	Erik Panzer\thanks{All Souls College, University of Oxford}
}

\begin{document}
\maketitle
\begin{abstract}
	We study shift relations between Feynman integrals via the Mellin transform through parametric annihilation operators. These contain the momentum space integration by parts relations, which are well-known in the physics literature.
	Applying a result of Loeser and Sabbah, we conclude that the number of master integrals is computed by the Euler characteristic of the Lee-Pomeransky polynomial. We illustrate techniques to compute this Euler characteristic in various examples and compare it with numbers of master integrals obtained in previous works.
\end{abstract}

\section{Introduction}
\label{sec:intro}%

At higher orders in perturbative quantum field theory, the computation of observables via Feynman diagrams involves a rapidly growing number of Feynman integrals.
Fortunately, the number of integrals which need to be computed explicitly can be reduced drastically by use of linear relations 
\begin{equation*}
	\sum_i c_{i} \FI_{i} = 0
	\tag{$\ast$}
\end{equation*}
between different Feynman integrals $\FI_{i}$, with coefficients $c_{i}$ that are rational functions of the space-time dimension $\Dim$ and the kinematic invariants characterizing the physical process (masses and momenta of elementary particles).

The most commonly used method to derive such identities is the integration by parts (IBP) method introduced in \cite{CheTka,Tkachov:calculability-4loop}. In this approach, the relations $(\ast)$ are obtained as integrals of total derivatives in the momentum-space representation of Feynman integrals. Combining these relations, any integral of interest can be expressed as a linear combination of some finite, preferred set of \emph{master integrals}.\footnote{%
	For different applications, various criteria for choosing the master integrals have been suggested. These include uniform transcendentality \cite{Henn:MultiloopSimple}, finiteness \cite{ManteuffelPanzerSchabinger:QuasiFinite} or finiteness of coefficients \cite{ChetyrkinFaisstSturmTentyukov:epsilonfinite}.
}
Laporta's algorithm \cite{Laporta:HighPrecision} provides a popular approach to obtain such reductions, and various implementations of it are available \cite{AnastasiouLazopoulos:Automatic,KIRA,ManteuffelStuderus:Reduze2,Smirnov:FIRE,Smirnov:FIRE5,Fire4LiteRed,Studerus:Reduze}.
However, the increase of complexity of todays computations has recently motivated considerable theoretical effort to improve our understanding of the IBP approach
and the efficiency of automated reductions \cite{Grozin:IBP,Lee:GroupStructureIBP,Lee:LiteRed,Lee:LiteRed1.4,ManteuffelSchabinger:Novel,RuijlUedaVermaseren:Diamond,ApplyingGroebner,Zhang:IBPfromDiffGeo,Ita:2lDecomposition}. This includes a method by Baikov \cite{Baikov:ExplicitSolutions-3loop,Baikov:ExplicitSolutions-multiloop,SmirnovSteinhauser:SolvingRecurrence,Baikov:PracticalCriterion}, which is based on a parametric representation of Feynman integrals.

It is interesting to ask for the number of master integrals which remain after such reductions. This number provides an estimate for the complexity of the computation and informs the problem of constructing a basis of master integrals by an Ansatz. In the recent literature, algorithms to count master integrals were proposed and implemented in the computer programs {\Mint} \cite{LeePom} and {\Azurite} \cite{GeorgoudisLarsenZhang:Azurite}. 

We propose an unambiguous definition for the number of master integrals as the dimension of an appropriate vector space. This definition and our entire discussion are independent of the method of the reduction. The main result of this article shows that this number is a well understood topological invariant: the Euler characteristic of the complement of a hypersurface $\set{\G=0}$ associated to the Feynman graph. 
Therefore, many powerful tools are available for its computation.

To arrive at our result, we follow Lee and Pomeransky \cite{LeePom} and view Feynman integrals as a Mellin transform of $\G^{-\Dim/2}$, where $\G$ is a certain polynomial in the Schwinger parameters. Each of these parameters corresponds to a denominator (inverse propagators or irreducible scalar products) of the momentum-space integrand.
The classical IBP relations relate Feynman integrals which differ from each other by integer shifts of the exponents of these denominators.
As Lee \cite{Lee:ModernTechniques} and Baikov \cite{Baikov:ExplicitSolutions-multiloop} pointed out, such shift relations correspond to annihilation operators of the integrands of parametric representations.\footnote{%
	Tkachov's idea \cite{Tkachov:WhatsNext} to insert Bernstein-Sato operators in the integrand with two Symanzik polynomials was used for numerical computations of one- and two-loop integrals 
\cite{FerrogliaPassarinoPasseraUccirati:Frontier,FerrogliaPasseraPassarinoUccirati:AllPurpose,Passarino:ApproachTowardNumerical,PassarinoUccirati:AlgebraicNumerical2loop,BardinKalinovskayaTkachov:1loopExperience}.
}
In our set-up, these are differential operators $P$ satisfying
\begin{equation*}
	P \G^{-\Dim/2} = 0.
\end{equation*}
We recall that such parametric annihilators provide all shift relations between Feynman integrals, in particular the ones known from the classical IBP method in momentum space. The obvious question, whether the latter suffice to obtain all shift relations, seems to remain open. As a positive indication in this direction, we show that the momentum space relations contain the inverse dimension shift.

Ideals of parametric annihilators are examples of $D$-modules.
Loeser and Sabbah studied the \emph{algebraic Mellin transform} \cite{LoeserSabbah:EDFdeterminants} of holonomic $D$-modules and proved a dimension formula in \cite{LoeserSabbah:IrredTore,LoeserSabbah:IrredToreII}, which, applied to our case, identifies the number of master integrals as an Euler characteristic.
The key property here is holonomicity, which was studied in the context of Feynman integrals already in \cite{KashiwaraKawai:HolonomicSystemsFI} and of course is crucial in the proof \cite{SmiPet} that there are only finitely many master integrals.

It is furthermore worthwhile to notice that algorithms \cite{Oaku:AlgorithmsB,OakuTakayama:AlgorithmsForDmodules} have been developed to compute generators for the ideal of all annihilators of $\G^{-\Dim/2}$, see also \cite[section~5.3]{SaitoSturmfelsTakayama}.
Today, efficient implementations of these algorithms via Gr\"{o}bner bases are available in specialized computer algebra systems such as {\Singular} \cite{Singular,Andetal}.
We hope that these improvements may stimulate further progress in the application of $D$-module theory to Feynman integrals \cite{Tarasov:LL1998,SmirnovSmirnov:ACAT2007,ApplyingGroebner,FujimotoKaneko:LL2012}.

We begin our article with a review of the momentum space and parametric representations of scalar Feynman integrals and recall how the Mellin transform translates shift relations to differential operators that annihilate the integrand.
In section~\ref{sec:ibp-momentum-space} we illustrate how the classical IBP identities obtained in momentum space supply special examples of such annihilators. The relations between integrals in different dimensions are addressed in section~\ref{sec:Dimension shifts}, where we relate them to the Bernstein-Sato operators and show that these can be obtained from momentum space IBPs.
Our main result is presented in section~\ref{sec:Euler-NoM}, where we apply the theory of Loeser and Sabbah to count the master integrals in terms of the Euler characteristic.
Practical applications of this formula are presented in the following section~\ref{sec:examples}, which includes a comparison to other approaches and results in the literature.
Finally we discuss some open questions and future directions.

In appendix~\ref{sec:app-representations} we give an example to illustrate our definitions in momentum space, present proofs of the parametric representations and demonstrate algebraically that momentum space IBPs are parametric annihilators.
The theory of Loeser and Sabbah is reviewed in appendix~\ref{sec:L-S}, which includes complete, simplified proofs of those theorems that we invoke in section~\ref{sec:Euler-NoM}.
Finally, appendix~\ref{sec:two-loop example} discusses the parametric annihilators of a two-loop example in detail.

\paragraph{Acknowledgments}
We like to thank Mikhail Kalmykov, Roman Lee, Robert Schabinger, Lorenzo Tancredi and Yang Zhang for interesting discussions and helpful comments on integration by parts for Feynman integrals, J{\o}rgen Rennemo for suggesting literature relevant for section~\ref{sec:no-masters} and Viktor Levandovskyy for communication and explanations around $D$-modules and in particular their implementation in {\Singular}.

Thomas Bitoun thanks Claude Sabbah for feedback and bringing the work \cite{LoeserSabbah:IrredTore,LoeserSabbah:IrredToreII,LoeserSabbah:EDFdeterminants} to our attention. Thomas acknowledges funding through EPSRC grant EP/L005190/1.

Christian Bogner thanks Deutsche Forschungsgemeinschaft for support under the project BO~4500/1-1.

This research was supported by the \href{http://www.munich-iapp.de/}{Munich Institute for Astro- and Particle Physics (MIAPP)} of the DFG cluster of excellence ``Origin and Structure of the Universe''. Furthermore we are grateful for support from the \href{http://hu.berlin/kmpb}{Kolleg Mathematik Physik Berlin (KMPB)} and hospitality at Humboldt-Universit\"{a}t Berlin.

Images in this paper were created with {\JaxoDraw} \cite{JaxoDraw} (based on {\Axodraw} \cite{Axodraw}) and {\FeynArts} \cite{Hahn:FeynArts3}.

\section{Annihilators and integral relations}
\label{sec:Ann and rel}%

In this section we elaborate a method to obtain relations between Feynman integrals from differential operators with respect to the Feynman parameters and show that these relations include the well known IBP relations from momentum space.

\subsection{Feynman integrals and Schwinger parameters}
\label{sec:representations}%

At first we fix conventions and notation for Feynman integrals in momentum space and recall their representations using Schwinger parameters. While the former is the setting for most traditional approaches to study IBP identities, it is the latter (in particular in its form with a single polynomial) which provides the direct link to the theory of $D$-modules that our subsequent discussion will be based on.

We consider integrals (also called \emph{integral families} \cite{ManteuffelStuderus:Reduze2}) that are defined by an $\Dim\Loops$-fold integral ($\Loops$ is the \emph{loop number}), over so-called \emph{loop momenta} $\LoopMom_1,\ldots,\LoopMom_{\Loops}$ in $\Dim$-dimensional Minkowski space, of a product of powers of \emph{denominators} $\Den=(\Den_1,\ldots,\Den_N)$:
\begin{equation}
	\FI(\nu_{1},\ldots,\nu_{N})
	= \left(\prod_{j=1}^{L}\int\frac{\dd[\Dim] \LoopMom_{j}}{\iu \pi^{\Dim/2}}\right)
	\prod_{a=1}^{N}\Den_{a}^{-\nu_{a}}
	.
	\label{eq:FI-momentum}%
\end{equation}
The denominators are (at most) quadratic forms in the $\Loops$ loop momenta and some number $\ExtMoms$ of linearly independent \emph{external momenta} $\ExtMom_1,\ldots,\ExtMom_\ExtMoms$.
In most applications, the denominators are inverse Feynman propagators associated with the momentum flow through a Feynman graph that arises from imposing momentum conservation at each vertex, see example~\ref{ex:bubble-parametric}. However, we will only restrict ourselves to graphs from section~\ref{sec:graph-polys} onwards, and keep our discussion completely general until then.

An integral \eqref{eq:FI-momentum} is a function of the \emph{indices} $\nu=(\nu_1,\ldots,\nu_N)$ (denominator exponents), the dimension $\Dim$ of spacetime and kinematical invariants (masses and scalar products of external momenta). However, we suppress the dependence on kinematics in the notation and treat kinematical invariants as complex numbers throughout.
The dimension and indices are understood as free variables; that is, we consider Feynman integrals as meromorphic functions of $(\Dim,\nu)\in\C^{1+N}$ in the sense of Speer \cite{Speer:GeneralizedAmplitudes}.

\begin{defn}
	\label{def:M-Q-J}%
	To each denominator $\Den_a$ of a list $\Den=(\Den_1,\ldots,\Den_N)$, we associate a variable $x_a$ ($1\leq a \leq N$) called \emph{Schwinger parameter}.
	The linear combination
\begin{equation}
	\sum_{a=1}^{N}x_{a} \Den_{a}
	=-\sum_{i,j=1}^{\Loops} \loopMat_{ij}(\LoopMom_{i} \cdot \LoopMom_{j})
	+\sum_{i=1}^{\Loops} 2 (\loopLin_{i} \cdot \LoopMom_{i})
	+\loopConst
	\label{eq:def:M-Q-J}%
\end{equation}
	decomposes into quadratic, linear and constant terms in the loop momenta.\footnote{%
		The scalar products $\LoopMom_{i} \cdot \LoopMom_{j} = \LoopMom_i^1 \LoopMom_j^1 - \sum_{\mu=2}^{\Dim} \LoopMom_{i}^{\mu} \LoopMom_{j}^{\mu}$ are understood with respect to the Minkowski metric.
	}
	This defines a symmetric $\Loops\times\Loops$ matrix $\loopMat$, a vector $Q$ of $\Loops$ linear combinations of external momenta and a scalar $\loopConst$.
	We define furthermore the polynomials
\begin{equation}
	\U \defas \det \loopMat,
	\quad
	\F
	\defas
	 \U \left( \loopLin^{\Transpose} \loopMat^{-1} \loopLin + \loopConst \right)
	    \quad\text{and}\quad
	  \G \defas \U+\F  
	.
	\label{eq:def-U-F}%
\end{equation}
\end{defn}
Schwinger parameters yield useful representations of the Feynman integrals \eqref{eq:FI-momentum}:
\begin{prop} \label{prop: parametric reps}
	Let us denote the \emph{superficial degree of convergence} by
	\begin{equation}
		\sdc \defas \abs{\nu} - \Loops\frac{\Dim}{2}
		\quad\text{where}\quad
		\abs{\nu} \defas \sum_{i=1}^{N} \nu_i
		.
		\label{eq:sdc}%
	\end{equation}
	Then the Feynman integral \eqref{eq:FI-momentum} can be written as
\begin{align}
	\FI(\nu_{1},\ldots,\nu_{N})
	&=
	\left(\prod_{i=1}^{N}\int_{0}^{\infty}\frac{x_{i}^{\nu_{i}-1}\dd x_{i}}{\Gamma\left(\nu_{i}\right)}\right)
	\frac{e^{-\F/\U}}{\U^{\Dim/2}}
	,
	\label{eq:parametric-exp}%
	\\
	\FI(\nu_{1},\ldots,\nu_{N})
	&=\Gamma(\sdc)
	\left(
		\prod_{i=1}^{N}
		\int_{0}^{\infty}
		\frac{x_{i}^{\nu_{i}-1}\dd x_{i}}{\Gamma\left(\nu_{i}\right)}
	\right)
	\frac{\delta\left(1-\sum_{j=1}^{N}x_{j}\right)}{\U^{\Dim/2-\sdc}\F^{\sdc}}
	\quad\text{and}
	\label{eq:FI-UF}%
	\\
	\FI(\nu_{1},\ldots,\nu_{N})
	&= \frac{\Gamma\left(\frac{\Dim}{2}\right)}{\Gamma\left(\frac{\Dim}{2}-\sdc\right)}
	\left(\prod_{i=1}^{N}
		\int_{0}^{\infty}
		\frac{x_{i}^{\nu_{i}-1}\dd x_{i}}{\Gamma(\nu_{i})}
	\right)
	\G^{-\Dim/2}
	.
	\label{eq:Lee-Pom}%
\end{align}
\end{prop}
The formulas \eqref{eq:parametric-exp} and \eqref{eq:FI-UF} are known since the sixties and we refer to \cite{Nak,Smirnov:AnalyticToolsForFeynmanIntegrals} for detailed discussions and for the original references. The trivial consequence~\eqref{eq:Lee-Pom} was popularized only much more recently by Lee and Pomeransky \cite{LeePom} and it is this representation that we will use in the following.
In appendix~\ref{sec:Momentum-space} we include proofs for these equations and provide further technical details.
\begin{figure}
	\centering%
	\includegraphics[width=0.35\textwidth]{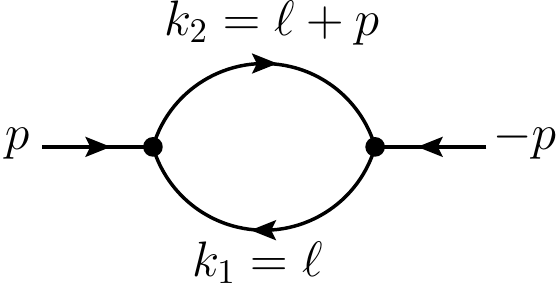}%
	\caption{The one-loop bubble graph with momentum flow.}%
	\label{fig:bubble}%
\end{figure}
\begin{example}
	\label{ex:bubble-parametric}%
	Consider the graph in figure~\ref{fig:bubble} with massless Feynman propagators, $\Den_1=-\LoopMom^2$ and $\Den_2=-(\LoopMom+\ExtMom)^2$.
	We find $\loopMat = x_1 + x_2$, $\loopLin = -x_2 p$ and $\loopConst = -x_2 p^2$ according to \eqref{eq:def:M-Q-J}.
	Hence the graph polynomials \eqref{eq:def-U-F} become
	$\U = x_1+x_2$ and $\F = (-p^2)x_1 x_2$ such that
	\begin{equation}
		\FI(\nu_1,\nu_2)
		= \frac{\Gamma(\frac{\Dim}{2})}{\Gamma(\Dim-\nu_1-\nu_2)}
		\int_0^{\infty} \frac{x_1^{\nu_1-1}\dd x_1}{\Gamma(\nu_1)}
		\int_0^{\infty} \frac{x_2^{\nu_2-1}\dd x_2}{\Gamma(\nu_2)}
		\left( x_1 + x_2 -p^2 x_1 x_2 \right)^{-\Dim/2}
		.
		\label{eq:bubble-parametric}%
	\end{equation}
\end{example}
\begin{rem}[meromorphicity]\label{rem:continuation}
	Depending on the values of $\Dim$ and $\nu$, the integrals \eqref{eq:FI-momentum} and \eqref{eq:parametric-exp}--\eqref{eq:Lee-Pom} can be divergent. However, there exists a non-empty, open domain in $\C^{1+N} \ni (\Dim,\nu)$ where all of them are convergent and agree with each other.\footnote{%
	The only exception are cases of zero-scale subintegrals, like massless tadpoles. Such integrals are zero in dimensional regularization and therefore irrelevant for our considerations.
}
	From there, analytic continuation defines a unique, meromorphic extension of every Feynman integral to the whole parameter space $\C^{1+N}$. The poles are simple and located on affine hyperplanes defined by linear equations with integer coefficients. For these foundations of analytic regularization\footnote{%
	The widely used \emph{dimensional regularization} is the special case when $\nu\in\Z^N$ are all integers and only the dimension $\Dim$ remains as a regulator. In this case, the poles are not necessarily simple anymore.
} we refer to \cite{Speer:SingularityStructureGenericFeynmanAmplitudes,Speer:GeneralizedAmplitudes}.

	The uniqueness of the analytic continuation of a Feynman integral (as a function of $\Dim$ and $\nu$) is very important; in particular, it means that an identity between Feynman integrals is already proven once it has been established locally in the non-empty domain of convergence of the involved integral representations. 
	In other words, in any calculation with analytically regularized Feynman integrals, we may simply assume, without loss of generality, that the parameters are such that the integrals converge. The resulting relation then necessarily remains true everywhere by analytic continuation.
\end{rem}
\begin{example}
	The one-loop propagator in example~\ref{ex:bubble-parametric} can be computed in terms of $\Gamma$-functions:
	\begin{equation}
		\FI(\nu_1,\nu_2)
		= (-p^2)^{\Dim/2-\nu_1-\nu_2}
		\frac{
			\Gamma(\Dim/2-\nu_1)
			\Gamma(\Dim/2-\nu_2)
			\Gamma(\nu_1+\nu_2-\Dim/2)
		}{
			\Gamma(\nu_1)
			\Gamma(\nu_2)
			\Gamma(\Dim-\nu_1-\nu_2)
		}
		.
		\label{eq:bubble-gamma}%
	\end{equation}
	The integral \eqref{eq:bubble-parametric} converges only in a certain domain of $(\nu,\Dim)$, but there it evaluates to \eqref{eq:bubble-gamma}, which has a unique meromorphic continuation. Its poles lie on the infinite family of hyperplanes defined by $\set{\Dim/2-\nu_1=k}$, $\set{\Dim/2-\nu_2=k}$ and $\set{\nu_1+\nu_2-\Dim/2=k}$, indexed by $k \in \Z_{\leq 0}$.
\end{example}

\subsection{Integral relations and the Mellin transform}
\label{sec:relations-Mellin}

In this section we summarize how relations between $\nu$-shifted Feynman integrals can be identified with differential operators that annihilate $\G^{-\Dim/2}$. This method was suggested in \cite{Lee:ModernTechniques} (and in \cite{Baikov:ExplicitSolutions-multiloop} for the Baikov representation).

The parametric representation \eqref{eq:Lee-Pom} can be interpreted as a multi-dimensional Mellin transform. For our purposes, we slightly deviate from the standard definition, as for example given in \cite{Brychkov}, and include the factors $\Gamma(\nu_i)$ that occur in \eqref{eq:Lee-Pom}.
\begin{defn} 
    \label{def:Mellin transform}%
	Let $\nu=(\nu_{1},\ldots,\nu_{N})\in\C^{N}$. The twisted (multi-dimensional) Mellin transform of a function $f\colon \R_{+}^{N}\longrightarrow\C$ is defined as 
\begin{equation}
	\Mellin{f}(\nu)
	\defas 
	\left(\prod_{i=1}^{N}\int_{0}^{\infty} \frac{x_i^{\nu_{i}-1}\dd x_{i}}{\Gamma(\nu_i)}\right)f(x_{1},\ldots,x_{N}),
	\label{eq:Mellin}%
\end{equation}
	whenever this integral exists. As a special case we define
	\begin{equation}
		\FImel(\nu)
		\defas \Mellin{\G^{-\Dim/2}}(\nu)
		\quad\text{such that}\quad
		\FI (\nu)
		= \frac{\Gamma(\Dim/2)}{\Gamma(\Dim/2-\sdc)}
		  \FImel(\nu)
		.
		\label{eq:FI-Mellin}%
	\end{equation}
\end{defn}
Recall that, as mentioned in remark~\ref{rem:continuation}, we do not have to worry about the actual domain of convergence of \eqref{eq:Mellin} in the algebraic derivations below. The key features of the Mellin transform for us are the following elementary relations; see \cite{Brychkov} for their form without the $\Gamma$'s in \eqref{eq:Mellin}.
\begin{lem}
	\label{lem:Mellin-properties}
	Let $\alpha,\beta\in\C$, $\nu\in\C^{N}$, $1\leq i \leq N$ and $f,g\colon \R_+^{N} \longrightarrow \C$. Writing $\uv{i}$ for the $i$-{th} unit vector, the (twisted) Mellin transform has the following properties:
	\begin{enumerate}
		\item Linearity: 
		$
			\Mellin{\alpha f+\beta g}(\nu)
			=\alpha \Mellin{f} (\nu) + \beta \Mellin{g}(\nu)
		$,

		\item Multiplication:
		$
			\Mellin{ x_i f }(\nu)
			=\nu_i \Mellin{f} (\nu+\uv{i})
		$ and

		\item Differentiation:
		$
			\Mellin{ \partial_i f}(\nu)
			=-\Mellin{ f} (\nu-\uv{i})
		$.
	\end{enumerate}
\end{lem}
\begin{proof}
	The linearity is immediate from \eqref{eq:Mellin} and the functional equation $\Gamma(\nu_i+1) = \nu_i\Gamma(\nu_i)$ provides the multiplication rule $ \left[x_i^{\nu_i-1}/\Gamma(\nu_i) \right] x_i = \nu_i x_i^{\nu_i}/\Gamma(\nu_i+1)$. The differentiation rule is a consequence of integration by parts,
	\begin{equation*}
		\int_0^{\infty} \frac{x_i^{\nu_i-1} \dd x_i}{\Gamma(\nu_i)}
		\partial_i f
		= \left[ \frac{x_i^{\nu_i-1}}{\Gamma(\nu_i)} f \right]_{x_i=0}^{\infty}
		- \int_0^{\infty} \frac{x_i^{\nu_i-2} \dd x_i}{\Gamma(\nu_i-1)} f,
	\end{equation*}
	because the boundary terms vanish inside the convergence domain of both integrals. This just says that if $\lim_{x_i\rightarrow 0} (x_i^{\nu_i-\epsilon-1} f)$ is finite for some $\epsilon>0$, then $\lim_{x_i\rightarrow 0} (x_i^{\nu_i-1} f) = 0$ vanishes and the analogous argument applies to the upper bound $x_i\rightarrow \infty$.
\end{proof}

\subsection{Operator algebras and annihilators}
\label{sec:annihilators}

The Mellin transform relates differential operators acting on $\G^{-\Dim/2}$ with operators that shift the indices $\nu$ of the Feynman integrals $\FI(\nu)$. In this section we formalize this connection algebraically in the language of $D$-modules.
For the most part, we only need basic notions which we will introduce below, and point out \cite{Coutinho:Primer} as a particularly accessible introduction to the subject.
\begin{defn}
	\label{def:Weyl-algebra}%
	The Weyl algebra $\Weyl{N}$ in $N$ variables $x_{1},\ldots,x_{N}$ is the non-commutative algebra of polynomial differential operators
\begin{equation}
	\Weyl{N} 
	\defas \C\left\langle
		x_{1},\ldots,x_{N},\partial_{1},\ldots,\partial_{N}
		\ \ | \ \ 
		\partial_{i}x_{j}=x_{j}\partial_{i}+\delta_{ij}\ \text{for all}\ 1\leq i,j\leq N
	\right\rangle,
	\label{eq:Weyl-algebra}%
\end{equation}
such that the commutators are $[x_i,x_j]=[\partial_i,\partial_j]=0$ and $[\partial_i, x_j] = \delta_{i,j}$ (Kronecker delta).
\end{defn}
Note that with the multi-index notations
$
	x^{\alpha} = x_{1}^{\alpha_{1}}\!\cdots x_{N}^{\alpha_{N}}
$
and
$
	\partial^{\beta} = \partial_{1}^{\beta_{1}}\!\cdots\partial_{N}^{\beta_{N}}
$,
every operator $P\in \Weyl{N}$ can be written uniquely in the form 
\begin{equation*}
	P=\sum_{\alpha,\beta}
		c_{\alpha\beta} x^{\alpha} \partial^{\beta}
	\quad\text{with}\quad
	\alpha,\beta\in\N_0^N,
\end{equation*}
by commuting all derivatives to the right (only finitely many of the coefficients $c_{\alpha\beta}\in\C$ are non-zero).
Extending the coefficients $c_{\alpha\beta}$ from $\C$ to polynomials $\C[s]$ in a further, commuting variable $s$, we obtain the algebra
$\Weyl{N}[s] \defas \Weyl{N} \tp_{\C} \C[s]$.
Later, we will also consider the case $\Weyl{N}_k \defas \Weyl{N} \tp_{\C} k$ of coefficients that are rational functions $k\defas \C(s)$.
\begin{defn}\label{def:f^s}
	Given a polynomial $f \in \C[x]$, the $\Weyl{N}[s]$-module $\C[s,x,1/f]\cdot f^s$ consists of elements of the form $(p/f^k) \cdot f^s$ (where $p\in \C[s,x]$, $k\in \N_0$) with the $\Weyl{N}[s]$-action 
	\begin{equation}
		q \left( \frac{p}{f^k} \cdot f^s \right)
		\defas \frac{q p}{f^k} \cdot f^s
		,\quad
		\partial_i \left( \frac{p}{f^k} \cdot f^s \right)
		\defas \frac{f(\partial_i p) + (s-k) p (\partial_i f)}{f^{k+1}} \cdot f^s
		\label{eq:f^s-def}%
	\end{equation}
	for any polynomial $q\in\C[s,x]$. This is just the natural action by multiplication and differentiation, $\partial_i \mapsto \partial/(\partial x_i)$. We abbreviate $f^{s+k} \defas f^k \cdot f^s$ for $k \in \Z$.
	The cyclic submodule generated by $f^s$ is denoted as $\Weyl{N}[s] f^s$.
\end{defn}
\begin{defn}
	\label{def:annihilator}%
	The $s$-parametric annihilators of a polynomial $f$ in $x_1,\ldots,x_N$ are the elements of the (left) ideal of operators in $\Weyl{N}[s]$ whose action on $f^{s}$ is zero:
	\begin{equation*}
		\Ann_{\Weyl{N}[s]}(f^{s})
		\defas \setexp{P\in \Weyl{N}[s]}{P f^{s}=0}.
		\label{eq:annihilator}%
	\end{equation*}
\end{defn}
\begin{example}
	Given $f\in\C[x_1,\ldots,x_N]$, we always have the trivial annihilators
	\begin{equation}
		f \partial_i - s (\partial_i f) 
		\in \Ann\left( f^s \right)
		\quad\text{for}\quad
		1 \leq i \leq N.
		\label{eq:triv-Ann}%
	\end{equation}
\end{example}
Note that an annihilator ideal is a module over $\Weyl{N}[s]$, that is, whenever $P f^s=0$, also $(QP)f^s = Q(P f^s) = 0$ for any operator $Q\in\Weyl{N}[s]$.
These ideals are studied in $D$-module theory~\cite{SaitoSturmfelsTakayama,Coutinho:Primer} and in principle annihilators can be computed algorithmically with computer algebra systems such as {\Singular}~\cite{Singular,Andetal,Andres:Master}. 
For the study of the Feynman integrals~\eqref{eq:Lee-Pom}, we set $s=-\Dim/2$ and $f=\G=\U+\F$ is the polynomial from \eqref{eq:def-U-F}. Via the Mellin transform, the elements of $\Weyl{N}[s] \G^s$ are the integrands of shifts of the Feynman integral $\Mellin{\G^s}$.
Due to lemma~\ref{lem:Mellin-properties}, every annihilator
$
	P \in \Ann_{\Weyl{N}[s]}\left(\G^s\right)
$
corresponds to an identity of Feynman integrals with shifted indices $\nu$.
\begin{defn}
	\label{def:Shifts}%
	The algebra $\Shifts{N}$ of shift operators in $N$ variables is defined by
\begin{equation}
	\Shifts{N}
	\defas
	\C\left\langle \PlusD{1},\ldots,\PlusD{N},\Minus{1},\ldots,\Minus{N}
	\ \ |\ \ 
	[\PlusD{i},\PlusD{j}]=[\Minus{i},\Minus{j}]=0
	\ \text{and}\ 
	[\PlusD{i},\Minus{j}] = \delta_{i,j}
	\right\rangle
	.
	\label{eq:Shifts}%
\end{equation}
\end{defn}
This algebra is clearly isomorphic to the Weyl algebra $\Weyl{N}$, since under the identifications $\PlusD{i} \leftrightarrow \partial_i$ and $\Minus{j} \leftrightarrow x_j$ the commutation relations are identical. In fact, a different isomorphism is given by $\PlusD{i} \leftrightarrow x_i$ and $\Minus{j} \leftrightarrow -\partial_j$, and it is this identification that corresponds to the Mellin transform (see lemma~\ref{lem:Mellin-properties}).
We therefore denote it by\footnote{
	The use of the symbol $\Mellin{\cdot}$ for both the analytic Mellin transform \eqref{eq:Mellin} and the isomorphism \eqref{eq:Mellin-iso} should not lead to any confusion. It is suggestive of, and justified by, corollary~\ref{cor:Mellin-compatibility}.
}
\begin{equation}
	\Mellin{\cdot} \colon \Weyl{N} \stackrel{\cong}{\longrightarrow} \Shifts{N}
	,\quad
	P \mapsto \Mellin{P} \defas
	\restrict{P}{
		x_i \mapsto \PlusD{i}, \partial_i \mapsto - \Minus{i}
		\ \text{for all $1\leq i \leq N$}
	}
	.
	\label{eq:Mellin-iso}%
\end{equation}

The conceptual difference between $\Shifts{N}$ and $\Weyl{N}$ is that we think of $\Weyl{N}$ as acting on functions $f(x)$ by differentiation, whereas $\Shifts{N}$ acts on functions $F(\nu)$ of a different set $\nu=(\nu_1,\ldots,\nu_N)$ of variables (the indices of Feynman integrals) by shifts of the argument and, in case of $\PlusD{i}$, a multiplication with $\nu_i$:
\begin{equation}
	(\Minus{i} F)(\nu)  \defas F(\nu - \uv{i})
	\quad\text{and}\quad
	(\PlusD{i}  F)(\nu) \defas \nu_i F(\nu + \uv{i})
	.
	\label{eq:def-shift-action}%
\end{equation}
These operators are used very frequently in the literature on IBP relations, as for example in \cite{ApplyingGroebner,Smirnov:AnalyticToolsForFeynmanIntegrals,Lee:GroupStructureIBP,Grozin:IBP}.
An important role is played by the operators
\begin{equation}
	\nOp{i} \defas \PlusD{i} \Minus{i},
	\quad\text{which act by simple multiplication:}\quad
	(\nOp{i} F)(\nu) = \nu_i F(\nu)
	.
	\label{eq:def-n}%
\end{equation}
Their commutation relations are
\begin{equation}
	[\PlusD{i}, \nOp{j}]
	=  \PlusD{i} \delta_{i,j}
	\quad\text{and}\quad
	[\Minus{i}, \nOp{j}]
	= - \Minus{i} \delta_{i,j}
	.
	\label{eq:n-commutators}%
\end{equation}
In terms of the $\Shifts{N}$ action \eqref{eq:def-shift-action}, we can rephrase the essence of lemma~\ref{lem:Mellin-properties} as
\begin{cor}
	\label{cor:Mellin-compatibility}%
	The (twisted) Mellin transform \eqref{eq:Mellin} is compatible with the actions of the algebras $\Weyl{N}$ and $\Shifts{N}$ under their identification \eqref{eq:Mellin-iso}. In other words,
	\begin{equation}
		\Mellin{P f^s}
		= \Mellin{P} \left[ \Mellin{f^s} \right]
		\quad\text{for all operators}\quad
		P \in \Weyl{N}[s]
		.
		\label{eq:Mellin-compatibility}%
	\end{equation}
\end{cor}
\begin{cor}
	\label{cor:FI-annihilators}%
	Every annihilator $P \in \Ann_{\Weyl{N}[\Dim]}\left( \G^{s} \right)$ of $\G^{s}=\G^{-\Dim/2}$ yields a shift relation $\Mellin{P} \in \Shifts{N}[s] \defas \Shifts{N} \tp \C[s]$ of the Feynman integral $\FImel$ from \eqref{eq:FI-Mellin}: $\Mellin{P} \FImel = 0$.
\end{cor}
\begin{example}
	\label{ex:bubble-relations}%
	For the bubble graph in figure~\ref{fig:bubble} with $\G=x_1+x_2-p^2 x_1 x_2$ from example~\ref{ex:bubble-parametric},
\begin{equation*}
	(-p^2) x_1(s-x_1\partial_1) + (s-x_1\partial_1-x_2\partial_2)
	   \in \Ann(\G^s)
\end{equation*}
	is easily checked to annihilate $\G^s$. We therefore get the shift relation
\begin{equation*}
	(s+\nOp{1} +\nOp{2})\FImel
	= p^2 \PlusD{1} (s+\nOp{1})\FImel
	= p^2 (s+\nOp{1}+1)\PlusD{1} \FImel
	.
\end{equation*}
	According to \eqref{eq:def-n}, this relation can also be written as
	\begin{equation}
		(-p^2) \nu_1 \FImel(\nu_1+1,\nu_2)
		= - \frac{s+\nu_1+\nu_2}{s+\nu_1+1} \FImel(\nu_1,\nu_2)
		.
		\label{eq:bubble-relation-mel}%
	\end{equation}
\end{example}
We prefer to work with the modified Feynman integral $\FImel$ from \eqref{eq:FI-Mellin}, because it is directly related to the Mellin transform. However, it is straightforward to translate relations between $\FImel$ into relations for the actual Feynman integral $\FI$.
Namely, if 
$
	P
	=\sum_{\alpha,\beta} c_{\alpha,\beta} x^{\alpha} \partial^{\beta} 
	\in \Ann_{\Weyl{N}[s]}\left( \G^s \right)
$, we substitute \eqref{eq:FI-Mellin} to see
\begin{equation*}
	0 = \Mellin{P} \FImel
	=\sum_{\alpha,\beta} c_{\alpha,\beta}
	\left(\prod_{i,j=1}^N ( \PlusD{i})^{\alpha_i} (-\Minus{j})^{\beta_j}\right) \frac{\Gamma(-s-\sdc)}{\Gamma(-s)} \FI
\end{equation*}
and then recall from \eqref{eq:sdc} that $ \sdc  = \sum_i \nu_i + \Loops s$ to conclude
\begin{equation}
	0 = \frac{\Gamma(-s)\Mellin{P} \FImel}{\Gamma(-s-\sdc)} 
	=\sum_{\alpha,\beta} c_{\alpha,\beta}
	\frac{\Gamma(-s-\sdc-\abs{\alpha}+\abs{\beta})}{\Gamma(-s-\sdc)}
	\left(\prod_{i,j=1}^N ( \PlusD{i})^{\alpha_i} (-\Minus{j})^{\beta_j}\right)  \FI
	.
	\label{eq:FImel-relation-to-FI}%
\end{equation}
Notice that the fraction $\Gamma(-s-\sdc-\abs{\alpha}+\abs{\beta})/\Gamma(-s-\sdc)$ is a rational function in $\sdc$ (and thus in $\nu$ and $s$), due to the functional equation for $\Gamma$, since $\abs{\alpha}$ and $\abs{\beta}$ are integers.
\begin{example}
	Substituting $\FImel(\nu_1,\nu_2)=\FI(\nu_1,\nu_2) \Gamma(-2s-\nu_1-\nu_2)/\Gamma(-s)$ and $\FImel(\nu_1+1,\nu_2) =\FI(\nu_1+1,\nu_2) \Gamma(-2s-\nu_1-1-\nu_2)/\Gamma(-s)$ from \eqref{eq:FI-Mellin} into \eqref{eq:bubble-relation-mel} results in
	\begin{equation*}
		(-p^2) \nu_1 \FI(\nu_1+1,\nu_2) = \frac{(s+\nu_1+\nu_2)(2s+\nu_1+\nu_2+1)}{s+\nu_1+1} \FI(\nu_1,\nu_2)
		,
	\end{equation*}
	which also follows from the expression \eqref{eq:bubble-gamma} of $\FI(\nu)$ in terms of $\Gamma$-functions.
\end{example}
Let us recapitulate these observations: By the Mellin transform \eqref{eq:Mellin-iso}, every annihilator 
$P \in \Ann_{\Weyl{N}[s]}(\G^s)$
gives rise to a linear relation $\Mellin{P} \FImel=0$ of the rescaled Feynman integral $\FImel$ via a shift operator $\Mellin{P}\in \Shifts{N}[s]$.
Also we noted that according to $\FImel = \FI\, \Gamma(-s-\sdc)/\Gamma(-s) $ from \eqref{eq:FI-Mellin}, such a relation is equivalent to a relation of the original Feynman integral $\FI$ as in \eqref{eq:FImel-relation-to-FI}.

On the other hand, if we are given a shift relation $R\FImel = 0$ where $R\in \Shifts{N}[s]$, then $R=\Mellin{P}$ corresponds to the differential operator $P=\Mellin[-1]{R}\in \Weyl{N}[s]$ under the Mellin transform \eqref{eq:Mellin-iso}. From the vanishing $\Mellin{P \G^s}=0$ we can conclude that this operator must be an annihilator,
$P \in \Ann_{\Weyl{N}[s]}(\G^s)$. This follows from
\begin{thm}[{Inverse Mellin transform, e.g.~\cite[Theorem~3.5]{Brychkov}}]
	\label{thm:Mellin-inverse}%
	Suppose that the (twisted) Mellin transform $f^{\star}(\nu) \defas \Mellin{f}(\nu)$ of $f(x)$ from \eqref{eq:Mellin} converges in a domain of the form $a_i\leq \Realteil(\nu_i) \leq b_i$ for all $1\leq i \leq N$, where $a,b \in \R^N$.
	Then its inverse is given by
	\begin{equation}
		f(x)
		= \Mellin[-1]{f^\star}(x)
		= \left(\prod_{i=1}^N 
			\int_{\sigma_i+\iu\R}
			\hspace{-4mm}\dd \nu_i\ 
			\frac{\Gamma(\nu_i)}{ x_i^{\nu_i} \cdot 2\pi\iu}
		\right)
		f^{\star}(\nu),
		\quad\text{where}\quad
		x \in \R_+^N.
		\label{eq:Mellin-inverse}%
	\end{equation}
	This multiple integral along lines parallel to the imaginary axis converges for $a_i\leq \sigma_i \leq b_i$ ($1\leq i \leq N$) and does not depend on the concrete choice of $\sigma_i$.
\end{thm}
So not only do we get relations for Feynman integrals from parametric annihilators of $\G^s$, but in fact \emph{every} relation of the form $P\FImel=0$ for a polynomial shift operator $P\in\Shifts{N}[s]$ does arise in this way.
\begin{cor}
	\label{cor:Ann-iso}%
	Let $\Ann_{\Shifts{N}[s]}(\FImel) \subseteq \Shifts{N}[s]$ denote the $\Shifts{N}[s]$-module of polynomial shift operators that annihilate a (rescaled) Feynman integral~\eqref{eq:FI-Mellin}.
	Then the Mellin transform~\eqref{eq:Mellin-iso} restricts to a bijection between these relations and parametric annihilators:
	\begin{equation}
		\Mellin{\cdot}\colon 
		\Ann_{\Weyl{N}[s]}\left( \G^s \right) 
		\stackrel{\isomorph}{\longrightarrow} 
		\Ann_{\Shifts{N}[s]}\left( \FImel \right)
		.
		\label{eq:Ann-iso}%
	\end{equation}
\end{cor}
Instead of focussing on the annihilators themselves, we can also look at the $\Weyl{N}[s]$-module $\Weyl{N}[s] \cdot \G^s \isomorph \Weyl{N}[s]/\Ann_{\Weyl{N}[s]}(\G^s)$ of the integrands and the $\Shifts{N}[s]$-module $\Shifts{N}[s] \cdot \FImel \isomorph \Shifts{N}[s]/\Ann_{\Shifts{N}[s]}(\FImel)$ of all (shifted) Feynman integrals. The Mellin transform gives an isomorphism
\begin{equation}
	\Mellin{\cdot}\colon
	\Weyl{N}[s]\cdot \G^s \stackrel{\isomorph}{\longrightarrow} \Shifts{N}[s] \cdot \FImel
	.
	\label{eq:integrands-FIs-iso}%
\end{equation}

\subsection{On the correspondence to momentum space}
\label{sec:ibp-momentum-space}%

In this section we first recall the integration by parts (IBP) relations for Feynman integrals that are derived in momentum space, following \cite{Grozin:IBP}. We then note that these provide a special set of parametric annihilators and discuss some open questions in regard of this comparison of IBPs in parametric and momentum space.

Since the denominators $\Den_a$ from \eqref{eq:FI-momentum} are quadratic forms in the $\AllMoms=\Loops+\ExtMoms$ momenta $q=(q_1,\ldots,q_{\AllMoms}) \defas (\LoopMom_1,\ldots,\LoopMom_{\Loops},\ExtMom_1,\ldots,\ExtMom_{\ExtMoms})$, we can write them in the form
\begin{equation}
	\Den_{a}
	=\sum_{ \set{i,j} \in \bilabels} \spMat_{a}^{\set{i,j}} s_{\set{i,j}}+\lambda_{a}
	\label{eq:bilinear-decomposition}%
\end{equation}
such that the coefficients $\spMat_a^{\set{i,j}}$ and $\lambda_a$ are independent of loop momenta and the pairs
\begin{equation}
	\bilabels
	\defas \setexp{ \set{i,j} }{1\leq i \leq \Loops\  \text{and}\ 1\leq j \leq \AllMoms}
	\label{eq:bilabels}%
\end{equation}
label the $\abs{\bilabels} = \frac{\Loops(\Loops+1)}{2}+\Loops\ExtMoms$ loop-momentum dependent scalar products 
\begin{equation}
	s_{\set{i,j}} \defas q_{i}q_{j} = q_{j}q_{i}.
	\label{eq:scalar-products}%
\end{equation}
In order to express the IBP relations coming from momentum space in terms of the integrals \eqref{eq:FI-momentum}, we need to assume, for this and the following section, that we consider $N=\abs{\bilabels}$ denominators such that the $N\times N$ square matrix $\spMat$ defined by \eqref{eq:bilinear-decomposition} is invertible.\footnote{%
	For integrals associated to Feynman graphs, the number of edges is often less than $\abs{\bilabels}$. In this case, one augments the list of inverse propogators by an appropriate choice of additional quadratic forms in the loop momenta, called \emph{irreducible scalar products} (ISPs), to achieve $N=\abs{\bilabels}$.
	See example~\ref{ex:doubletri-ISP}.
}
We think of $\spMat_{a}^{\set{i,j}}$ as the element of $\spMat$ in row $a$ and column $\set{i,j}$ and write $\spMat_{\set{i,j}}^{a}$ for the entry in row $\set{i,j}$ and column $a$ of $\spMat^{-1}$,
such that the inverse of \eqref{eq:bilinear-decomposition} can be written as
\begin{equation}
	s_{\set{i,j}} = \sum_{a=1}^{N}\spMat_{\set{i,j}}^{a}\left(\Den_{a}-\lambda_{a}\right)
	\quad\text{for all}\quad
	\set{i,j} \in \bilabels
	.
	\label{eq:sp-from-props}%
\end{equation}
We are interested in relations of the Feynman integral $\FI$ from \eqref{eq:FI-momentum}, that is,
\begin{equation}
	\FI(\nu_{1},\ldots,\nu_{N})
	= \left(\prod_{j=1}^{\Loops}\int\frac{\dd[\Dim] \LoopMom_{j}}{\iu \pi^{\Dim/2}}\right) f
	\quad\text{with the integrand}\quad
	f = \prod_{i=1}^{N} \Den_{i}^{-\nu_{i}}
	.
	\label{eq:momentum-integrand}%
\end{equation}
\begin{defn}
	\label{def:mom-ibp-generators}%
	The \emph{momentum space IBP relations} of $\FI(\nu_{1},\ldots,\nu_{N})$ are those relations between scalar Feynman integrals that are obtained from Stokes' theorem
	\begin{equation}
		\left(\prod_{n=1}^{\Loops}\int \dd[\Dim] \LoopMom_{n}\right) \momIBP{i}{j} f
		= 0,
		\label{eq:momentum-stokes}%
	\end{equation}
	where the operators $\momIBP{i}{j}$ are defined in terms of the momenta\footnote{%
		Recall that $q_{1},\ldots, q_{\Loops}$ denote the loop momenta
		whereas $q_{\Loops+1},\ldots, q_{\AllMoms}$ are the external momenta.%
	} as
	\begin{equation}
		\momIBP{i}{j} 
		\defas 
		\frac{\partial}{\partial q_{i}} \cdot q_{j}
		=
		\sum_{\mu=1}^{\Dim} \frac{\partial}{\partial q_{i}^{\mu}}q_{j}^{\mu}
		\quad\text{for}\quad
		i\in\set{1,\ldots,\Loops},
		j\in\set{1,\ldots,M}.
		\label{eq:momIBP}%
	\end{equation}
\end{defn}
The following, explicit form of these relations as difference equations is essentially due to Baikov~\cite{Baikov:ExplicitSolutions-nloop,Baikov:ExplicitSolutions-multiloop}; see also Grozin~\cite{Grozin:IBP}. For completeness, we include the proof in appendix~\ref{sec:Momentum-space}.
\begin{prop}\label{prop:IBP_Grozin}
	Given a set of $N=\abs{\bilabels}$ denominators $\Den$ such that the matrix $\spMat$ defined by \eqref{eq:bilinear-decomposition} is invertible,
	every momentum-space IBP relation can be written explicitly as
	\begin{equation}
		\shiftIBP{i}{j} \FI \left(\nu_{1},\ldots ,\nu_{N}\right) = 0
		\label{eq:shift-IBP-relation}%
	\end{equation}
	where $\shiftIBP{i}{j}$ denotes shift operators, indexed by $1\leq i \leq \Loops$ and $1\leq j \leq \AllMoms$, that are given by
	\begin{equation}
		\shiftIBP{i}{j}
		\defas \begin{cases}
		\displaystyle
			d\delta_{ij}-\sum_{a,b=1}^{N} C_{aj}^{bi}\PlusD{a}\left( \Minus{b}-\lambda_{b}\right)
			&\text{for $j\leq \Loops$ and}
		\\\displaystyle
			-\sum_{a,b=1}^{N} C_{aj}^{bi} \PlusD{a}\left(\Minus{b}-\lambda_{b}\right)-\sum_{a=1}^{N}\sum_{m=\Loops+1}^{\AllMoms}A_{a}^{\set{i,m}} q_{j}q_{m}\PlusD{a}
			&\text{for $j>\Loops$.}
		\end{cases}
		\label{eq:shift-IBP}%
	\end{equation}
	The coefficients $C^{bi}_{aj}$ are defined as
	\begin{equation}
		C_{aj}^{bi} 
		\defas
		\begin{cases}
			\sum_{m=1}^{\AllMoms} \spMat_{a}^{\{i,m\}} \spMat_{\set{m,j}}^{b}\left(1+\delta_{mi}\right) & \text{if $j\leq\Loops$ and}
		\\
			\sum_{m=1}^{\Loops} \spMat_{a}^{\{i,m\}} \spMat_{\set{m,j}}^{b}\left(1+\delta_{mi}\right) & \text{if $j>\Loops$.}
		\end{cases}
		\label{eq:IBP-coeffs}%
	\end{equation}
\end{prop}

\begin{cor}
	\label{cor:parIBP}%
	To every difference equation $\shiftIBP{i}{j} \FI=0$ from momentum space IBP, there corresponds a parametric annihilator $\parIBP{i}{j} \in \Ann_{\Weyl{N}[s]}\left( \G^s \right)$ of the form
	\begin{align}
		\parIBP{i}{j} 
		& = d\delta_{ij}+\sum_{a,b=1}^{N} C_{aj}^{bi}x_{a}\left(\partial_{b}+\lambda_{b}H\right)
		\quad\text{for}\quad i, j\leq \Loops \quad\text{and}
		\label{eq:parIBP<=L}\\
\parIBP{i}{j}
		& = \sum_{a,b=1}^{N} C_{aj}^{bi}x_{a}\left(\partial_{b}+\lambda_{b}H\right)-\sum_{a=1}^{N}\sum_{m=\Loops+1}^{\AllMoms} \spMat_{a}^{\set{i,m}}q_{j}q_{m}x_{a}H
			\quad\text{for}\quad
			i\leq \Loops<j,
		\label{eq:parIBP>L}%
	\end{align}
	where
	$
		H \defas \frac{(\Loops+1)\Dim}{2} + \sum_{c=1}^{N}x_{c}\partial_{c}
	$.
\end{cor}
\begin{proof}
	First recall the rescaling \eqref{eq:FI-Mellin} between the Feynman integral $\FI$ and the Mellin transform $\FImel$ of $\G^s$. As we saw in \eqref{eq:FImel-relation-to-FI}, this means that
	\begin{equation*}
		\frac{\Gamma(-s-\sdc)}{\Gamma(-s)}
		\PlusD{a} \FI
		= \Gamma(-s-\sdc) \PlusD{a} \frac{\FImel}{\Gamma(-s-\sdc)}
		= \left( -s-\sdc-1 \right) \PlusD{a} \FImel
		= \PlusD{a} (-s-\sdc) \FImel,
	\end{equation*}
	and so if we substitute \eqref{eq:FI-Mellin} into $\shiftIBP{i}{j} \FI=0$ for the operators from \eqref{eq:shift-IBP}, then apart from the substitution $\PlusD{a}\Minus{b} \mapsto x_a(-\partial_b)$ which does not change $\sdc$, the remaining terms with shifts $\PlusD{a}$ do increment $\sdc$ by one and thus acquire an additional factor of
	\begin{equation*}
		-s-\sdc
		=(\Loops+1)(-s) - \sum_i \nOp{i}
		\mapsto
		(\Loops+1)(-s) + \sum_{c=1}^N x_c \partial_c
		= H.
	\end{equation*}
	This proves that $\Mellin{\parIBP{i}{j}}\FImel = 0$ for the operators in \eqref{eq:parIBP<=L} and \eqref{eq:parIBP>L} and theorem~\ref{thm:Mellin-inverse} concludes the proof.
\end{proof}

Note that the proof of the identity $\parIBP{i}{j} \G^s = 0$ given in corollary~\ref{cor:parIBP} rests on the inverse Mellin transform. An alternative, direct algebraic proof is given in appendix \ref{sec:Algebraic proof}.\footnote{%
	For an algebraic proof of the analogous statement in the Baikov representation, see \cite[section~9]{Grozin:IBP}.
}
\begin{defn}
	\label{def:Mom}%
	By $\Mom$ we denote the left $\Weyl{N}[s]$-module generated by the annihilators $\parIBP{i}{j}$ from corollary~\ref{cor:parIBP}, corresponding to the momentum space IBP identities:
	\begin{equation}
		\Mom \defas \sum_{i,j} \Weyl{N}[s] \cdot \parIBP{i}{j}
		\quad\subseteq\quad \Ann_{\Weyl{N}[s]}\left( \G^s \right)
		.
		\label{eq:def-Mom}%
	\end{equation}
\end{defn}
Since the $\parIBP{i}{j}$ are first order differential operators, we have the inclusions
\begin{equation}
	\Mom 
	\subseteq \Ann^1_{\Weyl{N}[s]}(\G^s)
	\subseteq \Ann_{\Weyl{N}[s]}(\G^s),
	\label{eq:Mom-Ann1-Ann}%
\end{equation}
where $\Ann^1$ denotes the $\Weyl{N}[s]$-module generated by all first order annihilators. Note that for a generic polynomial $\G$, one would not expect that all of its annihilators can be obtained from linear ones ($\Ann^1\subsetneq \Ann$). It is therefore interesting that we observe the equality $\Ann^1=\Ann$ in all cases of Feynman integrals that we checked.
\begin{quest}\label{con:Ann=Ann1}
	For a Feynman integral with a complete set of irreducible scalar products, is the second inclusion in \eqref{eq:Mom-Ann1-Ann} an equality? In other words, are the $s$-parametric annihilators of Lee-Pomeransky polynomials $\G$ linearly generated?
\end{quest}
Regarding the first inclusion in \eqref{eq:Mom-Ann1-Ann}, we do know that it is strict (see the example in appendix~\ref{sec:Comparing}). 
However, it seems that $\Mom$ does provide all identities once we enlarge the coefficients to rational functions in the dimension $s=-\Dim/2$ and the indices $\nu_e$.
Let us write $\theta=(\theta_1,\ldots,\theta_N)$ and $\theta_e \defas x_e \partial_e$ such that $\nu_e=\Mellin{-\theta_e}$. 
\begin{quest}\label{con:Ann=Mom}
	Given any annihilator $P\in\Ann_{\Weyl{N}[s]}(\G^s)$, does there exist a polynomial $q \in \C[s,\theta]$ such that $q P \in \Mom$? In other words, does
	\begin{equation}
		\C(s,\theta) \tp_{\C[s,\theta]} \Ann_{\Weyl{N}[s]}(\G^s)
		=
		\C(s,\theta) \tp_{\C[s,\theta]} \Mom 
		\quad\text{hold?}
		\label{eq:Ann=Mom}%
	\end{equation}
\end{quest}
To test this conjecture, we should take all known shift relations for Feynman integrals, and check if they can be realized as elements of $\Mom$ (after localizing at $\C(s,\theta)$). In the remainder of this section, we will address such a relation, namely the one originating from the well-known dimension shifts.

\subsection{Dimension shifts}
\label{sec:Dimension shifts}

The representation $\FI= \Mellin{e^{-\F/\U}\cdot \U^s}$ from \eqref{eq:parametric-exp} shows, through corollary~\ref{cor:Mellin-compatibility}, that
\begin{equation}
	\FI(\Dim)
	= \Mellin{\U} \FI(\Dim+2)
	\label{eq:FI(d+2)}%
\end{equation}
where $\Mellin{\U} = \U(\PlusD{1},\ldots,\PlusD{N})$ is obtained from the polynomial $\U(x_1,\ldots,x_N)$ by substituting $x_i \mapsto \PlusD{i}$.
This \emph{raising} dimension shift was pointed out by Tarasov \cite{Tarasov:ConnectionBetweenFeynmanIntegrals},%
\footnote{%
	Tarasov considered the special case where all inverse propagators are of the form $\Den_e = k_e^2 - m_e^2$ and hence $\U$ is just the graph (first Symanzik) polynomial from \eqref{eq:graph-polynomials}. 
}
and had been observed before in special cases \cite{DerkachovHonkonenPismak:3loopWalk}.
For $\FImel(\Dim)=[\Gamma(\Dim/2-\sdc)/\Gamma(\Dim/2)]\, \FI(\Dim)$, the relation takes the form
\begin{equation}
	\FImel (\Dim)
	= \frac{s}{s+\sdc} \U(\PlusD{1},\ldots,\PlusD{N}) \FImel(\Dim+2)
	.
	\label{eq:FImel(d+2)}%
\end{equation}
At the same time, the representation $\FImel=\Mellin{\G^s}$ implies also that
\begin{equation}
	\FImel(\Dim) 
	= \G(\PlusD{1},\ldots,\PlusD{N}) \FImel(\Dim+2)
	= \Mellin{\G} \FImel(\Dim+2)
	.
	\label{eq:FImel(d+2)G}%
\end{equation}
\begin{rem}\label{rem:HG+sU}%
	The equality of \eqref{eq:FImel(d+2)} and \eqref{eq:FImel(d+2)G} implies, via the Mellin transform, that
\begin{equation*}
	H(s) \G + s\U 
	= -(s-\sum_a x_a \partial_a + \Loops s) \G + s \U 
	\in \Ann(\G^{s-1})
	.
\end{equation*}
Indeed, 
$
	H(s) \G^s
	= s(\Loops\U + (\Loops+1)\F) \G^{s-1}-s(\Loops+1)\G^s
	= -s\U \G^{s-1}
$
follows from the homogeneity of $\U$ and $\F$, see \eqref{eq:G-euler}.\footnote{%
	In fact, 
$
	H(s)\G + s\U = \sum_e x_e \left( \G \partial_e - (s-1) (\partial_e \G) \right)
$
follows from the trivial annihilators \eqref{eq:triv-Ann} of $\G^{s-1}$.
}
\end{rem}
A \emph{lowering} dimension shift, expressing $\FI(\Dim+2)$ in terms of $\FI(\Dim)$, corresponds to a Bernstein-Sato operator of $\G$ under the Mellin transform $\FImel(\Dim)=\Mellin{\G^{-\Dim/2}}$:\footnote{%
	Tkachov proposed in \cite{Tkachov:WhatsNext} to use a generalization of \eqref{eq:Bernstein-Sato} to several polynomials (the individual Symanzik polynomials $\U$ and $\F$, instead of $\G=\U+\F$), which in the physics literature is referred to as Bernstein-Tkachov theorem. However, this result is in fact due to Sabbah \cite{Sabbah:ProximiteII} (see also \cite{Gyoja:BernsteinSatoSeveral}).
}
\begin{defn}\label{def:Bernstein-Sato}
	A Bernstein-Sato operator $P(s)\in\Weyl{N}[s]$ for a non-constant polynomial $f$ is a polynomial differential operator such that there exists a polynomial $b(s) \in \C[s]$ with
\begin{equation}
	P(s) f^{s+1} = b(s) f^{s}
	.
	\label{eq:Bernstein-Sato}%
\end{equation}
\end{defn}
Such operators always exist, and the \emph{Bernstein-Sato polynomial} is the unique monic polynomial $b(s)$ of smallest degree for which \eqref{eq:Bernstein-Sato} has a non-zero solution \cite{Bernshtein:AnalyticContinuation,SatoShintaniMuro:Translation}. Given a solution of \eqref{eq:Bernstein-Sato} for $f=\G$, we get a lowering dimension shift relation:
\begin{equation}
	\FImel(\Dim+2)
	= \frac{1}{b(s-1)} \Mellin{P(s-1)} \FImel(\Dim)
	.
	\label{eq:Bernstein-shift-FI}%
\end{equation}
\begin{cor}\label{cor:dim-shift-closed}
	If we allow the coefficients to be rational functions $k=\C(s)$, every integral in $\Dim+2k$ dimensions can be written as an integral in $\Dim$ dimensions. In other words, the multiplication with $f$ is invertible on $\Weyl{N}_k f^s$. Put still differently, $f^s$ generates the full module
	\begin{equation}
		\Weyl{N}_k \cdot f^s
		= k[x,1/f] \cdot f^s
		.
		\label{eq:dim-shift-closed}%
	\end{equation}
\end{cor}
\begin{proof}
	By \eqref{eq:Bernstein-Sato}, $f^{s-n} = \frac{P(s-n)}{b(s-n)} \cdots \frac{P(s-1)}{b(s-1)} f^s \in \Weyl{N}_k \cdot f^s$ for all $n\in\N$.
\end{proof}
In general, computing a Bernstein operator is not at all trivial. 
But in the case of a complete set of irreducible scalar products ($N=\abs{\bilabels}$ and $\spMat$ is invertible), an explicit formula for the lowering dimension shift follows from Baikov's representation \cite{Baikov:ExplicitSolutions-multiloop} of Feynman integrals. We use the form~\eqref{eq:Baikov-FI} given by Lee in \cite{Lee:LL2010,Lee:DimensionalRecurrenceAnalyticalProperties}:

Recall that $(q_1,\ldots, q_{\AllMoms})=(\LoopMom_{1},\ldots,\LoopMom_{\Loops},\ExtMom_{1},\ldots,\ExtMom_{E})$ denotes the combined loop- and external momenta ($\AllMoms = \Loops+\ExtMoms$). We introduce the Gram determinants
\begin{equation}
	\Gram_n(s) \defas \det
	\begin{pmatrix}
		q_n \cdot q_n & \cdots & q_n \cdot q_{\AllMoms} \\
		\vdots & \ddots & \vdots \\
		q_{\AllMoms} \cdot q_n & \cdots & q_{\AllMoms} \cdot q_{\AllMoms} \\
	\end{pmatrix}
	= \det \left( s_{\set{i,j}} \right)_{n \leq i,j\leq \AllMoms }
	\label{eq:Gram-det}%
\end{equation}
and remark that $\Gram \defas \Gram_{\Loops+1}(s) = \det (\ExtMom_i \cdot \ExtMom_j)_{1\leq i,j\leq \ExtMoms}$ depends only on the external momenta.
Furthermore, note that $\Gram_1(s)$ is a polynomial in the scalar products $s_{\set{i,j}}=q_i \cdot q_j$. By \eqref{eq:sp-from-props}, we can think of it also as a polynomial in the denominators $\Den$.
\begin{defn}\label{def:Baikov}
	The \emph{Baikov polynomial} $\BP(y) \in \C[y_1,\ldots,y_N]$ is the polynomial defined by $\BP(\Den_1,\ldots,\Den_N) = \Gram_1(s)$.
\end{defn}
\begin{thm}[Baikov representation \cite{Lee:LL2010}]
	\label{thm:Baikov}%
	The Feynman integral \eqref{eq:FI-momentum} can be written as
	\begin{equation}
		\FI(\Dim)
	= \frac{c \cdot \pi^{-\Loops \ExtMoms/2-\Loops(\Loops-1)/4}}{\Gamma\left( \frac{\Dim-\ExtMoms-\Loops+1}{2} \right) \cdots \Gamma\left( \frac{\Dim-\ExtMoms}{2} \right)}
	\cdot \frac{(-1)^{\Loops\Dim/2}}{\Gram^{(\Dim-E-1)/2}}
	\left( \prod_{e=1}^{N} \int \frac{\dd y_e}{y_e^{\nu_e}} \right)
	\cdot \left\{ \BP(y) \right\}^{(\Dim-N-1)/2}
		\label{eq:Baikov-FI}%
	\end{equation}
	where $c\in\Q$ is a rational constant and the Baikov polynomial $\BP(y)$ has degree at most $\AllMoms=\Loops+\ExtMoms$.
	The contour of integration in \eqref{eq:Baikov-FI} is such that $\BP$ vanishes on its boundary.
\end{thm}
We include the proof in appendix~\ref{sec:Baikov}. For us, the interesting feature of this alternative formula is that the dimension appears with a positive sign in the exponent of the integrand. We can therefore directly read off
\begin{cor}[lowering dimension shift \cite{Lee:DimensionalRecurrenceAnalyticalProperties}]
	A Feynman integral in $\Dim+2$ dimensions can be expressed as an integral in $\Dim$ dimensions by\footnote{
		The first fraction can also be written as $2^{\Loops}/\left(\Dim-\Loops-\ExtMoms+1 \right)_{\Loops}$ in terms of the Pochhammer symbol (raising factorial) $a_{\Loops} = a (a+1) \cdots (a+\Loops-1)$.}
	\begin{equation}
		\FI(\Dim+2)
		= 
		\frac{(-1)^{\Loops}}{
			\left(\frac{\Dim-\Loops-\ExtMoms+1}{2} \right)
			\cdots
			\left(\frac{\Dim-\ExtMoms}{2} \right)
		}
		\frac{\BP(\Minus{1},\ldots,\Minus{N})}{\Gram} \FI(\Dim)
		.
		\label{eq:FI(d-2)}%
	\end{equation}
\end{cor}
\begin{proof}
	According to \eqref{eq:Baikov-FI}, $\FI(\Dim+2)$ is obtained by multiplying the integrand of $\FI(\Dim)$ with $(-1)^{\Loops} \BP(y)/\Gram$ and adjusting the $\Gamma$-factors in the prefactor as $\Gamma\left(\frac{\Dim+2-\ExtMoms}{2}\right) = \frac{\Dim-\ExtMoms}{2} \Gamma\left( \frac{\Dim-\ExtMoms}{2} \right)$ and so on. Multiplying the integrand of the Baikov representation with $y_e$ is equivalent to decrementing $\nu_e$, hence the multiplication of the integrand by $\BP(y)$ can be written as the action of $\BP(\Minus{1},\ldots,\Minus{N})$ on the integral.
\end{proof}

The equation~\eqref{eq:FI(d-2)} can be thought of as a Bernstein equation for $\FI(\Dim)$, or, equivalently, as a special type of integral relation: Combining \eqref{eq:FI(d-2)} with the raising dimension shift \eqref{eq:FI(d+2)}, we find that
\begin{equation}
	\left\{
		\frac{\BP(\Minus{1},\ldots,\Minus{N})}{\Gram}
		\U(\PlusD{1},\ldots,\PlusD{N})
		-
		  \frac{\ExtMoms-\Dim+2}{2}
		  \cdots
		  \frac{\ExtMoms-\Dim+(\Loops+1)}{2}
	\right\}
	\FI(\Dim)
	=0
	.
	\label{eq:Bernstein-Sato-FI}%
\end{equation}

We ask if this annihilator is contained in $\Mom$; in other words, whether the lowering dimension shift relation \eqref{eq:FI(d-2)} is a consequence of the momentum space IBP identities. This is what we will establish in the following proposition~\ref{prop:d-2 in Mom}.
First let us write the Baikov polynomial $\BP(y)$ explicitly: The block decomposition $(s_{i,j})_{1\leq i,j\leq\AllMoms} = \left(\begin{smallmatrix} V & B \\ B^{\Transpose} & G \\ \end{smallmatrix} \right)$ with $V=(\LoopMom_{i}\cdot \LoopMom_{j})_{i,j \leq \Loops}$, $B=(\LoopMom_{i} \cdot \ExtMom_{j})_{i\leq\Loops,j\leq\ExtMoms}$ and $G=(\ExtMom_{i}\cdot\ExtMom_{j})_{i,j\leq\ExtMoms}$ shows that
\begin{equation*}
	\frac{\BP(y)}{\Gram}
	= \det Q(y)
	\quad\text{where}\quad
	Q(\Den) \defas V - B G^{-1} B^{\Transpose}
\end{equation*}
is an $\Loops \times \Loops$ matrix whose entries are quadratic in the denominators. By \eqref{eq:sp-from-props},
\begin{equation}
	Q_{i,j}
	= \spMat^a_{\set{i,j}}( \Den_a - \lambda_a)
	- \spMat^a_{\set{i,r}}( \Den_a - \lambda_a) 
	  G^{-1}_{r,s}
	  \spMat^b_{\set{j,s}}( \Den_b - \lambda_b) 
	\label{eq:Qij-spMat}%
\end{equation}
where we suppress the explicit summation signs over $a=1,\ldots,N$ in the first summand and over $a,b=1,\ldots,N$ and $r,s=\Loops+1,\ldots,\AllMoms$ in the second summand.
\begin{prop}
	\label{prop:d-2 in Mom}%
	The annihilator \eqref{eq:Bernstein-Sato-FI}, corresponding to the lowering dimension shift \eqref{eq:FI(d-2)}, is contained in the ideal of shift operators generated by the momentum space IBP's from proposition~\ref{prop:IBP_Grozin}:
	\begin{equation}
		\det Q(\Minus{1},\ldots,\Minus{N})
		\cdot \U(\PlusD{1},\ldots,\PlusD{N})
		- \prod_{j=2}^{\Loops+1} \frac{\left(\ExtMoms+j-\Dim\right)}{2}
		\in 
		\sum_{i,j} \Shifts{N}[\Dim] \cdot \shiftIBP{i}{j}.
		\label{eq:FI-dimAnn-in-Mom}%
	\end{equation}
\end{prop}
\begin{proof}
	Let us abbreviate $\tilde{\loopMat} \defas \loopMat(\PlusD{1},\ldots,\PlusD{N})$ from \eqref{eq:M-from-A} such that $\U=\det \Lambda$ and similarly $\tilde{Q} \defas Q(\Minus{1},\ldots,\Minus{N})$ for the matrix \eqref{eq:Qij-spMat}. Using \eqref{eq:Qij-spMat} and \eqref{eq:Shifts}, we compute
	\begin{align*}
		\commu{\tilde{Q}_{i,j}}{\PlusD{c}}
		&=
			\spMat^a_{\set{i,j}} \commu{\Minus{a}}{\PlusD{c}} 
			- G^{-1}_{r,s} \spMat^a_{\set{i,r}} \spMat^b_{\set{j,s}}
			\left(
				(\Minus{a}-\lambda_a) \commu{\Minus{b}-\lambda_b}{\PlusD{c}}
				+ \commu{\Minus{a}-\lambda_a}{\PlusD{c}} (\Minus{b}-\lambda_b)
			\right)
		\\
		&= -\spMat^c_{\set{i,j}} + G^{-1}_{r,s} \left(
			\spMat^a_{\set{i,r}} \spMat^c_{\set{j,s}} (\Minus{a} - \lambda_a)
			+\spMat^c_{\set{i,r}} \spMat^b_{\set{j,s}} (\Minus{b} - \lambda_b)
		\right)
		.
	\end{align*}
	Contracting with the matrix $\spMat_c^{\set{k,l}}$ (for $k,l\leq \Loops$) by summing over $c$, we conclude that
	\begin{equation*}
		\commu{\tilde{Q}_{i,j}}{\spMat_c^{\set{k,l}} \PlusD{c}}
		= -\delta_{\set{i,j},\set{k,l}} + G^{-1}_{r,s} \left(
			\spMat^a_{\set{i,r}} \delta_{\set{j,s},\set{k,l}} (\Minus{a} - \lambda_a)
			+\delta_{\set{i,r},\set{k,l}} \spMat^b_{\set{j,s}} (\Minus{b} - \lambda_b)
		\right)
		.
	\end{equation*}
	Note that the indices $r$ and $s$ take values $>\Loops$, whereas $i$ and $j$ are $\leq \Loops$. Hence $\delta_{\set{j,s},\set{k,l}}=\delta_{\set{i,r},\set{k,l}} = 0$. So, recalling \eqref{eq:M-from-A}, we finally arrive at
	\begin{equation}
		\commu{\tilde{Q}_{i,j}}{\tilde{\loopMat}_{k,l}}
		= \frac{1+\delta_{k,l}}{2} \delta_{\set{i,j},\set{k,l}}
		= \frac{\delta_{i,k}\delta_{j,l} + \delta_{i,l}\delta_{j,k}}{2}
		.
		\label{eq:Qij,Lkl}%
	\end{equation}
	Recall \eqref{eq:def-U-F}, that $\U(x) = \det \loopMat(x)$, such that 
	\begin{equation*}
		\det Q(\Minus{1},\ldots,\Minus{N}) \cdot \U(\PlusD{1},\ldots,\PlusD{N})
		= \det \tilde{Q} \cdot \det \tilde{\loopMat}.
	\end{equation*}
	We can now invoke an identity of Turnbull \cite{Turnbull:SymmetricCayleyCapelli}, see also \cite{FoataZeilberger:CombinatorialCapelliTurnbull} for a combinatorial and \cite{CaraccioloSportielloSokal:NonDetI} for an algebraic proof, which relates this product of determinants to a determinant of the product $\tilde{Q} \cdot \tilde{\loopMat}$. This is non-trivial, because the elements of these two matrices do not commute, according to \eqref{eq:Qij,Lkl}.
	Turnbull's identity, as stated in \cite[Proposition~1.4]{CaraccioloSportielloSokal:NonDetI}, applies precisely to this kind of very mild non-commutativity \eqref{eq:Qij,Lkl} and states that
	\begin{equation}
		\det \tilde{Q} \cdot \det \tilde{\loopMat}
		= \coldet \left( \tilde{Q} \cdot \tilde{\loopMat} + Q_{col} \right)
		,
		\quad\text{where}\quad
		(Q_{col})_{i,j}
		\defas -\frac{\Loops-i}{2} \delta_{i,j}
		\label{eq:Turnbull}%
	\end{equation}
	is a simple diagonal matrix and $\coldet$ denotes the column-ordered determinant
	\begin{equation}
		\coldet A
		\defas \sum_{\sigma \in S_N} \sign(\sigma) A_{\sigma(1),1} \cdots A_{\sigma(N),N}
		.
		\label{eq:coldet}%
	\end{equation}
	So let us now compute the entries of the product of $\tilde{Q}$ from \eqref{eq:Qij-spMat} with $\tilde{\loopMat}$. Firstly,
	\begin{equation*}
		\spMat^b_{\set{j,s}}(\Minus{b}-\lambda_b) \tilde{\loopMat}_{j,k}
		= -\frac{1}{2} \spMat^b_{\set{j,s}}(1+\delta_{j,k})\spMat^{\set{j,k}}_c (\Minus{b}-\lambda_b) \PlusD{c}
		= -\frac{1}{2} C_{cs}^{bk} \left\{ \PlusD{c}(\Minus{b}-\lambda_b)-\delta_{b,c} \right\}
	\end{equation*}
	according to \eqref{eq:IBP-coeffs}. Note that $C_{bs}^{bk} = \sum_{b,j} (1+\delta_{j,k}) \spMat_b^{\set{j,k}} \spMat^b_{\set{j,s}} = \sum_j (1+\delta_{j,k}) \delta_{\set{j,k},\set{j,s}} = 0$ due to $s>\Loops \geq k$. So we can rewrite, due to \eqref{eq:shift-IBP},
	\begin{align*}
		\spMat^b_{\set{j,s}}(\Minus{b}-\lambda_b) \tilde{\loopMat}_{j,k}
		=
		\frac{1}{2} \shiftIBP{k}{s} 
		+ \frac{1}{2}\sum_{m>\Loops} \spMat^{\set{k,m}}_b \PlusD{b} (q_s \cdot q_m)
		.
	\end{align*}
	Note that $q_s \cdot q_m = G_{s,m}$ such that contraction of the second summand with $G^{-1}_{r,s}$ produces $\delta_{r,m}$. So the sum over $m$ collapses, and up to the term with $\shiftIBP{k}{s}$, $\tilde{Q}_{i,j} \tilde{\loopMat}_{j,k}$ is
	\begin{multline*}
		-\frac{1+\delta_{j,k}}{2} \spMat^a_{\set{i,j}} (\Minus{a} - \lambda_a) \spMat_c^{\set{j,k}} \PlusD{c}
		-\frac{1}{2} \spMat^a_{\set{i,r}}(\Minus{a} - \lambda_a) \spMat^{\set{k,r}}_b \PlusD{b}
		=-\frac{1}{2} C^{ak}_{bi} (\Minus{a}-\lambda_a) \PlusD{b}
		\\
		=\frac{1}{2} \left\{ \shiftIBP{k}{i} -\Dim \delta^{k}_{i} + C^{ak}_{ai} \right\}
		=\frac{1}{2} \left\{ \shiftIBP{k}{i} -\Dim \delta^{k}_{i} + (\AllMoms+1) \delta^k_i \right\},
	\end{multline*}
	where the term
	$C_{ai}^{ak}
	= \sum_{a,j} (1+\delta_{j,k}) \spMat_a^{\set{j,k}} \spMat^a_{\set{j,i}}
	= \sum_j (1+\delta_{j,k}) \delta_{\set{j,k},\set{j,i}}
	=\delta_{i}^{k} (\AllMoms+1)$
	comes from commuting $\Minus{a}$ with $\PlusD{b}$.
	Putting our results together, we arrive at
	\begin{equation}
		\tilde{Q}_{i,j} \tilde{\loopMat}_{j,k}
		= \frac{1}{2} \left\{ 
			(\AllMoms+1-\Dim)\delta^{k}_{i}
			+ \shiftIBP{k}{i} 
			- \spMat^a_{\set{i,r}} (\Minus{a}-\lambda_a) G^{-1}_{r,s} \shiftIBP{k}{s}
		\right\}
		.
		\label{eq:Qij-loopMatjk}%
	\end{equation}
	So if we ignore all terms that lie in the (left) ideal generated by the momentum space (shift) operators $\shiftIBP{i}{j}$, the column determinant \eqref{eq:coldet} of the matrix $\tilde{Q} \cdot \tilde{\loopMat} + Q_{col}$ from \eqref{eq:Turnbull} can be replaced by an ordinary determinant $\det B$ of the diagonal matrix
	\begin{equation*}
		B_{i,j} 
		= \frac{\delta_{i,j}}{2} \left(\AllMoms+1-\Dim -\Loops+i \right)
		= \frac{\delta_{i,j}}{2} \left(\ExtMoms+1-\Dim+i \right),
	\end{equation*}
	such that indeed we conclude with the result that
	\begin{equation*}
		\det \tilde{Q} \cdot \det \tilde{\loopMat}
		\equiv \prod_{i=1}^{\Loops} \frac{\ExtMoms+1-\Dim+i}{2}
		\mod
		\sum_{i,j} \Shifts{N}[\Dim] \cdot \shiftIBP{i}{j}.
		\qedhere
	\end{equation*}
\end{proof}
The Mellin transform $\FI(\Dim)=\Mellin{\U^s e^{-\F/\U}}$ identifies \eqref{eq:Bernstein-Sato-FI} with the annihilator
\begin{equation}
	\left\{ \det Q(-\partial) \,\cdot\, \U - \tilde{b}(s) \right\} \bullet \U^s e^{-\F/\U}
	= 0,
	\quad\text{where}\quad
	\tilde{b}(s) \defas \prod_{j=2}^{\Loops+1} \left(s+\frac{\ExtMoms+j}{2} \right)
	.
	\label{eq:Bernstein-Sato-FI-par}%
\end{equation}
To phrase this in terms of $\FImel(\Dim) = \Mellin{\G^s} = \FI(\Dim) \cdot \Gamma(-s-\sdc)/\Gamma(-s)$, we can use $\U \G^s = -H(s+1) \G^{s+1}/(s+1)$ from remark~\ref{rem:HG+sU} to conclude that
\begin{equation}
	\Gamma(H(s)) \cdot \det Q(-\partial) \cdot \frac{1}{\Gamma(H(s)-\Loops-1)} \cdot \G^{s+1} 
	= -(s+1)\tilde{b}(s) \G^{s},
	\label{eq:Bernstein-Sato-par}%
\end{equation}
where $H(s)=\Mellin[-1]{-s-\sdc} = -s(\Loops+1)+\sum_{i=1}^N \theta_i$. Recall from \eqref{eq:FImel-relation-to-FI} that the left-hand side of \eqref{eq:Bernstein-Sato-par} can be written, in terms of the homogeneous components $Q_r$ (with degree $r$) of $\det Q(-\partial) = \sum_r Q_r$, as
\begin{equation*}
	\sum_r \frac{\Gamma(H(s))}{\Gamma(H(s)-L-1+r)} \, Q_r
	= \sum_{r \leq \Loops+1} \left[ \prod_{i=1}^{\Loops+1-r} (H(s)-i) \right] Q_r
	+ \sum_{r > \Loops+1} \left[ \prod_{i=0}^{r-\Loops-1} \frac{1}{H(s)+i} \right] Q_r
	.
\end{equation*}
If $r \leq \Loops + 1$, this is a polynomial differential operator, and we thus obtained an explicit Bernstein-Sato operator as in definition~\ref{def:Bernstein-Sato}.
\begin{cor}
	If the degree of the Baikov polynomial $\BP(y)$ is not more than $\Loops+1$, then the Bernstein-Sato polynomial $b(s)$ of the Lee-Pomeransky polynomial $\G$ is a divisor of $(s+1)\tilde{b}(s)$. In particular, all roots of $b(s)/(s+1)$ are simple and at half-integers.
\end{cor}
Note that $\deg \BP(y) \leq \min\set{2\Loops,\AllMoms} = \Loops + \min\set{\Loops,\ExtMoms}$ by definition~\ref{def:Baikov} and equation \eqref{eq:Qij-spMat}, so in particular, the corollary applies to all propagator graphs ($\ExtMoms=1$) and to all graphs with one loop ($\Loops=1$).

\section{Euler characteristic as number of master integrals}
\label{sec:Euler-NoM}%

Here we will show, using the theory of Loeser and Sabbah \cite{LoeserSabbah:IrredTore}, that the number of master integrals equals the Euler characteristic of the complement of the hypersurface defined by $\G=0$ inside the torus $\Gm^N$ (we write $\Gm = \Aff\setminus\set{0}$ for the multiplicative group and $\Aff$ for the affine line).
For a full understanding of this section, some knowledge of basic $D$-module theory is indispensable; but we tried to include sufficient detail for the main ideas to become clear to non-experts as well. In particular, we will give self-contained proofs that only use $D$-module theory at the level of \cite{Coutinho:Primer}.
\begin{defn}
	By $V_{\G}$ we denote the vector space of all Feynman integrals associated to $\G$, over the field $\C(s,\nu) \defas \C(s,\nu_1,\ldots,\nu_N)$ of rational functions (in the dimension and indices).
	More precisely, with $\FImel_{\G} \defas \Mellin{\G^s}$,
	\begin{equation}
		V_{\G}
		\defas
		\sum_{n \in \Z^N} \C(s,\nu) \cdot \FImel_{\G}(\nu+n)
		=
		\C(s,\nu) \tp_{\C[s,\nu]} \left( \Shifts{N}[s] \cdot \FImel_{\G} \right)
		.
		\label{eq:def-V_G}%
	\end{equation}
	The \emph{number of master integrals} is the dimension of this vector space:
	\begin{equation}
		\NoM{\G}
		\defas \dim_{\C(s,\nu)} V_{\G}
		.
		\label{eq:def-number-of-masters}%
	\end{equation}
\end{defn}
Note that this is the same as the dimension of the space $\sum_n \C(s,\nu) \FI_{\G}(\nu+n)$ of Feynman integrals \eqref{eq:FI-momentum}, because the ratios $\FI_{\G}(\nu+n)/\FImel_{\G}(\nu+n) = \Gamma(-s)/\Gamma(-s-\sdc-\abs{n})$ with $\abs{n}=n_1+\ldots+n_N$ are all related by a rational function in $\C(s,\nu)$, see \eqref{eq:FImel-relation-to-FI}.
\begin{rem}[\textbf{Warning}]
	The phrase ``master integrals'' is used in the existing physics literature to denote various different quantities, and none of those notions coincides exactly with ours.
	The main sources for discrepancies are:
	\begin{enumerate}
		\item Almost always the integrals are considered only for integer indices $\nu \in \Z^N$, instead of as functions of arbitrary indices. In this setting, integrals with at least one $\nu_e=0$ can be identified with quotient graphs (``subtopologies'') and are often discarded from the counting of master integrals.
		\item We only discuss relations of integrals that are expressible as linear shift operators acting on a single integral. 
	This setup cannot account for relations of integrals of different graphs (with some fixed values of the indices), as for example discussed in \cite{KniehlKotikov:EffectiveMass}.
	It also excludes symmetry relations, which are represented by permutations of the indices $\nu_e$.
	\item Some authors do not count integrals if they can be expressed in terms of $\Gamma$-functions or products of simpler integrals, for example \cite{KniehlKotikov:EffectiveMass,KalmykovKniehl:CountingMasters}.
	\end{enumerate}
	Taking care of these subtleties, we will demonstrate in section~\ref{sec:examples} that our definition gives results that do match the counting of master integrals obtained by other methods.
\end{rem}

A fundamental result for methods of integration by parts reduction is that the number of master integrals is finite. This was proven in \cite{SmiPet} for the case of integer indices $\nu \in \Z^N$, using the momentum space representation. Below we will show that this result holds much more generally, for unconstrained $\nu$, and that it becomes a very natural statement once it is viewed through the parametric representation.
Notably, it remains true for Mellin transforms $\Mellin{\G^s}$ of \emph{arbitrary} polynomials $\G$---the fact that $\G$ comes from a (Feynman) graph is completely irrelevant for this section.

Recall that, by the Mellin transform, we can rephrase statements about integrals in terms of the parametric integrands. In line with \eqref{eq:integrands-FIs-iso} and \eqref{eq:def-V_G}, we can rewrite \eqref{eq:def-number-of-masters} as
\begin{equation*}
	\NoM{\G} 
	= \dim_{\C(s,\theta)} \left(
		\C(s,\theta)
		\tp_{\C[s,\theta]}
		\Weyl{N}[s] \cdot \G^s
	\right)
	,
\end{equation*}
where $\C[s,\theta] = \C[s,\theta_1,\ldots,\theta_N]$ denotes the polynomials in the dimension $s=-\Dim/2$ and the operators $\theta_e \defas x_e \partial_e = \Mellin[-1]{-\nu_e}$, and $F \defas \C(s,\theta)$ stands for their fraction field (the rational functions in these variables). Since $F$ contains $k \defas \C(s)$, we can equivalently work over this base field throughout and write
\begin{equation}
	\NoM{\G} 
	= \dim_F (F\tp_R\M)
	\label{eq:NoM-as-abstract-Mellin}%
\end{equation}
in terms of $R \defas k[\theta]$ and the module $\M = \Weyl{N}_k \!\cdot\! \G^s$ over the Weyl algebra $\Weyl{N}_k \defas \Weyl{N} \tp_{\C} k = \Weyl{N}[s] \tp_{\C[s]} k$ over the field $k=\C(s)$.
Crucially, $\Weyl{N}_k \!\cdot\! \G^s$ is a \emph{holonomic} $\Weyl{N}_k$-module, which is a fundamental result due to Bernstein~\cite{Bernshtein:AnalyticContinuation}.

Holonomic modules are, in a precise sense, the most constrained, and behave in many ways like finite-dimensional vector spaces. For example, sub- and quotient modules, direct and inverse images of holonomic modules are again holonomic \cite{Kashiwara:RationalityOfRoots,Kashiwara:OnHolonomicSystemsII}, and holonomic modules in zero variables are precisely the finite-dimensional vector spaces. The holonomicity of the parametric integrand was already exploited in \cite{KashiwaraKawai:HolonomicSystemsFI} to show that Feynman integrals fulfill a holonomic system of differential equations, and it is also a key ingredient in the proof in \cite{SmiPet}.

The number defined in \eqref{eq:NoM-as-abstract-Mellin} has been studied by Loeser and Sabbah~\cite{LoeserSabbah:IrredTore} in a slightly different setting, namely for holonomic modules over the algebra 
\begin{equation}
	\WeylT{N}_k
	\defas k[x_1^{\pm 1},\ldots,x_N^{\pm 1}] \langle \partial_1,\ldots,\partial_N \rangle
	= k[x^{\pm 1}] \tp_{k[x]} \Weyl{N}_k 
	= \Weyl{N}_k [x^{-1}]
	\label{eq:WeylT}%
\end{equation}
of linear differential operators on the torus $\Gm[k]^N$. Note that $\WeylT{N}_k$ is just the localization of $\Weyl{N}_k$ at the coordinate hyperplanes $x_i=0$; that is, the coefficients of the derivations are extended from polynomials $\OO(\Aff^N_k)=k[x]$ to rational functions $\OO(\Gm[k]^N)=k[x^{\pm 1}]$ whose denominator is a monomial in the coordinates $x_i$. Equivalently, we can also view $\WeylT{N}_k= \GmA^{\ast} \Weyl{N}_k$ as the pull-back under the (open) inclusion
\begin{equation*}
	\GmA\colon \Gm[k]^N \injects \Aff^N_k
	.
\end{equation*}
The pull-back along $\GmA$ turns every $\Weyl{N}_k$-module $\M$ into a $\WeylT{N}_k$-module $\GmA^{\ast} \M$, namely the localization $\GmA^{\ast} \M = k[x^{\pm 1}] \tp_{k[x]} \M = \M[x^{-1}]$. Importantly, if $\M$ is holonomic, so is its pull-back $\GmA^{\ast} \M$. The starting point for this section is
\begin{thm}[Loeser \& Sabbah \cite{LoeserSabbah:IrredTore,LoeserSabbah:IrredToreII}]
	\label{thm:L-S}%
	Let $\M$ denote a holonomic $\Weyl{N}_k$-module. Then $F \tp_R \M$ is a finite-dimensional vector space over $F$. Moreover, its dimension is given by the Euler characteristic $\dim_F (F \tp_R \M) = \chi\left(\GmA^{\ast} \M \right)$.
\end{thm}
In appendix~\ref{sec:L-S}, we provide a self-contained proof of this crucial theorem, simpler and more explicit than in \cite{LoeserSabbah:IrredToreII}.
For now, let us content ourselves with reducing it to the known situation on the torus.
\begin{proof}
	We can invoke $\dim_F(F \tp_R \GmA^{\ast} \M) = \chi(\GmA^{\ast}\M) < \infty$ from \cite[Th\'{e}or\`{e}me~2]{LoeserSabbah:IrredToreII}. To conclude, we just need to note that $F \tp_R \M$ and $F \tp_R \M[x^{-1}] = F \tp_R \GmA^{\ast}\M$ are isomorphic vector spaces (over $F$). This is clear since each coordinate $x_i$ is invertible after localizing at $F$: due to $\partial_i x_i = 1+x_i \partial_i$, we find that $(1+x_i \partial_i)^{-1} \tp \partial_i \in F\tp_R \Weyl{N}_k$ is an inverse to $1\tp x_i$.
\end{proof}
This result not only implies the mere finiteness of the number of master integrals, but in addition gives a formula for this number---it is the Euler characteristic, given by
\begin{equation}
	\chi\left( \M' \right)
	\defas
	\chi(\DR(\M'))
	= \sum_{i} (-1)^i \dim_k H^i \big( \DR( \M')  \big)
	,
	\label{eq:def-Euler}%
\end{equation}
of the algebraic de Rham complex of $\M'\defas \GmA^{\ast} \M$. This is the complex
\begin{equation}
	\DR(\M')
	\defas
	\left( \Omega^\bullet_{\Gm[k]^N} \otimes_{\OO(\Gm[k]^N)} \M'[N], \dd \right)
	\label{eq:def-DR}%
\end{equation}
of $\M'$-valued differential forms on the torus $\Gm[k]^N$, with the connection $\dd (\omega \tp m) = \dd \omega \tp m + \sum_{i=1}^N (\dd x_i \wedge \omega) \tp \partial_i m$. Note that the $r$-forms $\omega$ are shifted to sit in degree $r-N$ of the complex, which is thus supported in degrees between $-N$ and $0$; hence \eqref{eq:def-Euler} is a finite sum over $0 \leq i \leq N$. The extremal cohomolgy groups are easily identified as
\begin{equation}
	H^{-N}\left( \DR(\M') \right)
	= \bigcap_{i=1}^N \ker \partial_i
	\quad\text{and}\quad
	H^0\left( \DR(\M') \right)
	\isomorph \left. \M' \middle/ \sum_{i=1}^N \partial_i \M' \right.
	= \pi_{\ast} \M'
	,
\end{equation}
with the latter also known as push-forward of $\M'$ under the projection $\pi\colon \Gm[k]^N \surjects \Aff_k^0$ to the point.
Since holonomicity is preserved under direct images, we conclude that $\dim_k H^0\left( \DR(\M') \right)$ is finite.\footnote{%
	Recall that a holonomic module over the point $\Aff^0_k$ is the same as a finite-dimensional $k$-vector space.
}
In fact, the same is true for the other de Rham cohomology groups, which shows that \eqref{eq:def-Euler} is indeed well-defined.\footnote{%
	The de Rham complex $\DR(\Weyl{N}_k)$ is a resolution of $k[x]$ by free $\Weyl{N}_k$-modules, such that $H^{\bullet} (\DR(\M'))$ are the (left) derived functors of $\pi_{\ast} \M'=H^0(\DR(\M'))$. In the language of derived categories, saying that $\pi_{\ast} \M'$ is holonomic actually means precisely that $\DR(\M')$ is a complex with cohomology groups that are holonomic modules over the point---that is, finite-dimensional vector spaces over $k$.
}

Since we are interested in Feynman integrals, we consider the special case where the $\Weyl{N}_k$-module $\M$ is simply $\M=\Weyl{N}_k \cdot \G^s$ from definition~\ref{def:f^s}. Its elements can be written uniquely in the form $h\cdot \G^s$, where $h \in k[x,\G^{-1}]$, such that $\Weyl{N}_k \cdot \G^s \isomorph k[x,\G^{-1}]$ by \eqref{eq:dim-shift-closed} are isomorphic as $k[x]$-modules: $x_i (h \G^s) = (x_i h) \G^s$.
The action \eqref{eq:f^s-def} of the derivatives, however, is twisted by a term proportional to $s$:
$\partial_i (h \G^s) = \G^s (\partial_i h + sh(\partial_i \G)/\G)$.
Despite this twist, we find that the Euler characteristic stays the same:
\begin{prop}\label{prop:chi(G^s)=chi(1/G)}
	Let $\G\in \C[x_1,\ldots,x_N]$ be a polynomial and set $k=\C(s)$. Then the Euler characteristics of the algebraic de Rham complexes of the holonomic $\Weyl{N}_k$-module $\GmA^{\ast} \Weyl{N}_k \G^s$ and the holonomic $\Weyl{N}_{\C}$-module $\C[x^{\pm 1},\G^{-1}]=\OO(\Gm[k]^N \setminus \Vanish(\G))$ coincide:
	\begin{equation}
		\chi(\GmA^{\ast} \Weyl{N}_{k} \G^s)
		= \chi\left(\C[x^{\pm 1},\G^{-1}]\right)
		.
		\label{eq:chi(G^s)}%
	\end{equation}
\end{prop}
In particular, we can dispose of the parameter $s$ completely and compute with the algebraic de Rham complex of $\C[x^{\pm 1},\G^{-1}]$, which is the ring of regular functions of the complement of the hypersurface $\Vanish(\G)=\setexp{x}{\G(x)=0}$ in the torus $\Gm^N$. 
Combining theorem~\ref{thm:L-S} with proposition~\ref{prop:chi(G^s)=chi(1/G)}, we thus obtain our main result:
\begin{cor}\label{cor:number-of-masters}
	The number of master integrals of an integral family with $N$ denominators is
	\begin{equation}
		\NoM{\G}
		= \chi \left(\C[x^{\pm 1},\G^{-1}] \right),
		\label{eq:NoM-alg}%
	\end{equation}
	the Euler characteristic of the algebraic de Rham complex of the complement of the hypersurface $x_1\cdots x_N \cdot \G=0$ inside the affine plane $\Aff^N$.
	Via Grothendieck's comparison isomorphism, this is the same as the topological Euler characteristic, up to a sign:\footnote{
		This sign arises from the shift by $N$ in the definition \eqref{eq:def-DR} of the de Rham complex $\DR$.
	}
	\begin{equation}
		\NoM{\G}
		= (-1)^N \chi \left(\C^N \setminus \set{x_1\cdots x_N \cdot \G=0} \right)
		= (-1)^N \chi \left(\Gm^N \setminus \set{\G=0} \right)
		.
		\label{eq:NoM-topo}%
	\end{equation}
\end{cor}
\begin{rem}
	We stress that this geometric interpretation of the number of master integrals is valid for dimensionally regulated Feynman integrals, that is, we consider them as meromorphic functions in $\Dim$ (and $\nu$). This is reflected in our treatment of $s=-\Dim/2$ as a symbolic parameter.
	
	If, instead, one specializes to a fixed dimension like $\Dim=2$ ($s=-1$) or $\Dim=4$ ($s=-2$), then \eqref{eq:dim-shift-closed} is no longer true in general.\footnote{%
		It fails precisely if, for some $r \in \N$, $s-r$ is a zero of the Bernstein-Sato polynomial of $\G$.
	}
	It can thus happen that $\Weyl{N} \cdot \G^s \subsetneq \C[x,\G^{-1}]$ is a proper subalgebra (note $k=\C(s)=\C$). While theorem~\ref{thm:L-S} still applies and relates the number of master integrals in a fixed dimension to the Euler characteristic of $\WeylT{N} \cdot \G^s$, this is not always equal to the topological Euler characteristic \eqref{eq:NoM-topo}.
	This is expected, since the number of master integrals is known to be different in fixed dimensions \cite{Tancredi:IBPinteger}.
\end{rem}
\begin{proof}[Proof of proposition~\ref{prop:chi(G^s)=chi(1/G)}]
	Given an $\Weyl{N}_{\C}$-module $\M$ and a polynomial $f \in \C[x]$, set $\M' \defas \M[f^{-1}]$ and consider the $\Weyl{N}_{\C}$-module $\M'f^s$ formed by products of $f^s$ with elements $m\in \M'[s] \defas \M' \tp_{\C} \C[s]$. As vector spaces, $\M' f^s \isomorph \M'[s]$ via $mf^s \mapsto m$, but the $\Weyl{N}_{\C}$ action on $\M'f^s$ has twisted derivatives to take into account the factor $f^s$:
	\begin{equation}
		x_i \bullet mf^s \defas x_i m f^s
		\quad\text{and}\quad
		\partial_i \bullet m f^s
		\defas \left\{ (\partial_i m) + s m \frac{\partial_i f}{f} \right\} f^s
		.
		\label{eq:f^s-twist}%
	\end{equation}
	Following Malgrange \cite{Malgrange:BernsteinIsolee}, we introduce the action of a further variable $t$ by setting
	\begin{equation}
		t \bullet m(s) f^s
		\defas m(s+1) f^{s+1}
		\quad\text{and}\quad
		\partial_t \bullet m(s) f^s
		\defas -s m(s-1) f^{s-1}
		,
		\label{eq:Malgrange-action}%
	\end{equation}
	where we use the intuitive abbreviation $f^{s+r}\defas f^r \cdot f^s$ for $r\in\Z$.
	One easily verifies $[\partial_t,t] = 1$ and $[\partial_i,\partial_t]=[\partial_i,t]=[x_i,\partial_t]=[x_i,t] = 0$, such that $\M' f^s$ becomes an $\Weyl{N+1}_{\C}$-module in the $N+1$ variables $(x_1,\ldots,x_N,t)$. Note that $\partial_t t = -s$, so
	\begin{equation*}
		\frac{\M' f^s}{\partial_t \M' f^s}
		= \frac{\M' f^s}{s t^{-1} \M' f^s}
		= \frac{\M' f^s}{s \M' f^s}
		\isomorph
		\M'
	\end{equation*}
	is an isomorphism of $\Weyl{N}_{\C}$-modules.\footnote{%
		This is also clear from the fact that $\M' f^s/(\partial_t \M' f^s) = \pi_{\ast} (\M' f^s)$ is the push-forward of $\M' f^s$ under the projection $\pi\colon \Aff^{N+1} \longrightarrow \Aff^N$ forgetting the last coordinate. Namely, since $\M' f^s = F_{\ast} \M'$ as we discuss below, $\pi_{\ast} (F_{\ast} \M') = (\pi \circ F)_{\ast} \M' = \id_{\ast} \M' = \M'$.
	}
	Since $\partial_t$ is injective on $\M' f^s$ (it raises the degree in $s$), the de Rham complex $\DR(\M'f^s)$ is quasi-isomorphic to $\DR(\M' f^s/\partial_t \M' f^s) = \DR(\M')$, see corollary~\ref{cor:Koszul-qis}. So we can conclude the equality
	\begin{equation}
		\chi(\M' f^s)
		= \chi(\M')
		,
		\label{}
	\end{equation}
	once we assume that $\M$ is holonomic to ensure that these Euler characteristics are well defined. Indeed, the holonomicity of $\M'$ and $\M' f^s$ holds because
	\begin{itemize}
		\item $\M'=j^{\ast} \M$ is the pull-back of $\M$ under the inclusion $j\colon \Aff^N_{\C} \setminus \set{f=0} \injects \Aff^N_{\C}$,
		\item $\M' f^s = F_{\ast} \M'$ is the push-forward of $\M'$ under the closed embedding $F\colon \Aff^N_{\C} \injects \Aff^{N+1}_{\C}$ which sends $x$ to $(x, f(x))$.\footnote{%
		It follows from \eqref{eq:Malgrange-action} that $\M'f^s=\bigoplus_{n \geq 0} \partial_t^n \M' \isomorph F_{\ast} \M$ as $\C[x]$-modules, since $\partial^n_t \M'=\partial^n_t t^n \M' \equiv s^n \M' \mod s^{<n} \M'$. Furthermore, the derivatives act on $F_{\ast} \M$ by
		$\partial_i \bullet m(s)
		= (\partial_i - (\partial_i f) \partial_t)m(s)
		= (\partial_i +s (\partial_i f) t^{-1})m(s)
		= (\partial_i +s (\partial_i f)/f)m(s-1)
		$
		in accordance with \eqref{eq:f^s-twist}.
	}
	\end{itemize}
	Alternatively, the filtration $\Gamma_j \M' f^s \defas f^{-j} \sum_{i\leq j} s^i \Gamma_{2j(\deg f)-i} \M$ induced by any good filtration $\Gamma_{\bullet}$ on $\M$ directly shows the holonomicity of $\M' f^s$, since its dimension grows like $j^{N+1}$ for large $j$.
	We now invoke the theory of Loeser-Sabbah to deduce that
	\begin{equation*}
		\chi(\M'f^s)
		= \dim_{{\C}(\theta,t\partial_t)} \M'f^s (\theta, t\partial_t)
		= \dim_{k(\theta)} \Nmod (\theta)
		= \chi(\Nmod)
	\end{equation*}
	where $\Nmod \defas \M'f^s (t\partial_t)$ denotes the algebraic Mellin transform \eqref{eq:def-alg-Mellin} of $\M'f^s$ with respect to the coordinate $t$. But note that, according to \eqref{eq:Malgrange-action}, localizing at $t\partial_t = -s-1$ just extends the coefficients to $k=\C(s)$. So $\Nmod = \M' \tp_{\C} \C(s)f^s = \M\tp_{\C} k f^s$ is just the holonomic $\Weyl{N}_{k}$-module on the left-hand side of \eqref{eq:chi(G^s)}, because, over $k$, $f$ is invertible by corollary~\ref{cor:dim-shift-closed}.
	We have proven $\chi(\M \tp_{\C} k f^s) = \chi(\M[f^{-1}])$, and the special case of $\M=\C[x^{\pm 1}]$ with $f = \G$ proves the claim.
\end{proof}
\begin{rem}
	More abstractly, Proposition~\ref{prop:chi(G^s)=chi(1/G)} can also be seen as an application of the theory of characteristic cycles \cite{Ginsburg:CharVarVanCyc}: 
It is known that the Euler characteristic only depends on the characteristic cycle of a $\Weyl{N}_k$-module, which follows from the Dubson-Kashiwara formula \cite[equation~(6.6.4)]{Laumon:DeriveeFiltres}.
Therefore it is sufficient to show that the $\Weyl{N}_k$-modules $k[x^{\pm 1},\G^{-1}]$ and $k[x^{\pm 1}] \G^s$ have the same characteristic cycles, via \cite[Theorem~3.2]{Ginsburg:CharVarVanCyc}. This follows from the fact that these modules are identical up to the twist by the isomorphism $\partial_i \mapsto \partial_i + s (\partial_i \G)/\G$ of $\WeylT{N}_k[\G^{-1}]$.
\end{rem}

\subsection{No master integrals}
\label{sec:no-masters}

Corollary~\ref{cor:number-of-masters} shows in particular that there are no master integrals, $\NoM{\G}=0$, precisely when the Euler characteristic $\chi(\Gm^{N}\setminus \Vanish(f)) \defas \chi(\C[x^{\pm 1},f^{-1}]) = \chi( (\C^{\ast})^N \setminus \Vanish(f))$ vanishes for $f=\G$. For example, this happens if $f$ is homogeneous in a generalized sense: Suppose we can find $\lambda_0,\ldots,\lambda_N \in \Z$, not all zero, such that
\begin{equation}
	f\left(x_1 t^{\lambda_1},\ldots,x_N t^{\lambda_N}\right)
	= t^{\lambda_0} f(x_1,\ldots,x_N)
	\quad\text{in}\quad
	\C[x,t^{\pm 1}]
	;
	\label{eq:quasi-hom}%
\end{equation}
which is equivalent (apply $\partial_t$ and set $t=1$) to the existence of a linear annihilator,
\begin{equation}
	P^{\lambda}_s\bullet f^s = 0,
	\quad\text{of the form}\quad
	P^{\lambda}_s \defas \sum_{i=1}^N \lambda_i \theta_i - s\lambda_0
	\in \Z[s,\theta] \setminus \set{0}
	.
	\label{eq:lin-euler-ann}%
\end{equation}
\begin{lem}\label{lem:chi=0-poly-ann}
	Given $f\in\C[x_1,\ldots,x_N]$, the Mellin transform $\Mfrak \defas \M \tp_{\C(s)[\theta]} \C(s,\theta)$ of $\M \defas \WeylT{N}_k\!\cdot\! f^s$ is zero if, and only if, $f^s$ is annihilated by a polynomial in the Euler operators $\theta$:
	\begin{equation}
		\Ann_{\Weyl{N}_k}(f^s) \cap \C[s,\theta] \neq \set{0}.
		\label{eq:euler-poly-ann}%
	\end{equation}
\end{lem}
\begin{proof}
	Clearly, $\Mfrak=\set{0}$ requires $f^s$ to be mapped to zero in the localization $\Mfrak$ of $\M$ at $\C[s,\theta]\setminus\set{0}$, and therefore the existence of a non-zero polynomial $P(\theta,s)\in \C[s,\theta]$ with $P(\theta,s)\bullet f^s =0$. Conversely, given such an operator, its shifts $P(\theta-\alpha,s+r)$ by $(r,\alpha)\in\Z^{1+N}$ annihilate the elements $x^{\alpha} \cdot f^{s+r}$, which are therefore all mapped to zero in $\Mfrak$. By linearity, this proves $\Mfrak=\set{0}$, because every element of $\Mfrak$ can be written as $g f^{s+r} \tp h$ for some $r \in \Z$, $h \in \C[s,\theta]\setminus\set{0}$ and a Laurent polynomial $g \in k[x^{\pm 1}]$.
\end{proof}
In particular, the presence of a linear annihilator \eqref{eq:lin-euler-ann} implies $\Mfrak=\set{0}$ and hence $\chi(\Gm^N\setminus\Vanish(f))=0$ via corollary~\ref{cor:number-of-masters}. Note that we could equally phrase this in terms of the hypersurface $\Vanish(f) \subset \Gm^N$ itself as $\chi(\Vanish(f))=0$, because the Euler characteristics are related through $\chi(\Vanish(f))=-\chi(\Gm^N\setminus\Vanish(f))$ (see section~\ref{sec:Grothendieck}).

When $f=\G$ comes from Feynman graph $G$ (as in the next section), it is not difficult to see that the homogeneity \eqref{eq:quasi-hom} occurs precisely when $G$ has a \emph{tadpole}.\footnote{%
	A tadpole here means a proper subgraph $\gamma \subsetneq G$ which shares only a single vertex with the rest of $G$ and does not depend on masses or external momenta. In this case, $\G_G=\U_{\gamma} \F_{G/\gamma}$ factorizes such that the variables $x_i$ with $i\in\gamma$ only appear in the homogeneous polynomial $\U_{\gamma}$ of degree $\lambda_0\defas\Loops_{\gamma}$. Thus we obtain \eqref{eq:quasi-hom} by setting  $\lambda_e = 1$ if $e\in \gamma$ and $\lambda_e=0$ otherwise.
}
If this is the case, the integrals from proposition~\ref{prop: parametric reps} do not converge for any values of $s$ and $\nu$. In fact, $\Mfrak=\set{0}$ dictates that the only value one can assign to $\Mellin{f^s}(\nu)$ which is consistent with integration by parts relations is zero. This reasoning explains a common practice in Feynman integral calculations, namely that Feynman integrals associated to graphs with tadpoles are declared to vanish.

The purpose of this section is to show that the simple homogeneity condition~\eqref{eq:quasi-hom} is not only sufficient for a vanishing Mellin transform, but it is also necessary:
\begin{prop}\label{prop:chi-0-hom}
	Let $f \in \C[x_1,\ldots,x_N]$ denote a polynomial. Then the hypersurface $\set{f=0}$ inside the torus $\Gm^N$ has vanishing Euler characteristic precisely when there are $\lambda_0,\ldots,\lambda_N \in \Z$, not all zero, such that \eqref{eq:quasi-hom} holds.
\end{prop}
Geometrically, the homogeneity \eqref{eq:quasi-hom} can be interpreted as follows: dividing by the greatest common divisor, we may assume that $\lambda_0,\ldots,\lambda_N$ are relatively prime. Thus we may extend $(\lambda_1,\ldots,\lambda_N)$ to a basis of the lattice $\Z^N$ and hence construct a matrix $A \in \GL{N}{\Z}$ with first row $A_{1i} = \lambda_i$.
In the associated coordinates $y$, defined by
\begin{equation*}
	x_i=\prod_{j=1}^N y_j^{A_{ji}}
	\quad\text{and}\quad
	y_i=\prod_{j=1}^N x_j^{A^{-1}_{ji}}
	\quad\text{where}\quad
	A^{-1}_{ji} \defas \left( A^{-1} \right)_{ji},
\end{equation*}
the polynomial $f$ takes the form $f(x)=y_1^{\lambda_0} g(\bar{y})$ for some Laurent polynomial $g \in \C[\bar{y}^{\pm1}]$ in the remaining variables $\bar{y}=(y_2,\ldots,y_N)$. In particular, the hypersurface $\set{f=0}=\set{g=0}$ can be defined by an equation independent of the coordinate $y_1$.
\begin{cor}\label{cor:chi=0-product}
	Let $f \in \C[x_1,\ldots,x_N]$ denote a polynomial. Then $\Vanish(f) \subset \Gm^N$ has Euler characteristic zero if and only if it is isomorphic to a product of $\Gm$ times a hypersurface $\set{g=0} \subset \Gm^{N-1}$.
\end{cor}
To prove proposition~\ref{prop:chi-0-hom}, we will look at the \emph{Newton polytope} $\NewPol{f}$ of $f$, which is defined as the convex hull of the exponents of monomials that appear in $f$:
\begin{equation}
	\NewPol{\sum_{\alpha \in \Z^N} c_{\alpha} x^{\alpha}}
	\defas \conv \setexp{\alpha \in \Z^N}{c_{\alpha} \neq 0}
	\subset \R^N
	.
	\label{eq:def-Newton}%
\end{equation}
Since every monomial $x^{\alpha}$ is an eigenvector of the operators \eqref{eq:lin-euler-ann}, $P^{\lambda}_s(\theta)\bullet x^{\alpha} = P^{\lambda}_s(\alpha) x^{\alpha}$, it is annihilated by $P^{\lambda}_1$ exactly when $\alpha$ belongs to $F_{\lambda} \defas\set{P^{\lambda}_1(\alpha)=0}$, the hyperplane $F_{\lambda}= \setexp{\alpha}{\alpha_1\lambda_1+\ldots+\alpha_N \lambda_N=\lambda_0}$.
In particular, $0=P^{\lambda}_s \bullet f^s = s f^{s-1} P^{\lambda}_1 \bullet f$ is equivalent to the Newton polytope $\NewPol{f} \subset F_{\lambda}$ being contained in that hyperplane. We can therefore reformulate the equivalent conditions \eqref{eq:quasi-hom} and \eqref{eq:lin-euler-ann} as
\begin{equation}
	\dim \NewPol{f} < N.
	\label{eq:NewPol<N}%
\end{equation}
Such polytopes have zero $N$-dimensional volume, and we call them \emph{degenerate}.

\begin{proof}[Proof of proposition~\ref{prop:chi-0-hom}]
	We proceed by induction over the dimension $N$, and we will assume $f$ to be non-constant (the proposition holds trivially for any constant $f\in \C$).
	In the case $N=1$, the variety $\Vanish(f) \subset \C^{\ast}$ is a finite set and hence its Euler characteristic coincides with its cardinality. Therefore, $\chi(\Vanish(f))=0$ if and only if $f$ has no zero inside the torus. This is only possible if $f$ is proportional to a monomial $x_1^r$; in particular $f$ must be homogeneous and we are done.
	
	Now consider $N>1$ and assume that $\chi(\Gm^N\setminus\Vanish(f))=\chi(\Vanish(f))=0$.
	Recall that \eqref{eq:quasi-hom} is equivalent to degeneracy \eqref{eq:NewPol<N} of $\NewPol{f}$, so we only need to rule out the non-degenerate case. We achieve this by exploiting the hypothesis $\dim \NewPol{f}=N$ to construct a linear annihilator $P^{\lambda}_s$ of $f^s$, which implies \eqref{eq:NewPol<N} in contradiction to the non-degeneracy of $\NewPol{f}$.

	To start, we use lemma~\ref{lem:chi=0-poly-ann} to find a polynomial $0\neq P(\theta,s) \in \C[s,\theta]$ such that $P(\theta,s)\bullet f^s = 0$, and we choose one with minimal total degree in $s$ and $\theta$.
	Then pick an $(N-1)$ dimensional face $\sigma = \NewPol{f} \cap F_{\lambda}$, which we can write as the intersection of $\NewPol{f}$ with a hyperplane $F_{\lambda}$ for some integers $\lambda_0,\ldots,\lambda_N$ such that $\NewPol{f} \subseteq \setexp{\alpha}{P^{\lambda}_1(\alpha) \leq 0}$.
	Under the rescaling \eqref{eq:quasi-hom}, all monomials of $f=\sum_{\alpha} c_{\alpha} x^{\alpha}$ with $\alpha \in \sigma \subset F_{\lambda}$ acquire a factor of $t^{\lambda_0}$, while the remaining monomials with $\alpha \in \NewPol{f} \setminus \sigma$ come with a smaller exponent $\sum_{i=1}^N \alpha_i \lambda_i<\lambda_0$ of $t$:
	\begin{equation}
		f(x_1 t^{\lambda_1},\ldots,x_N t^{\lambda_N})
		= t^{\lambda_0} f_{\sigma}(x) \left( 1+\bigo{t^{-1}} \right),
		\quad\text{where}\quad
		f_{\sigma}(x) \defas \sum_{\alpha \in \Z^N \cap \sigma} c_{\alpha} x^{\alpha}
		\label{eq:leading-power-facet}%
	\end{equation}
	and $\bigo{t^{-1}}$ denotes a rational function in $t^{-1} \C(x)[t^{-1}]$. Note that $P(\theta,s)\bullet f^s(\set{x_i t^{\lambda_i}}) = 0$ is still zero, because the rescaling of $x$ commutes with the Euler operators $\theta_i\bullet h(x_i t^{\lambda_i}) = (\theta_i \bullet h(x_i))|_{x_i\mapsto x_i t^{\lambda_i}}$. Therefore, applying $P(\theta,s)$ to the $s$-th power of the right-hand side of \eqref{eq:leading-power-facet} and dividing by $t^{s\lambda_0}$ yields
	\begin{equation*}
		0 = P(\theta,s) \bullet f_{\sigma}^s(x) \left( 1+\bigo{t^{-1}} \right)^s
		= P(\theta,s)\bullet f_{\sigma}^s + \bigo{t^{-1}}f_{\sigma}^s,
	\end{equation*}
	where $\bigo{t^{-1}}$ on the right-hand side denotes a formal series in $t^{-1} \C(x,s) [[t^{-1}]]$. In particular, the coefficient of $t^0$ must vanish, and we conclude that $P(\theta,s) \bullet f_{\sigma}^s=0$. Label the variables such that $\lambda_N \neq 0$, then we can divide $P(\theta,s)$ by the linear form $P^{\lambda}_s(\theta,s)$ from \eqref{eq:lin-euler-ann}, as a polynomial in $\theta_N$, to obtain a decomposition
	\begin{equation*}
		P(\theta,s) = P(\theta',0,s) + P^{\lambda}_s (\theta,s) \cdot Q(\theta,s)
	\end{equation*}
	for some polynomial $Q(\theta,s)\in\C[\theta,s]$, such that the first summand depends only on $\theta'\defas(\theta_1,\ldots,\theta_{N-1})$ and $s$.
	Since $\NewPol{f_{\sigma}}=\sigma \subset F_{\lambda}$ is contained in the hyperplane $F_{\lambda}=\setexp{\alpha}{P^{\lambda}_1(\alpha)=0}$, we see $P^{\lambda}_s \bullet f^s_{\sigma} = f^{s-1} P^{\lambda}_1 \bullet f_{\sigma} = 0$ and thus $Q(\theta,s)$ drops out in
	\begin{equation*}
		0 = P(\theta,s) \bullet f^s_{\sigma} = P(\theta',0,s) \bullet f^s_{\sigma}
		= P(\theta',0,s) \bullet g^s
		,
	\end{equation*}
	where $g \defas f|_{x_N=1} \in \C[x_1,\ldots,x_{N-1}]$ is a polynomial in less than $N$ variables. If $P(\theta',0,s)$ were non-zero, lemma~\ref{lem:chi=0-poly-ann} would show $\chi(\Gm^{N-1}\setminus \Vanish(g))=0$, such that we could apply our induction hypothesis to $g$ and conclude that $g$ is homogeneous in our generalized sense. We saw that this is equivalent to the degeneracy of $\NewPol{g}$, which contradicts that $\NewPol{g} \isomorph \NewPol{f_{\sigma}}=\sigma$ is of dimension $N-1$.\footnote{%
		Observe that $\NewPol{g}$ is the orthogonal projection of $\NewPol{f_{\sigma}}$ onto the coordinate hyperplane $\set{\alpha_N=0}$. This projection restricts to an isomorphism between $\set{\alpha_N=0}$ and $F_{\lambda}$ (because $\lambda_N \neq 0$), and therefore $\dim \NewPol{g}=\dim \NewPol{f_{\sigma}}$.
	}

	Therefore, $P(\theta',0,s)$ must be zero and we conclude that $P(\theta,s)=P^{\lambda}_s \cdot Q$ has a linear factor $P^{\lambda}_s(\theta,s)$.\footnote{%
		The existence of a linear annihilator could also be deduced from \cite[Th\'{e}or\`{e}me~9.2]{GabberLoeser:FaisceauxEll}.
	}
	Now set $m\defas Q(\theta,s)\bullet f^s$, which is non-zero, because $P(\theta,s)$ was chosen as an annihilator of $f^s$ of minimal degree. We may write this element in the form $m=a\cdot f^{s+r}$ for some $r\in\Z$ and a Laurent polynomial $a \in \C(s)[x^{\pm}]$. After multiplying with with a polynomial in $\C[s]$, we may even assume $0 \neq a \in \C[s,x^{\pm 1}]$ with $P^{\lambda}_s \bullet af^{s+r} = 0$. Applying the Leibniz rule and dividing by $af^{s+r}$, we find
	\begin{equation*}
		0
		= \frac{P^{\lambda}_0 \bullet a}{a} - s\lambda_0 + (s+r) \frac{P^{\lambda}_0 \bullet f}{f}
		.
	\end{equation*}
	Since the degree of $P^{\lambda}_0 \bullet a$ in $s$ is at most the degree (in $s$) of $a$ itself, this first summand on the right has a finite limit as $s\rightarrow 0$. We therefore must have a cancellation of the terms linear in $s$, $P^{\lambda}_0 \bullet f = \lambda_0 f$, which yields the sought-after $P^{\lambda}_s \bullet f^s = 0$.
\end{proof}

\begin{rem}
	In summary, we showed that the following six conditions on a polynomial $f\in\C[x_1,\ldots,x_N]$ are equivalent:
	\begin{inparaenum}[(1)]
		\item homogeneity \eqref{eq:quasi-hom} of $f$,
		\item existence \eqref{eq:lin-euler-ann} of an annihilator of $f^s$ linear in $\theta$,
		\item existence  \eqref{eq:euler-poly-ann} of an annihilator of $f^s$ polynomial in $\theta$,
		\item degeneracy \eqref{eq:NewPol<N} of the Newton polytope $\NewPol{f}$,
		\item vanishing of the Euler characteristic $\chi(\Gm^N\setminus\Vanish(f))=\chi(\Vanish(f))=0$ and
		\item divisibility of $\Vanish(f)$ by $\Gm$ as stated in corollary~\ref{cor:chi=0-product}.
	\end{inparaenum}
\end{rem}
To conclude, let us interpret our observation in the light of the well-known result, due to Kouchnirenko \cite[Th\'{e}or\`{e}me~IV]{Kouchnirenko:NewtonMilnor} and Khovanskii \cite[Theorem~2~in~section~3]{Khovanskii:NewtonGenusComplete}, that relates the Euler characteristic to the volume of the Newton polytope:
\begin{thm}
	For almost all polynomials $f \in \C[x_1,\ldots,x_N]$ with a fixed Newton polytope, we have $\chi(\Gm^N\setminus\Vanish(f))=(-1)^N \cdot N! \cdot \Vol \NewPol{f}$.
	\label{thm:Kouchnirenko}%
\end{thm}
For polynomials $f$ whose non-zero coefficients are sufficiently generic, proposition~\ref{prop:chi-0-hom} follows from theorem~\ref{thm:Kouchnirenko}. Our proof shows that when $\Vol \NewPol{f}=0$, the theorem applies without any constraints on the non-zero coefficients of $f$. In fact, this statement extends to the case when $N!\Vol \NewPol{f}=1$, because it is known that $(-1)^N\cdot \chi(\Gm^N\setminus\Vanish(f))$ is always bounded from above by $N!\Vol \NewPol{f}$, for all $f$, leaving only the possibilities $\set{0,1}$ for the signed Euler characteristic $(-1)^N\cdot \chi(\Gm^N\setminus\Vanish(f))$.

\subsection{Graph polynomials}
\label{sec:graph-polys}%

Our discussion so far applies to all integrals of the type \eqref{eq:FI-momentum}---the defining data is thus the set $\Den=(\Den_1,\ldots,\Den_N)$ of denominators, which is sometimes also called an \emph{integral family} \cite{ManteuffelStuderus:Reduze2}.
The denominators can be arbitrary quadratic forms in the loop momenta; the decomposition \eqref{eq:def:M-Q-J} then defines the associated polynomials $\U$, $\F$ and $\G$ through \eqref{eq:def-U-F}. In particular, the denominators do \emph{not} have to be related to the momentum flow through a (Feynman) graph in any way.

However, we will from now on consider the most common case in applications: integrals associated to a Feynman graph with Feynman propagators. 

\begin{defn}
Given a connected Feynman graph $G$ with $N$ internal edges, $\ExtMoms+1$ external legs and $\Loops$ loops, imposing momentum conservation at each vertex determines the momenta $\EdgeMom_e$ flowing through each edge $e$ in terms of the $\ExtMoms$ external and $\Loops$ loop momenta.\footnote{%
	Momentum conservation implies that the sum $p_1+\cdots+p_{\ExtMoms+1}=0$ of the incoming momenta on all external legs vanishes; hence only $\ExtMoms$ of them are independent.
}
The \emph{Symanzik polynomials} $\U_G$ and $\F_G$ of the graph $G$ are the polynomials $\U$ and $\F$ from \eqref{eq:def-U-F} for the set $\Den=(\Den_1,\ldots,\Den_N)$ of inverse Feynman propagators,\footnote{%
	The infinitesimal imaginary part $\iu\FeynEps$ is irrelevant for our purpose of counting integrals and will be henceforth ignored.
}
\begin{equation}
	\frac{1}{\Den_e} = \frac{1}{-\EdgeMom_e^2+m_e^2-\iu\FeynEps}
	\quad(1\leq e \leq N)
	,
	\label{eq:Feynman-prop}%
\end{equation}
where $m_e$ is the mass associated to the particle propagating along edge $e$. 
	The number $\NoM{G}$ of master integrals of the Feynman graph $G$ is defined in terms of \eqref{eq:def-number-of-masters} as
	\begin{equation}
		\NoM{G}
		\defas \NoM{\G_G}
		\quad\text{where}\quad
		\G_G \defas \U_G+\F_G
		.
		\label{eq:NoM-graph}%
	\end{equation}
\end{defn}
This class of integrals (using only the propagators in the graph) is sometimes referred to as \emph{scalar integrals} and might appear to be insufficient for applications, since in general one needs to augment the inverse propagators by additional denominators, called \emph{irreducible scalar products} (ISPs), in order to be able to express arbitrary numerators of the momentum space integrand in terms of the integrals \eqref{eq:FI-momentum}; see example~\ref{ex:doubletri-ISP}.
Therefore, one might expect that, in order to count all these integrals, one ought to replace $\G_G$ in \eqref{eq:NoM-graph} by the polynomial associated to the full set of denominators, including the ISPs.

However, it is well-known since \cite{Tarasov:ConnectionBetweenFeynmanIntegrals} that all such integrals with ISPs are in fact linear combinations of scalar integrals in higher dimensions $\Dim+2k$, for some $k\in \N$.
Furthermore, those can be written as scalar integrals in the original dimension $\Dim$ by corollary~\ref{cor:dim-shift-closed}. Therefore, \eqref{eq:NoM-graph} is the correct definition to count the number of master integrals of (any integral family determined by) a Feynman graph.

We can therefore invoke the following, well-known combinatorial formulas for the Symanzik polynomials \cite{BognerWeinzierl:GraphPolynomials,Smirnov:AnalyticToolsForFeynmanIntegrals}, which go back at least to \cite{Nak}.
\begin{prop}\label{prop:graph-polynomials}
	The Symanzik polynomials of a graph $G$ can be written as
\begin{equation}
	\U_G = \sum_{T} \prod_{e \notin T} x_e
	\quad\text{and}\quad
	\F_G = \U_G \sum_{e=1}^N x_e m_e^2
	     - \sum_F p_F^2 \prod_{e \notin T} x_e
	,
	\label{eq:graph-polynomials}%
\end{equation}
	where $T$ runs over the spanning trees of $G$ and $F$ enumerates the spanning two-forests of $G$ ($p_F$ denotes the sum of all external momenta flowing into one of the components of $F$).
\end{prop}
\hide{In particular it is well-known that for all connected graphs, $\U_G \neq 0$ and thus $\loopMat$ is invertible and \eqref{eq:def-U-F} makes sense.}
In section~\ref{sec:sunrise}, we will use these formulas to count the sunrise integrals.
\begin{example}\label{ex:bubble-polys}
	The graph polynomials of the bubble graph (figure~\ref{fig:bubble}) are, in general kinematics,
	\begin{equation}
		\U = x_1 + x_2
		\quad\text{and}\quad
		\F = (x_1 + x_2)(x_1 m_1^2 + x_2 m_2^2) - p^2 x_1 x_2
		.
		\label{eq:bubble-polys}%
	\end{equation}
\end{example}

It is important to keep in mind that, even with a fixed graph, the number of master integrals will vary depending on the kinematical configuration---e.g.\ whether a propagator is massive or massless, or whether an external momentum is non-exceptional or sits on a specific value (like zero or various \emph{thresholds}).
We will always explicitly state any assumptions on the kinematics, and hence stick with the simple notation \eqref{eq:NoM-graph}.

\subsection{The Grothendieck ring of varieties}
\label{sec:Grothendieck}

Since we are from now on only interested in the Euler characteristic, we can simplify calculations by abstracting from the concrete variety $\Vanish(\G) \defas \set{\G=0}$ to its class $[\G]$ in the Grothendieck ring $K_0(\text{Var}_{\C})$. This ring is the free Abelian group generated by isomorphism classes $[X]$ of
varieties over $\C$,
modulo the inclusion-exclusion relation $[X]=[X\setminus Z] + [Z]$ for closed subvarieties $Z \subset X$. 
It is a unital ring for the product $[X]\cdot[Y] = [X \times Y]$ with unit $1=[\Aff^0]$ given by the class of the point.
Crucially, the Euler characteristic factors through the Grothendieck ring, since it is compatible with these relations: $\chi(X) = \chi(X\setminus Z) + \chi(Z)$ and $\chi(X\times Y) = \chi(X)\cdot\chi(Y)$.
The class $\Lef = [\Aff^1]$ of the affine line is called \emph{Lefschetz motive} and fulfils $\chi(\Lef)=1$.
For several polynomials $P_1,\ldots,P_n$, we write $\Vanish(P_1,\ldots,P_n) \defas \set{P_1=\cdots=P_n=0}$.

If a variety is described by polynomials that are linear in one of the variables, we can eliminate this variable to reduce the ambient dimension.\footnote{%
	Such \emph{linear reductions} were first investigated by Stembridge in \cite{Stembridge:CountingPoints} and have led, via the $c_2$-invariant \cite{Schnetz:Fq} of Schnetz, to the discovery of graph hypersurfaces that are not of mixed Tate type \cite{BrownSchnetz:K3phi4}.
}
Let us state such a relation explicitly, since our setting is slightly different than usual: For us, the natural ambient space is the torus $\Gm^N$ and not the affine plane $\Aff^N$.

\begin{lem}\label{lem:fibration-torus}
	Let $A,B \in \C[x_1,\ldots,x_{N-1}]$ and consider the linear polynomial $A+x_N B$. Then
\begin{equation}
	[\Gm^N \setminus \Vanish\left(A+x_N B\right)]
	=\Lef \cdot \left[ \Gm^{N-1} \setminus \Vanish\left(A,B\right) \right]
	 -[ \Gm^{N-1} \setminus \Vanish(A) ]
	 -[ \Gm^{N-1} \setminus \Vanish(B) ]
	\label{eq:fibration-torus}%
\end{equation}
holds in the Grothendieck ring. In particular, the Euler characteristic is
\begin{equation}
	\chi \left( \Gm^N \setminus \Vanish\left(A+x_N B\right) \right)
	=- \chi\left( \Gm^{N-1} \setminus \Vanish\left(A\cdot B\right) \right)
	.
	\label{eq:fibration-torus-euler}%
\end{equation}
\end{lem}
\begin{proof}
Consider the hypersurface $\Vanish(\G)\subset \Gm^{N}$, defined by $\G \defas A+x_N B$, under the projection $\pi\colon \Gm^N \longrightarrow \Gm^{N-1}$ that forgets the last coordinate $x_N$.

As long as $AB\neq 0$, the unique solution of $\G=0$ in the fibre is $x_N=-A/B$. If $A=0$, the solution $x_N=0$ is not in $\Gm$ and the fibre is empty, and it is also empty whenever $B=0$. The only exception to this emptiness is over the intersection $A=B=0$, where $x_N$ is arbitrary and the fibre is the full $\Gm$. This fibration proves \eqref{eq:fibration-torus}; equivalently, we can write it via $[\Vanish(A\cdot B)] = [\Vanish(A)] + [\Vanish(B)] - [\Vanish(A,B)]$ and $[\Gm]=\Lef - 1$ as
\begin{equation*}
	[\Gm^N \setminus \Vanish\left(A+x_N B\right)]
	=
	[\Gm]\left(
		[ \Gm^{N-1} \setminus \Vanish(A) ]
		+[\Gm^{N-1} \setminus \Vanish(B) ]
	\right)
	-\Lef \cdot [ \Gm^{N-1} \setminus \Vanish(A\cdot B) ]
	.
\end{equation*}
Applying the Euler characteristic proves \eqref{eq:fibration-torus-euler} due to $\chi(\Aff^1)=1$ and $\chi(\Gm)=0$.
\end{proof}
\begin{cor}
	Set $\widetilde{\U} \defas \restrict{\U}{x_N=1}$ and $\widetilde{\F} \defas \restrict{\F}{x_N=1}$. Then
\begin{align}
	\chi\left( \Gm^N \setminus\!\!\Vanish (\G) \right)
	&= \chi \left(\Gm^{N-1} \setminus\!\!\Vanish( \widetilde{\U},\widetilde{\F}) \right)
	- \chi \left( \Gm^{N-1} \setminus\!\!\Vanish( \widetilde{\U}) \right)
	- \chi \left( \Gm^{N-1} \setminus\!\!\Vanish( \widetilde{\F}) \right)
	\label{eq:chi(G)->U,F}%
	\\
	&= - \chi \left( \Gm^{N-1} \setminus\!\!\Vanish( \widetilde{\U}\cdot \widetilde{\F} ) \right)
	.
	\label{eq:chi(A+xB)->chi(A*B)}%
\end{align}
\end{cor}
\begin{proof}
	Recall that $\U$ and $\F$ are homogeneous of degrees $\Loops$ and $\Loops+1$, respectively (corollary~\ref{cor:U-F-homogeneity}). Since multiplication with $x_N \in \Gm$ is invertible, we can rescale all variables $x_i$ with $i<N$ by $x_N$. This change of coordinates transforms $\G$ into $x_N^{\Loops} (\widetilde{\U}+x_N \widetilde{\F})$. Since $x_N \neq 0$, this shows that
$
[\Gm^N \setminus \Vanish(\G)]
= [\Gm^N \setminus \Vanish(\widetilde{\U} + x_N \widetilde{\F})]
$
such that the claim is just a special case of lemma~\ref{lem:fibration-torus}.
\end{proof}
\begin{example}\label{ex:series-parallel-graph}
	Both graphs consisting of a pair of massless edges,
	\begin{equation}
		G_{\mathrm{series}}=\Graph[0.5]{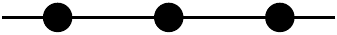}
		\quad\text{and}\quad
		G_{\mathrm{parallel}}=\Graph[0.5]{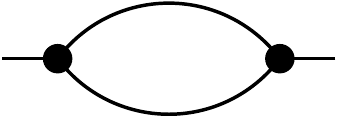},
		\label{eq:series-parallel-graph}%
	\end{equation}
	with the external momentum $p$ such that $p^2 \neq 0$, have a single master integral $\NoM{G}=1$.
\end{example}
\begin{proof}
	According to \eqref{eq:graph-polynomials}, the graph polynomials of the graphs in \eqref{eq:series-parallel-graph} are
	$ \U_{\mathrm{series}} = 1$,
	$ \F_{\mathrm{series}} = - p^2 (x_1 + x_2)$,
	$ \U_{\mathrm{parallel}} = x_1 + x_2$,
	$ \F_{\mathrm{parallel}} = -p^2 x_1 x_2$ such that
	\begin{equation*}
		\G_{\mathrm{series}}
		= 1 - p^2 (x_1 + x_2)
		\quad\text{and}\quad
		\G_{\mathrm{parallel}}
		= x_1 + x_2 -p^2 x_1 x_2
		.
	\end{equation*}
	In both cases, the number of master integrals \eqref{eq:NoM-topo} is
	$
	 	\NoM{G}
		=-\chi(\Gm\setminus \Vanish(\widetilde{\U}\widetilde{\F})) 
		= \chi (\Gm \cap \Vanish(\widetilde{\U} \widetilde{\F}))
		= \chi(\Gm \cap \Vanish(1+x_1))
		= \chi(\set{-1})
		=1
	$ according to \eqref{eq:chi(A+xB)->chi(A*B)}.
\end{proof}
Much more on the Grothendieck ring calculus of graph hypersurfaces $\Vanish(\U)$ can be found, for example, in \cite{AluffiMarcolli:DeletionContraction} and \cite{BrownSchnetz:K3phi4}. These techniques can be used to prove some general statements about the counts of master integrals. Let us give just one example:
\begin{figure}
	\centering\vspace{-6mm}
	$
	G=\Graph[0.6]{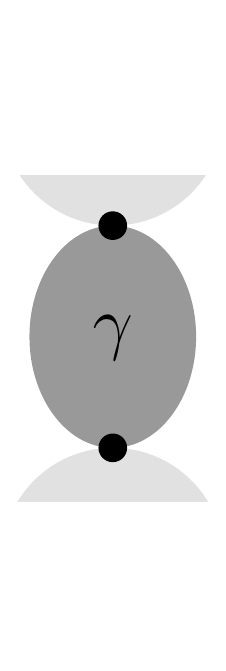}
	\qquad\mapsto\qquad
	G'=\Graph[0.6]{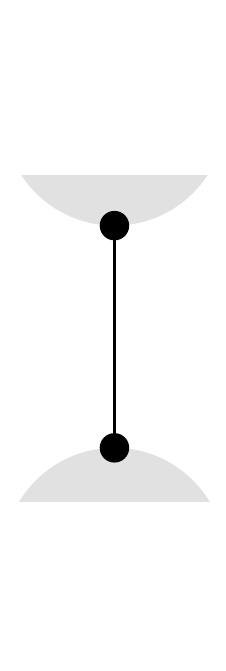}
	$\vspace{-8mm}
	\caption{A 1-scale subgraph $\gamma$ of $G$ is replaced by a single edge in $G'$.}%
	\label{fig:1-scale}%
\end{figure}
\begin{lem}\label{lem:1-scale}
	Let $G$ be a Feynman graph with a subgraph $\gamma$ such that all propagators in $\gamma$ are massless and $\gamma$ has only two vertices which are connected to external legs or edges in $G\setminus\gamma$.\footnote{%
		Such a graph $\gamma$ is called \emph{massless propagator} or \emph{$p$-integral}.
	}
	Write $G'$ for the graph obtained from $G$ by replacing $\gamma$ with a single edge (see figure~\ref{fig:1-scale}), then
	\begin{equation}
		\NoM{G} = \NoM{\gamma} \cdot \NoM{G'}
		.
		\label{eq:1-scale}%
	\end{equation}
\end{lem}
\begin{proof}
	Every spanning tree $T$ of $G$ restricts on $\gamma$ either to a spanning tree or to a spanning two-forest. In the first case, $T\setminus \gamma$ is a spanning tree of $G/\gamma$ (the graph where $\gamma$ is contracted to a single vertex); in the second case, $T\setminus \gamma$ is a spanning tree of $G\setminus \gamma$. Note that the two-forests $T\cap\gamma$ in the second case determine $\F_{\gamma}$ from \eqref{eq:graph-polynomials}, since all propagators in $\gamma$ are massless. Therefore, we find
	$
		\U_G = \U_{\gamma} \cdot \U_{G/\gamma}
		     + \F_{\gamma}' \cdot \U_{G\setminus \gamma}
	$
	where we set $\F_{\gamma}'\defas \restrict{\F_{\gamma}}{p^2=-1}$.
	Going through the same considerations for $\F_G$ shows that
	\begin{equation*}
		\G_G 
		     = \U_{\gamma} \cdot \G_{G/\gamma}
		     + \F_{\gamma}'\cdot \G_{G\setminus \gamma}
		.
	\end{equation*}
	Now label the edges in $\gamma$ as $1,\ldots,N_{\gamma}$ and rescale all Schwinger parameters $x_e$ with $2\leq e \leq N_{\gamma}$ by $x_1$. Due to the homogeneity of $\U_{\gamma}$ and $\F_{\gamma}$ from corollary~\ref{cor:U-F-homogeneity}, we see that $[\Gm^N \setminus \Vanish(\G_G)] = [\Gm^N \setminus \Vanish(A+x_1 B)]$ where $A=\widetilde{U}_{\gamma} \G_{G/\gamma}$ and $B=\widetilde{\F}_{\gamma}' \G_{G\setminus \gamma}$ in terms of $\widetilde{U}_{\gamma} \defas \restrict{\U_{\gamma}}{x_1=1}$ and $\widetilde{\F}_{\gamma}' \defas \restrict{\F_{\gamma}'}{x_1=1}$.
	Applying \eqref{eq:fibration-torus-euler}, we obtain a separation of variables:
	\begin{align*}
		\chi\left(\Gm^N \setminus \Vanish(\G_G)\right)
		&= - \chi\left( \Gm^{N-1} \setminus \Vanish(\widetilde{\U}_{\gamma} \widetilde{\F}_{\gamma}' \G_{G/\gamma} \G_{G\setminus \gamma}) \right)
		\\
		&= - \chi\left( \Gm^{N_{\gamma}-1} \setminus \Vanish(\widetilde{\U}_{\gamma} \widetilde{\F}_{\gamma}') \right)
		\cdot\chi\left( \Gm^{N-N_{\gamma}}\setminus \Vanish(\G_{G/\gamma} \G_{G\setminus \gamma}) \right)
		\\
		&= - \chi\left( \Gm^{N_{\gamma}} \setminus \Vanish(\G_{\gamma}) \right)
		\cdot\chi\left( \Gm^{N-N_{\gamma}+1}\setminus \Vanish(\G_{G'}) \right)
		.
	\end{align*}
	In the last line we used \eqref{eq:chi(A+xB)->chi(A*B)}, upon noting the contraction-deletion formula $\G_{G'} = \G_{G/\gamma} + x_0 \G_{G\setminus \gamma}$ in terms of the additional Schwinger parameter $x_0$ for the (massless) edge that replaces $\gamma$ in $G'$ (this formula is easily checked by considering which spanning trees and forests contain this edge or not).
	Note that for the subgraph $\gamma$, the value of $p^2$ does not matter for $\Vanish(\widetilde{\F}_{\gamma})=\Vanish(\widetilde{\F}_{\gamma}')$, as long as $p^2\neq 0$. Finally, recall \eqref{eq:NoM-topo}.
\end{proof}
Of course, this result is well-known on the function level: If $G$ has a $1$-scale subgraph $\gamma$, then the Feynman integral of $G$ factorizes into the product
\begin{equation}
	\FI_G(\nu)
	= \restrict{\FI_{\gamma}(\nu_{\gamma})}{p^2=-1}
	\cdot
	  \FI_{G'}(\sdc_{\gamma},\nu')
	\label{eq:1-scale-integrals}%
\end{equation}
of the integrals of $\gamma$ and $G'$. Here, we denote by $\nu_{\gamma}$ and $\nu'$ the indices corresponding to the edges in $\gamma$ and outside $\gamma$, respectively, such that $\nu=(\nu_{\gamma},\nu')$. 
Note that the edge replacing $\gamma$ in $G'$ (see figure~\ref{fig:1-scale}) gets the index $\sdc_{\gamma} = \sum_{e \in \gamma} \nu_e - \Loops_{\gamma}\cdot(\Dim/2)$ from \eqref{eq:sdc}, which depends on the indices of $\gamma$ ($\Loops_{\gamma}$ denotes the loop number of $\gamma$).
\begin{figure}
	\centering\vspace{-6mm}
	$
	\Graph[0.6]{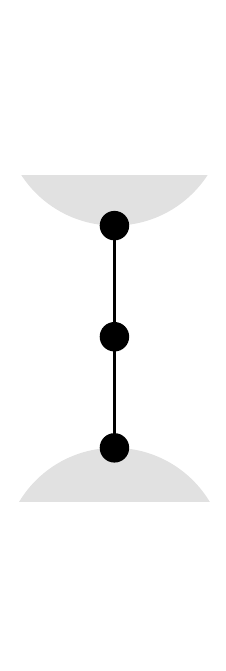}
	\xmapsto{(S)}
	\Graph[0.6]{edge}
	\reflectbox{$\xmapsto{\reflectbox{\scriptsize$(P)$}}$}
	\Graph[0.6]{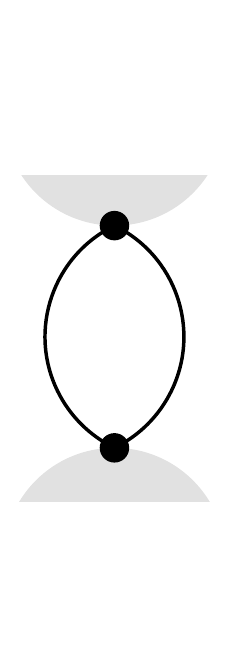}
	$\vspace{-8mm}
	\caption{The series (S) and parallel (P) operations consist of replacing a sequential or parallel pair of massless edges with a single edge.}%
	\label{fig:series-parallel}%
\end{figure}
\begin{cor}\label{cor:series-parallel}
	Let $G$ be a graph with a pair $\set{e,f}$ of massless edges in series or in parallel. Then $\NoM{G} = \NoM{G'}$ where $G'$ is the graph obtained by replacing the pair with a single edge.
	In other words, repeated application of the series-parallel operations from figure~\ref{fig:series-parallel} does not change the number of master integrals.
\end{cor}
\begin{proof}
	Combine lemma~\ref{lem:1-scale} with example~\ref{ex:series-parallel-graph}.
\end{proof}

\section{Tools and examples}
\label{sec:examples}

The Euler characteristic of a singular hypersurface can be computed algorithmically via several methods.\footnote{%
	We will not discuss Kouchnirenko's theorem~\ref{thm:Kouchnirenko} here, because in most examples we found that it does not apply. It seems that coefficients of graph polynomials are often not sufficiently generic.
}
In this section we demonstrate how some of these techniques can be used to compute the number of master integrals in various examples.

We begin with methods based on fibrations. In particular, the Euler characteristic can be computed very easily for the class of \emph{linearly reducible} graphs, see section~\ref{sec:lin-red}.
However, the decomposition of the Euler characteristic of the total space $E$ of a fibration $E\longrightarrow B$ with fibre $F$ into the product
\begin{equation}
	\chi(E) = \chi(B)\cdot \chi(F)
	\label{eq:Euler-fibration-product}%
\end{equation}
is true in general and not restricted to the linear case.
In section~\ref{sec:sunrise} we use a quadratic fibration in order to count the master integrals of all sunrise graphs.

Apart from these geometric approaches, which seem to work very well for Feynman graphs, there are general algorithms for the computation of de Rham cohomology and the Euler characteristic of hypersurfaces. In section~\ref{sec:programs} we discuss some available implementations of these algorithms in computer algebra systems.

In the final section~\ref{sec:comparison}, we comment on the relation of our result to other approaches in the physics literature.

\subsection{Linearly reducible graphs}
\label{sec:lin-red}

\begin{figure}
\centering%
\begin{align*}
	P_1 &= \Graph[0.95]{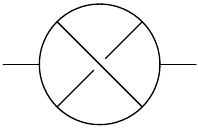}
	& 
	P_2 &= \Graph[0.95]{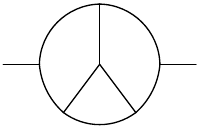}
	&
	P_3 &= \Graph[0.95]{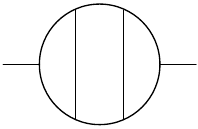}
	&
	P_4 &= \Graph[0.95]{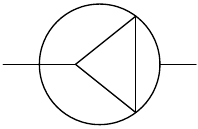}
	\\
	P_5 &= \Graph[0.95]{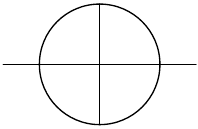}
	&
	P_6 &= \Graph[0.95]{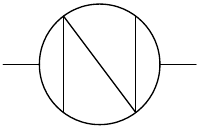}
	&
	P_7 &= \Graph[0.95]{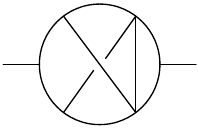}
	& &
\end{align*}\vspace{-3mm}
\begin{align*}
	F_1 &= \Graph[0.85]{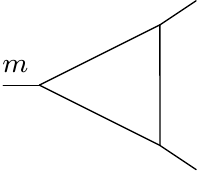}
	&
	F_2 &= \Graph[0.85]{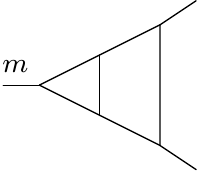}
	&
	F_3 &= \Graph[0.85]{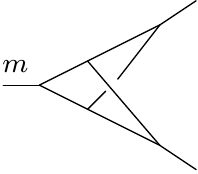}
	&
	F_4 &= \Graph[0.85]{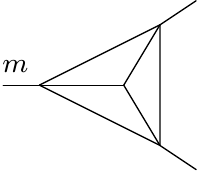}
	&
	F_5 &= \Graph[0.85]{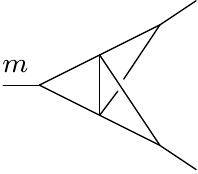}
	\\
	F_6 &= \Graph[0.85]{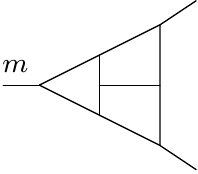}
	&
	F_7 &= \Graph[0.85]{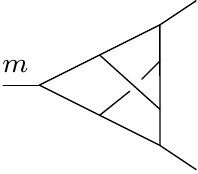}
	&
	F_8 &= \Graph[0.45]{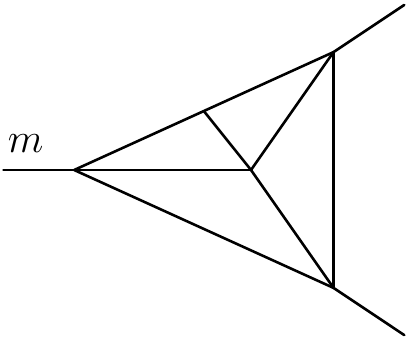}
	&
	F_9 &= \Graph[0.45]{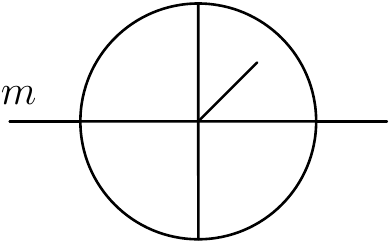}
	&
& 
\end{align*}\vspace{-4mm}
\caption{Some linearly reducible propagators ($P_i$) and form factors ($F_i$). All internal edges are massless, and the form factors have two massless external legs ($p_1^2=p_2^2=0$) and one massive leg $p_3^2\neq 0$, indicated by the label $m$.
}%
	\label{fig:linred-graphs}%
\end{figure}

If the polynomial $\Vanish(f)=a+x_N b$ is linear in a variable $x_N$, we saw in lemma~\ref{lem:fibration-torus} that we can easily eliminate this variable $x_N$ in the computation of the Euler characteristic (or the class in the Grothendieck ring) of the hypersurface $\Vanish(f)$ (or its complement).
Analogous formulas also exist in the case of a variety $\Vanish(f_1,\ldots,f_n)$ of higher codimension, given that all of the defining polynomials $f_i=a_i + x_N b_i$ are linear in $x_N$. Such \emph{linear reductions} have been used heavily in the study of graph hypersurfaces, and are straightforward to implement on a computer \cite{Stembridge:CountingPoints,Schnetz:Fq}.

If such linear reductions can be applied repeatedly until all Schwinger parameters have been eliminated, the graph is called \emph{linearly reducible} \cite{Brown:TwoPoint}. Linear reducibility is particularly common among graphs with massless propagators; we give some examples in figure~\ref{fig:linred-graphs}.

At the end of a full linear reduction, the class of $\Gm^N\setminus \Vanish(\G)$ in the Grothendieck ring is expressed as a polynomial in the Lefschetz motive $\Lef$. To get the Euler characteristic, one then merely needs to substitute $\Lef=[\Aff^1] \mapsto \chi(\Aff^1)=1$.
\begin{example}\label{ex:prop2l-class}
	Consider the $2$-loop propagator $\WS{3}'$ from figure~\ref{fig:wheels}. Linear reductions give
	\begin{equation}
		[\Gm^5 \setminus \Vanish(\G_{\WS{3}'})]
		= -\Lef^4+5\Lef^3-13\Lef^2+21\Lef-15
		\label{eq:prop2l-class}%
	\end{equation}
	in the Grothendieck ring. Substituting $\Lef \mapsto 1$ shows that $\NoM{\WS{3}'}=3$ via \eqref{eq:NoM-topo}.
\end{example}
In this way, we calculated the Euler characteristics for the graphs in figure~\ref{fig:linred-graphs}, using an implementation of the linear reductions similar to the method of Stembridge~\cite{Stembridge:CountingPoints}.
Our results are listed in table~\ref{tab:linred-NoM}.
\begin{table}
\centering%
\begin{tabular}{rcccccccccccccccc}
	$G$ & $P_1$ & $P_2$ & $P_3$ & $P_4$ & $P_5$ & $P_6$ & $P_7$ & $F_1$ & $F_2$ & $F_3$ & $F_4$ & $F_5$ & $F_6$ & $F_7$ & $F_8$ & $F_9$ \\
	\midrule
	$\NoM{G}$ & 16 & 10 & 10 & 10 & 10 & 15 & 22 & 1 & 4 & 5 & 4 & 5 & 20 & 24 & 12 & 13 \\
\end{tabular}%
\caption{The number $\NoM{G}$ of master integrals, computed as the Euler characteristic~\ref{eq:NoM-topo}, for the graphs in figure~\ref{fig:linred-graphs}.
}%
	\label{tab:linred-NoM}%
\end{table}
\hide{
\begin{table}
\centering%
\begin{tabular}{rccccccc}
	$G$ & $\Graph[0.7]{propj}$ & $\Graph[0.7]{propl}$ & $\Graph[0.7]{propm}$ & $\Graph[0.7]{propn}$ & $\Graph[0.7]{propo}$ & $\Graph[0.7]{propp}$ & $\Graph[0.7]{propq}$ \\ 
	\midrule
	$\NoM{G}$ & 16 & 10 & 10 & 10 & 10 & 15 & 22 \\
\end{tabular}
\\
\begin{align*}
	F_1 &= \Graph[0.85]{ffu}
	&
	F_2 &= \Graph[0.85]{ffv}
	&
	F_3 &= \Graph[0.85]{ffw}
	&
	F_4 &= \Graph[0.85]{ffx}
	&
	F_5 &= \Graph[0.85]{ffy}
	\\
	F_6 &= \Graph[0.85]{ffz}
	&
	F_7 &= \Graph[0.85]{fftheta}
	&
	F_8 &= \Graph[0.85]{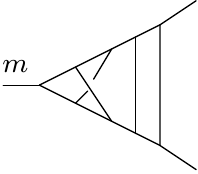}
	&
	F_9 &= \Graph[0.45]{KleinerBoels}
	&
	F_{10} &= \Graph[0.45]{KaspersGraph}
\end{align*}
\\
\begin{tabular}{rcccccccccc}
	Graph $G$ & $F_1$ & $F_2$ & $F_3$ & $F_4$ & $F_5$ & $F_6$ & $F_7$ & $F_8$ & $F_9$ & $F_{10}$ \\
	\midrule
	$\NoM{G}$ & 1 & 4 & 5 & 4 & 5 & 20 & 24 & 116 (Az?) & 13 & 12 \\
\end{tabular}
\caption{The number $\NoM{G}$ of master integrals, computed as the Euler characteristic~\ref{eq:NoM-topo}, for massless propagators (top) and form factors (below). All internal edges are massless, and the form factors have two massless external legs ($p_1^2=p_2^2=0$) and one massive leg $p_3^2\neq 0$, indicated by the label $m$.
}%
	\label{tab:massless-examples}%
\end{table}}

Beyond the computation of such results for individual graphs, it is possible to obtain results for some infinite families of linearly reducible graphs. In particular, efficient computations are possible for graphs of \emph{vertex width} three \cite{Brown:PeriodsFeynmanIntegrals}. For example, the class in the Grothendieck ring of $\Vanish(\U)$ was computed for all wheel graphs in \cite{BrownSchnetz:K3phi4}.
It is possible to adapt such calculations to our setting (where the ambient space is $\Gm^N$ instead of $\Aff^N$). For example, we could prove
\begin{prop}\label{prop:wheels}
	The number of master integrals of the massless propagators obtained by cutting a wheel $\WS{\Loops}$ with $\Loops$ loops, either at a rim or a spoke (see figure~\ref{fig:wheels}), is
\begin{equation}
	\NoM{\WS{\Loops}'}
	= \NoM{\WS{\Loops}''}
	= \frac{\Loops(\Loops-1)}{2}
	.
	\label{eq:wheels}%
\end{equation}
\end{prop}
\begin{figure}
	\begin{align*}
	\WS{3}' &= \Graph[0.5]{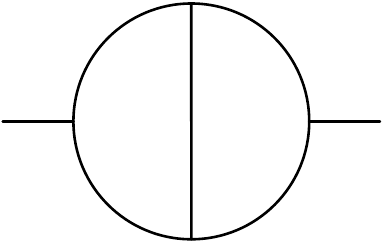}
	&
	\WS{4}' &= \Graph[0.5]{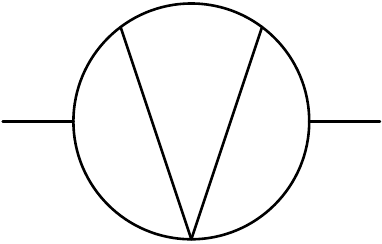}
	&
	\WS{5}' &= \Graph[0.5]{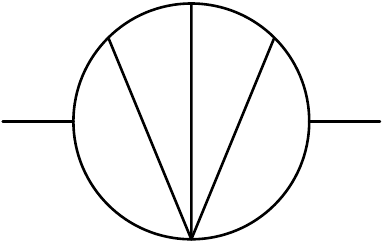}
	&
	\WS{L}' &= \Graph[0.5]{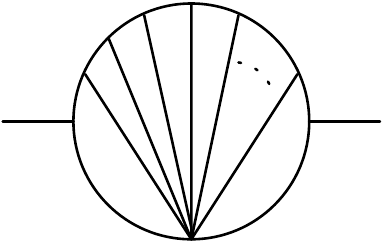}
	\\
	& &
	\WS{4}''&= \reflectbox{$\Graph[1]{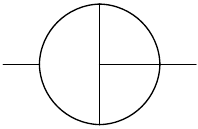}$}
	&
	\WS{5}''&= \Graph[1]{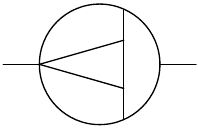}
	&
	\WS{\Loops}''&= \Graph[0.5]{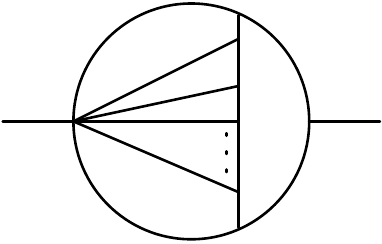}
	\end{align*}\vspace{-4mm}
	\caption{The propagator graphs $\WS{\Loops}'$ ($\WS{\Loops}''$) with $\Loops-1$ loops are obtained from cutting a rim (spoke) of the wheel $\WS{\Loops}$ with $\Loops$ loops.}%
	\label{fig:wheels}%
\end{figure}
The proof of this and related results will be presented elsewhere.

\subsection{Sunrise graphs}
\label{sec:sunrise}

In \cite{KalmykovKniehl:CountingMastersSunrise}, the number of master integrals was computed for all sunrise integrals; based on a Mellin-Barnes representation and the \emph{differential reduction} \cite{KalmykovKniehl:MellinBarnes,BytevKalmykovKniehl:DiffRedOneVar} of an explicit solution in terms of Lauricella hypergeometric functions.
To our knowledge, this has hitherto been the only non-trivial%
\footnote{%
	We consider families that arise simply by duplication of massless propagators, like those shown in \cite[figure~2]{BytevKalmykovKniehl:DiffRedOneVar}, as trivial (due to corollary~\ref{cor:series-parallel}).
}
infinite family of Feynman integrals with explicitly known master integral counts.
Their first result can be phrased as\footnote{%
	Beware that the number $2^{\Loops+1}-\Loops-2$ given in \cite[equation~(4.5)]{KalmykovKniehl:CountingMastersSunrise} counts only \emph{irreducible} master integrals, which means that it discards the $\Loops+1$ integrals associated to the subtopologies obtained by contracting any of the edges. Our conventions, however, do take these integrals into account.
}
\begin{figure}
	\centering
	$\Sun{1} = \Graph[0.6]{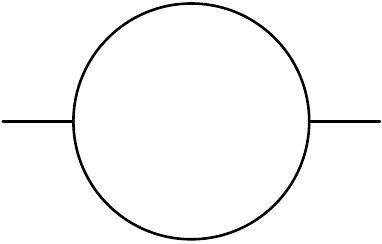}$
	\quad
	$\Sun{2} = \Graph[0.6]{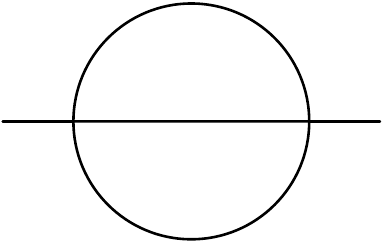}$
	\quad
	$\Sun{3} = \Graph[0.6]{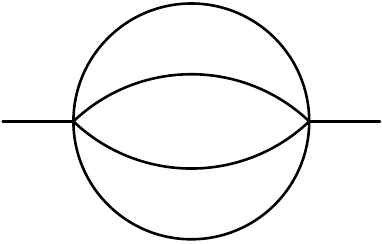}$
	\quad
	$\Sun{\Loops} = \Graph[0.6]{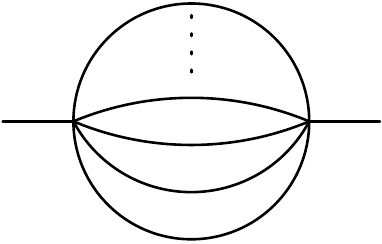}$
	\caption{The sunrise graphs $\Sun{\Loops}$ with $\Loops$ loops.}%
	\label{fig:sunrise}%
\end{figure}
\begin{prop}\label{prop:sunrise-massive}
	The $\Loops$-loop sunrise graph $\Sun{\Loops}$ from figure~\ref{fig:sunrise} with $\Loops+1$ non-zero masses (and non-exceptional external momentum) has $\NoM{\Sun{\Loops}} = 2^{\Loops+1}-1$ master integrals.
\end{prop}
We will now demonstrate that this result can be obtained from a straightforward computation of the Euler characteristic, according to corollary~\ref{cor:number-of-masters}.
\begin{proof}
	The graph polynomials \eqref{eq:graph-polynomials} for the sunrise graph are
\begin{equation}
	\U = \sum_{i=1}^{\Loops+1} \prod_{j \neq i} x_j
	= \left( \prod_{i=1}^{\Loops+1} x_i \right) \left( \sum_{i=1}^{\Loops+1} \frac{1}{x_i} \right)
	\quad\text{and}\quad
	\F = (-p)^2 \prod_{i=1}^{\Loops+1} x_i + \U \sum_{i=1}^{\Loops+1} x_i m_i^2
	.
	\label{eq:sunrise-polys}%
\end{equation}
We note that for the first term in \eqref{eq:chi(G)->U,F}, we find that $\U=\F=0$ imply $\prod_{i} x_i = 0$, which has no solutions in the torus---hence, this term contributes $\chi(\Gm^{\Loops}) = 0$. We thus obtain
\begin{equation}
	(-1)^{\Loops} \NoM{\Sun{\Loops}}
	= \chi \left( \Gm^{\Loops} \setminus \Vanish\left( 1+\sum_{i=1}^{\Loops} x_i^{-1} \right) \right)
	+ \chi\left( \Gm^{\Loops} \setminus X^{\Loops}_{p^2} \right),
	\label{eq:sunrise-chi-decomposition}%
\end{equation}
where we introduced the notation
\begin{equation}
	X^{\Loops}_{p^2}
	\defas
	\Vanish\left(-p^2+\Big[m_{\Loops+1}^2+\sum_{i=1}^{\Loops} m_i^2 x_i\Big] \cdot \Big[1+\sum_{i=1}^{\Loops} x_i^{-1} \Big] \right)
	\subset \Gm^{\Loops}
	.
	\label{eq:sunrise-surface}%
\end{equation}
The first Euler characteristic in \eqref{eq:sunrise-chi-decomposition} is readily evaluated to $(-1)^{\Loops}$ by applying \eqref{eq:fibration-torus-euler} repeatedly (being on the torus, we may replace $x_i^{-1}$ by $x_i$), so we conclude that
\begin{equation}
	\NoM{\Sun{\Loops}}
	= 1
	+ (-1)^{\Loops} \cdot \chi\left( \Gm^{\Loops} \setminus X^{\Loops}_{p^2} \right)
	.
	\label{eq:sunrise-chi-decomposition-2}%
\end{equation}
Now let us consider the projection $\pi\colon \Gm^{\Loops} \longrightarrow \Gm^{\Loops-1}$ that forgets $x_{\Loops}$.
Set $y \defas m_{\Loops+1}^2+\sum_{i<\Loops} m_i^2 x_i$ and $z \defas 1+\sum_{i<\Loops} x_i^{-1}$, such that $X_{p^2}^{\Loops} = \set{x_{\Loops} p^2 = (1+z x_{\Loops})(y+ m_{\Loops}^2 x_{\Loops})} \subset \Gm^{\Loops}$. We note that the discriminant $D$ of this quadric in $x_{\Loops}$ factorizes into
\begin{equation}
	D=(m_{\Loops}^2 - p^2 + yz)^2-4m_{\Loops}^2 yz
	= \left( yz - [p+m_{\Loops}]^2 \right) \cdot \left( yz - [p-m_{\Loops}]^2 \right)
	,
	\label{eq:sunrise-discriminant}%
\end{equation}
such that $\Gm^{\Loops-1} \subset \Vanish(D) = X^{\Loops-1}_{(p+m_{\Loops})^2} \cupdot X^{\Loops-1}_{(p-m_{\Loops})^2}$ is the disjoint union of two hypersurfaces.\footnote{%
	We assume $p^2 \neq 0$ and $m_{\Loops}^2 \neq 0$, which guarantees that $(p+m_{\Loops})^2 \neq (p-m_{\Loops})^2$.
}
Since the factors are related to the $(\Loops-1)$-loop sunrise by \eqref{eq:sunrise-chi-decomposition-2}, we find
\begin{equation}
	\chi(\Vanish(D))
	=2\cdot (-1)^{\Loops}\cdot (\NoM{\Sun{\Loops-1}}-1)
	.
	\label{eq:sun-discriminant-as-NoM}%
\end{equation}
Over a point $x'\in X^{\Loops-1}_{p\pm m_{\Loops}} \subset \Vanish(D)$ in the discriminant, the fibre of $\pi^{-1}(x')$ has precisely one solution $(x',x_{\Loops})$ in $X^L_{p^2}$, determined by $x_{\Loops}=-y/[m_{\Loops} (m_{\Loops}\pm p)]$:
\begin{equation}
	\left[(\Gm^{\Loops} \setminus X^{\Loops}_{p^2}) \cap \Vanish(D)\right]
	= ([\Gm]-1)\cdot [\Vanish(D)]
	.
	\label{eq:sun-over-discriminant}%
\end{equation}
If $D(x')\neq 0$ is non-zero and also $yz \neq 0$, then the fibre $\pi^{-1}(x')$ has precisely two distinct solutions $x_L$ in the quadric $X_{p^2}^{\Loops}$. Hence, $\chi(\pi^{-1}(x'))=\chi(\Gm)-2=-2$ and thus
\begin{equation}
	\chi\left( (\Gm^{\Loops}\setminus X_{p^2}^{\Loops}) \setminus \!\!\Vanish(yzD) \right)
	= -2 \chi\left( \Gm^{\Loops-1} \setminus \!\!\Vanish(yzD) \right)
	= 2 \chi(\Vanish(D)) + 2\chi(\Vanish(yz))
	,
	\label{eq:sun-outside-discriminant}%
\end{equation}
where used that $\Vanish(D) \cap \Vanish(yz)=\emptyset$ for non-exceptional values of $p^2$, such that $(p \pm m_{\Loops})^2 \neq 0$ in \eqref{eq:sunrise-discriminant}.
The reason that we need to exclude the case when $yz=0$ in \eqref{eq:sun-outside-discriminant} is that for $y=0$, one of the solutions of $X_{p^2}^{\Loops} = \set{x_{\Loops}p^2=(1+z x_{\Loops})m_{\Loops}^2 x_{\Loops}}$ is $x_{\Loops}=0 \notin \Gm$; whereas for $z=0$ the equation for $X_{p^2}^{\Loops}=\set{x_{\Loops}p^2=y+m_{\Loops}^2 x_{\Loops}}$ becomes linear. 
In both cases, there is only one solution in the fibre, and there is none if both $y=z=0$ vanish (we assume $p^2 \neq m_{\Loops}^2$):
\begin{equation}
	[(\Gm^{\Loops} \setminus X_{p^2}^{\Loops}) \cap \Vanish(yz)]
	= ([\Gm]-2) \cdot [\Vanish(yz)] + [\Vanish(y)] + [\Vanish(z)]
	.
	\label{eq:sunrise-exceptional}%
\end{equation}
We can now combine \eqref{eq:sun-over-discriminant}--\eqref{eq:sunrise-exceptional} via $[Y]=[Y\cap\Vanish(D)]+[Y\cap\Vanish(yz)]+[Y\setminus \Vanish(D\cdot yz)]$ for $Y=\Gm^{\Loops}\setminus X_{p^2}^{\Loops}$ into the reduction formula
\begin{equation}
	\chi( \Gm^{\Loops} \setminus X^{\Loops}_{p^2} )
	= \chi(\Vanish(D))
	+ \chi(\Vanish(y))
	+ \chi(\Vanish(z))
	= 2\cdot(-1)^{\Loops} \cdot \NoM{\Sun{\Loops-1}}
	.
	\label{eq:sun-chi-recursion}%
\end{equation}
Here, we inserted \eqref{eq:sun-discriminant-as-NoM} and used $\chi(\Vanish(y))=\chi(\Vanish(z))=-\chi(\Gm^{\Loops-1}\setminus\Vanish(z))=(-1)^{\Loops}$, which follows from repeated application of \eqref{eq:fibration-torus-euler}---just as above, when we computed the first term in \eqref{eq:sunrise-chi-decomposition}.
According to \eqref{eq:sunrise-chi-decomposition-2}, we can write the reduction as the recursion
\begin{equation*}
	\NoM{\Sun{\Loops}}
	= 2\NoM{\Sun{\Loops-1}} + 1,
\end{equation*}
which is obviously solved by the claimed $\NoM{\Sun{\Loops}}=2^{\Loops+1}-1$. It merely remains to verify the base case $\Loops=1$, and indeed, $\NoM{\Sun{1}}=1+\chi(X_{p^2}^1)=3$ follows easily from \eqref{eq:sunrise-chi-decomposition-2} since $X_{p^2}^1 = \set{x_1 p^2 = (1+x_1)(m_2^2+m_1^2 x_1)}\subset \Gm$ consists of precisely two points.
\end{proof}
It should be clear that our calculation can be adapted to the situation when some masses are zero. Let us demonstrate how to obtain another result of \cite{KalmykovKniehl:CountingMastersSunrise}:\footnote{%
	The additional term $-\delta_{0,\Loops-R}$ in \cite[equation~(4.13)]{KalmykovKniehl:CountingMastersSunrise} subtracts a \emph{reducible} integral that can be attributed to a subtopology. Our counting, however, accounts for all master integrals.
}
\begin{prop}\label{prop:sunrise-mixed}
	The $\Loops$-loop sunrise graph with $R \leq \Loops$ non-zero masses, $\Loops+1-R$ vanishing masses and non-exceptional external momentum, has $\NoM{\Sun{L}} = 2^{R}$ master integrals.
\end{prop}
\begin{proof}
	By corollary~\ref{cor:series-parallel}, we may replace all massless edges by a single (massless) edge without changing the number of master integrals; hence we can assume $\Loops=R \geq 1$ (the totally massless case $R=0$ reduces to the trivial case of a single edge).
	Label the edges such that the massless edge is $m_{\Loops+1}=0$.

	We can apply the exact same recursion as in the proof of proposition~\ref{prop:sunrise-massive}; the only difference to \eqref{eq:sun-chi-recursion} is that now,
	$\chi(\Vanish(y)) = 0$ vanishes because $y=\sum_{i<\Loops} m_i^2 x_i$ has become homogeneous in $x$ such that $[\Vanish(y)]=[\Gm]\cdot[\Vanish(y)\cap\set{x_1=1}]$. Therefore, \eqref{eq:sun-chi-recursion} takes the form
$
		\chi(\Gm^{\Loops}\setminus X_{p^2}^{\Loops})
		= (-1)^{\Loops}\cdot \left( 2\NoM{\Sun{\Loops-1}}-1 \right)
$
	and yields, via \eqref{eq:sunrise-chi-decomposition-2}, the recursion
	\begin{equation*}
		\NoM{\Sun{\Loops}} = 2 \NoM{\Sun{\Loops-1}}.
	\end{equation*}
	We are done after verifying the base case: Indeed, $\NoM{\Sun{1}}=1+\chi(X_{p^2}^1) = 2$ from \eqref{eq:sunrise-chi-decomposition-2} is clear since $X_{p^2}^1 = \set{p^2 x_1 = (1+x_1) m_1^2 x_1}$ is the single point $x_1=p^2/m_1^2-1$ in $\Gm$.
\end{proof}

\subsection{General algorithms}
\label{sec:programs}

The computer algebra system {\Macaulay} \cite{Macaulay2} provides the function
\href{http://www2.macaulay2.com/Macaulay2/doc/Macaulay2-1.9.2/share/doc/Macaulay2/CharacteristicClasses/html/\_\_\_Euler.html}{\texttt{Euler}} in the package \href{http://www2.macaulay2.com/Macaulay2/doc/Macaulay2-1.9.2/share/doc/Macaulay2/CharacteristicClasses/html/index.html}{\texttt{CharacteristicClasses}}. It implements the algorithm of \cite{Helmer:EulerCSMandSegre} for the computation of the Euler characteristic.
This program requires projective varieties as input, so we need to homogenize $\G$ to $\tilde{\G} = x_0 \U + \F$, and can then use one of
\begin{align}
	[\Aff^{N} \setminus \Vanish(x_1 \cdots x_N  \G)]
	&=
	[\PP^N \setminus \Vanish (x_1\cdots x_N \tilde{\G})]
	- [\PP^{N-1} \setminus \Vanish(x_1\cdots x_N \F)]
	\\
	&= 
	[\PP^N \setminus \Vanish (x_0 x_1\cdots x_N \tilde{\G})]
	\label{eq:homogenize}%
\end{align}
to express the sought after number of master integrals as the Euler characteristic of a projective hypersurface complement. We found that this algorithm performs well for small numbers of variables (edges): The examples in table~\ref{tab:massive-examples} require not more than a couple of minutes of runtime. For more variables, however, the computations tend to rapidly become much more time consuming and often impracticable.
Apart from the results in table~\ref{tab:massive-examples}, we also verified proposition~\ref{prop:sunrise-massive} for the sunrise graphs $\Sun{\Loops}$ using \href{http://www2.macaulay2.com/Macaulay2/doc/Macaulay2-1.9.2/share/doc/Macaulay2/CharacteristicClasses/html/\_\_\_Euler.html}{\texttt{Euler}} for up to six loops.

\begin{table}
\centering%
\begin{tabular}{rcccccccc}
	Graph $G$ & $\Graph[0.85]{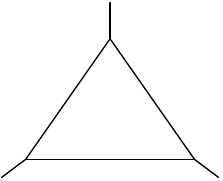}$ & $\Graph[0.85]{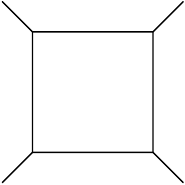}$ & $\Graph[0.6]{Wheel2}$ & $\Graph[0.85]{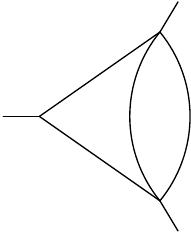}$ & $\Graph[0.85]{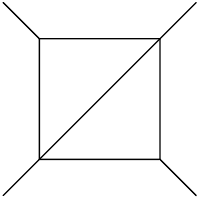}$ \\
\midrule
$\NoM{G}$ massless & 4 & 11 & 3 & 4 & 20\\
$\NoM{G}$ massive & 7 & 15 & 30 & 19 & 55\\
\end{tabular}%
\caption{Counts of master integrals according to \eqref{eq:NoM-topo} computed with {\Macaulay}'s \href{http://www2.macaulay2.com/Macaulay2/doc/Macaulay2-1.9.2/share/doc/Macaulay2/CharacteristicClasses/html/\_\_\_Euler.html}{\texttt{Euler}} for some graphs for massless and massive internal propagators (as no symmetries are regarded, all masses can be assumed to be different from each other). All external momenta are assumed to be non-degenerate (non-zero and not on any internal mass shell) in both cases.}%
	\label{tab:massive-examples}%
\end{table}
\begin{example}\label{ex:bubble-Euler}
	Consider the one-loop sunrise graph $\Sun{1}$ with $m_1^2=m_2^2=-p^2=1$, which is a non-degenerate kinematic configuration. According to example~\ref{ex:bubble-polys}, its Lee-Pomeransky polynomial is $\G = (x_1+x_2)(x_1+x_2+1)+x_1 x_2$. The {\Macaulay} script
	\begin{MacaulayInput}
load "CharacteristicClasses.m2"
R=QQ[x0,x1,x2]
I=ideal(x0*x1*x2*((x1+x2)*x0+(x1+x2)^2+x1*x2))
Euler(I)
	\end{MacaulayInput}
	computes the output $0$ for $\chi(\Vanish(x_0 x_1 x_2 \tilde{\G}) \cap \PP^2)$. Using $\chi(\PP^2)=3$ and \eqref{eq:homogenize}, we conclude $\NoM{\Sun{1}} = 3-0=3$ in agreement with proposition~\ref{prop:sunrise-massive}.
\end{example}
Recall that the number of master integrals depends on the kinematical configuration; in table~\ref{tab:massive-examples} we give the results both for massless and for massive internal propagators. In particular, note how the massless 2-loop propagator $\WS{3}'$ from example~\ref{ex:prop2l-class} with only $\NoM{\WS{3}'}=3$ master integrals grows to carry $\NoM{\WS{3}'}=30$ master integrals in the fully massive case.

Furthermore, {\Macaulay} also provides an implementation (the command \href{http://www2.macaulay2.com/Macaulay2/doc/Macaulay2-1.10/share/doc/Macaulay2/Dmodules/html/\_de\_\_Rham.html}{\texttt{deRham}}) of the algorithm \cite{OakuTakayama:dR-via-Dmod} of Oaku and Takayama for the computation of the individual de Rham cohomology groups. This uses $D$-modules and Gr\"{o}bner bases and tends to demand more resources than the method discussed above.
\begin{example}\label{ex:bubble-deRham}
	Consider again the massive one-loop sunrise from example~\ref{ex:bubble-polys}. The program
	\begin{MacaulayInput}
load "Dmodules.m2"
R=QQ[x1,x2]
f=x1*x2*(x1+x2+(x1+x2)^2+x1*x2)
deRham f
	\end{MacaulayInput}
	computes the following cohomology groups of $X=\Gm^2 \setminus \Vanish(\G)$:
	$H^0(X) \isomorph \Q$, $H^1(X) \isomorph \Q^3$ and $H^2(X)\isomorph\Q^5$. Hence, $\NoM{\Sun{1}}=\chi(X)=5-3+1=3$ as in example~\ref{ex:bubble-Euler}.
\end{example}
The same functionality is provided by {\Singular}'s \texttt{deRham.lib} library via the command
\href{https://www.singular.uni-kl.de/Manual/4-0-3/sing\_2398.htm}{\texttt{deRhamCohomology}}. The {\Singular} analogue of example~\ref{ex:bubble-deRham} is
\begin{MacaulayInput}
LIB "deRham.lib";
ring R = 0,(x1,x2),dp;
list L = (x1*x2*(x1+x2+(x1+x2)^2+x1*x2));
deRhamCohomology(L);
\end{MacaulayInput}

\subsection{Comparison to other approaches}
\label{sec:comparison}

We successfully reproduced all of our results above (the wheels $\WS{\Loops}'$ with $\Loops\leq 6$, the sunrises $\Sun{\Loops}$ with $\Loops \leq 4$ loops and the graphs from figure~\ref{fig:linred-graphs} and table~\ref{tab:massive-examples}) with the program {\Azurite} \cite{GeorgoudisLarsenZhang:Azurite}, which provides an implementation of Laporta's approach \cite{Laporta:HighPrecision}. While it employs novel techniques to boost performance, in the end it solves linear systems of equations between integrals obtained from annihilators of the integrand of Baikov's representation \eqref{eq:Baikov-FI} in order to count the number of master integrals.

The observed agreement with our results is to be expected, since the identification of integral relations with parametric annihilators that we elaborated on in section~\ref{sec:annihilators} works equally for the Baikov representation, which can also be interpreted as a Mellin transform.
Note, however, that we must use the options
\texttt{Symmetry -> False} and \texttt{GlobalSymmetry -> False}
for {\Azurite} in order to switch off the identification of integrals that differ by a permutation of the edges. The reason being that, in our approach, all edges $e$ carry their own index $\nu_e$ and no relation between these indices for different edges is assumed.

Unfortunately, due to the way {\Azurite} treats subsectors, this can occasionally lead to an apparent mismatch. However, this is rather a technical nuisance than an actual disagreement.
\begin{figure}
	\centering%
	$G=\!\!\!\!\Graph[0.75]{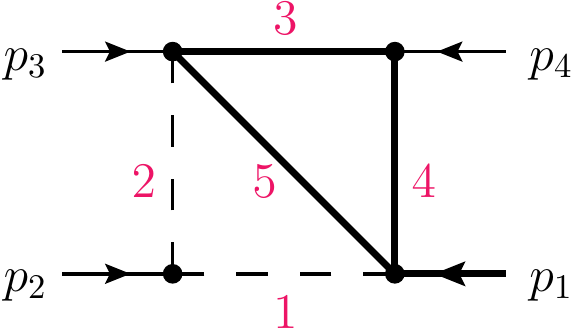}$
	\quad
	$G/\set{1,2}=\Graph[0.6]{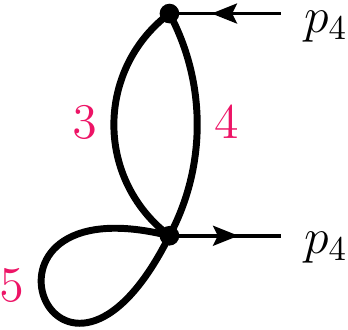}$
	$\mapsto \begin{cases}
		\Graph[0.4]{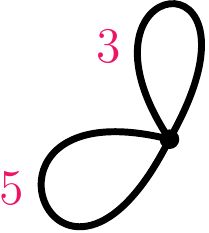} & 3 \\
		\Graph[0.4]{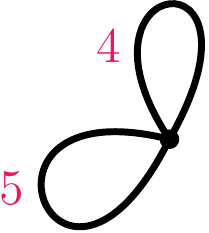} & 4 \\
	\end{cases}$
	\quad
	$G'=\Graph[0.6]{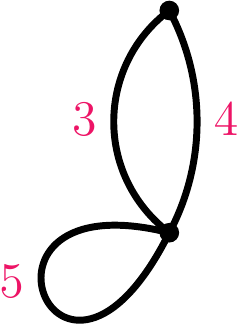}$
	\caption{A graph $G$ with a massive loop (edges 3, 4 and 5), two massless propagators $\set{1,2}$ and four external legs ($p_1^2\neq 0$ massive and $p_2^2=p_3^2=p_4^2=0$ massless). Since the only momentum $p_4$ running through the graph after contracting $1$ and $2$ is null ($p_4^2=0$), the associated integral is the same as for $G'$.}%
	\label{fig:lorenzo}%
\end{figure}
\begin{example}\label{ex:lorenzo}
	For the graph $G$ in figure~\ref{fig:lorenzo}, the Euler characteristic gives $\NoM{G} = 15$, whereas both {\Reduze} \cite{ManteuffelStuderus:Reduze2} and {\Azurite} produce $16$ master integrals. The problem arises from the \emph{subsector} where the edges $1$ and $2$ are contracted: As shown in figure~\ref{fig:lorenzo}, it does have a remaining external momentum $p_4$, such that the momenta running through edges $3$ and $4$ are different---however, since $p_4^2=0$, the graph polynomials (and hence the Feynman integrals) are identical to those of the vacuum graph $G'$ in figure~\ref{fig:lorenzo}. Since edges $3$ and $4$ in $G'$ have the same mass, they can be combined and thus $G'$ clearly has only a single master integral: the product of two tadpoles.

	But {\Azurite} and {\Reduze} instead consider the subsectors of $G/\set{1,2}$ obtained by contracting a further edge (3 or 4), and obtain the two tadpoles (see figure~\ref{fig:lorenzo}) consisting only of edges $\set{4,5}$ and $\set{3,5}$, respectively, as master integrals.
	Of course, these would be recognized as identical if symmetries were allowed; but the point is that even without using symmetries, there is only a single master integral for $G'$ (as computed by the Euler characteristic).
\end{example}
Our results are also consistent with the conclusions obtained within the differential reduction approach \cite{KalmykovKniehl:MellinBarnes}; indeed, we demonstrated in section~\ref{sec:sunrise} how the master integral counts of \cite{KalmykovKniehl:CountingMastersSunrise} for the sunrise graphs emerge directly from the computation of the Euler characteristic.
Let us point out again, however, that some care is required for these comparisons, since those works refer to \emph{irreducible} master integrals, which excludes integrals that can be expressed with gamma functions.
In particular, the fact that the two loop sunrise $\Sun{2}$ with one massless line has $\NoM{\Sun{2}}=4$ master integrals (see proposition~\ref{prop:sunrise-mixed}) is consistent with \cite{KalmykovKniehl:CountingMasters}. We are counting all master integrals and are not concerned here with the much more subtle question addressed by the observation that two of these integrals may be expressed with gamma functions.

Finally, let us note that also the work of Lee and Pomeransky \cite{LeePom} addresses a different problem: Considering only integer indices $\nu \in \Z^N$, how many \emph{top-level} master integrals are there for a graph $G$? This means that integrals obtained from subsectors (graphs $G/e$ with at least one edge $e$ contracted) are discarded. Geometrically, the number of the remaining master integrals is identified with the dimension of the cohomology group $H^N(\C^N\setminus \Vanish(\G))$.\footnote{%
	Actually, they initially refer to a different, relative cohomology group; but in the description of their implementation in {\Mint} they seem to work with this total cohomology group.
}
In most cases, the program {\Mint} computes this number correctly, which then agrees with the other mentioned methods.\footnote{%
	Occasional mismatches are known, like for the graph $F_9$ from figure~\ref{fig:linred-graphs} that was addressed in \cite[section~4.1]{BoelKniehlYang:Master4Sudakov}. These discrepancies are due to an error in the implementation of {\Mint} that misses contributions from critical points at infinity (we thank Yang Zhang for bringing this to our attention).
}
We refer to \cite[section~4]{BoelKniehlYang:Master4Sudakov} and \cite[section~6]{KalmykovKniehl:CountingMastersSunrise} for detailed discussions of this comparison.
Note that the dimension (and a basis) of the top cohomology group can also be computed with the command
\href{https://www.singular.uni-kl.de/Manual/4-0-3/sing\_625.htm}{\texttt{deRhamCohom}} from the {\Singular} library \texttt{dmodapp.lib}.

A direct comparison of {\Mint} with our results is not possible, since the concept of \emph{top-level} integrals does not make sense in our setting of arbitrary, non-integer indices $\nu$. Here, there is no relation at all between integrals of a quotient graph $G/e$ and integrals of $G$ (the former do not depend on $\nu_e$ at all; the latter do).

\begin{rem}\label{rem:top-level}
	Using the inclusion-exclusion principle, one might be tempted to define
\begin{equation}
	\NoMtop{G}
	\defas \sum_{\gamma \subseteq G} (-1)^{\abs{\gamma}} \NoM{G/\gamma}
	= \NoM{G} - \sum_e \NoM{G/e} + \sum_{e<f} \NoM{G/\set{e,f}} - \cdots
	\label{eq:NoMtop-def}%
\end{equation}
	as the number of \emph{top-level} master integrals, since it subtracts from all master integrals $\NoM{G}$ the integrals associated to subsectors (and corrects for double counting). Note that if $\gamma$ contains a loop, the corresponding term in the sum should be set to zero (we only consider contractions with the same loop number as $G$). The reverse relation,
	\begin{equation}
		\NoM{G}
		= \sum_{\gamma \subseteq G} \NoMtop{G/\gamma}
		,
		\label{eq:NoM-from-top}%
	\end{equation}
	is consistent with the intuition that the total set of master integrals is obtained as the union of all top-level masters. 
	By
	$\G_{G/\gamma}=\restrict{\G}{x_e=0\forall e \in\gamma}$,
	we find that
	\begin{equation}
		\NoMtop{G}
		= (-1)^{N} \chi\left( \Aff^{N} \setminus \Vanish(\G) \right)
		\label{eq:NoM-top-chi}%
	\end{equation}
	is the Euler characteristic of the hypersurface complement inside affine space (as compared to the torus $\Gm^N$ as ambient space).
	However, this number can take negative values and may thus not be interpreted as a dimension. For example, $\NoMtop{G/\set{1,2}}=-1$ for the graph from figure~\ref{fig:lorenzo}, which is required for consistency of \eqref{eq:NoM-from-top} to get
	\begin{equation*}
		1
		=\NoM{G/\set{1,2}}
		=\NoMtop{G/\set{1,2}} 
		+\NoMtop{G/\set{1,2,3}} 
		+\NoMtop{G/\set{1,2,4}} 
		= -1 + 1 + 1
		.
	\end{equation*}
\end{rem}

\section{Outlook}
\label{sec:Conclusions}

We have studied linear relations between Feynman integrals that arise from parametric annihilators of the integrand $\G^s$ in the Lee-Pomeransky representation. Seen as a multivariate (twisted) Mellin transform, the integration bijects these special partial differential operators with relations of various shifts (in the indices) of a Feynman integral. In particular, every classical IBP relation (derived in momentum space) is of this type.

The question whether \emph{all} shift relations of Feynman integrals (equivalently, all parametric annihilators of $\G^s$) follow from momentum space relations remains open (see question~\ref{con:Ann=Mom}). We showed that the well-known lowering and raising operators with respect to the dimension are consequences of the classical IBPs. A next step would be to clarify if the same applies to the relations implied by the trivial annihilators \eqref{eq:triv-Ann}.
Similarly, question~\ref{con:Ann=Ann1} asking whether the annihilator $\Ann(\G^s)=\Ann^1(\G^s)$ is linearly generated, remains to be settled.
A positive answer to either of these would imply that the labor-intensive computation of the parametric annihilators could be simplified considerably ($\Mom$ is known explicitly, and $\Ann^1$ can be calculated through syzygies).

The main insight of this article is a statement on the number of master integrals, which we define as the dimension of the vector space of the corresponding family of Feynman integrals over the field of rational functions in the dimension and the indices. Since we treat all indices $\nu_a$ as independent variables, this definition does not account for symmetries (automorphisms) of the underlying graph.
A natural next step is to incorporate these into our setup by studying the action of the corresponding permutation group.
The widely used partition of master integrals into \emph{top-level} and \emph{subsector} integrals, however, seems to be more difficult to take into account, as we touched on in remark~\ref{rem:top-level}.

Our result shows that the number of master integrals is not only finite, but identical to the Euler characteristic of the complement of the hyperspace $\set{\G=0}$ determined by the Lee-Pomeransky polynomial $\G$. This statement follows from a theorem of Loeser and Sabbah. We exemplified several methods to compute this number and found agreement with other established methods.
We expect that, combining the available tools for the computation of the Euler characteristic, it should be possible to compile a program for the efficient calculation of the number of master integrals for a wide range of Feynman graphs.

Let us conclude by emphasizing, once again, that the main objects of the approach elaborated here---the $s$-parametric annihilators generating the integral relations, and the Euler characteristic giving the number of master integrals---are well studied objects in the theory of $D$-modules and furthermore algorithms for their automated computation are available in principle.

In particular, we hope that this parametric, $D$-module theoretic and geometric approach can also shed light on the problems most relevant for perturbative calculations in QFT: the construction of a basis of master integrals, and the actual reduction of arbitrary integrals to such a basis. For this perspective, we would like to point out that our approach of treating the indices $\nu_a$ as free variables, in particular not tied to take integer values, is desirable in order to deal with dimensionally regulated integrals in position space, and for the ability to integrate out one-scale subgraphs (both situations introduce non-integer indices). For a recent step into this direction, see \cite{Tarasov:BoxArbitrary}.

\appendix
\section{Integral representations}
\label{sec:app-representations}

In this appendix we add technical details on the material of section~\ref{sec:Ann and rel}. We summarize the various well-known parametric representations, including their proofs, and the explicit relation to momentum space via propositions~\ref{prop: parametric reps} and~\ref{prop:IBP_Grozin}. Furthermore, we give an alternative, algebraic proof for corollary~\ref{cor:U-F-homogeneity}.

\subsection{Momentum space and Schwinger parameters}
\label{sec:Momentum-space}%

\begin{figure}\centering%
	\includegraphics[width=0.5\textwidth]{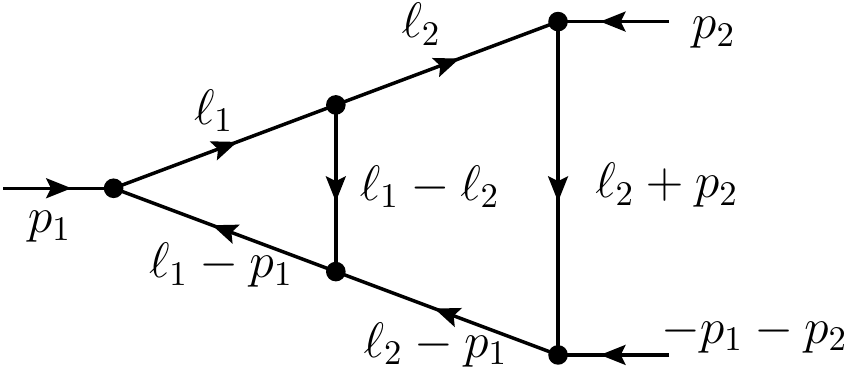}%
	\caption{A two-loop Feynman graph with a choice of loop momenta and the resulting momentum flow.}%
	\label{fig:momentum-flows}%
\end{figure}

As in section \ref{sec:graph-polys} we consider a connected Feynman graph $G$ with $N$ internal edges, $\ExtMoms+1$ external legs and loop-number $\Loops$, which is related by Euler's formula $\Loops=N-V+1$ to the number $V$ of vertices of $G$. Let us consider the case where each edge $e$ of $G$ is associated with a Feynman propagator,\footnote{We use the signature $(1,-1,\ldots,-1)$ for the Minkowski metric.}
\begin{equation*}
	\frac{1}{\Den_{e}}
	= \frac{1}{-\EdgeMom_{e}^{2}+m_{e}^{2} - \iu\FeynEps}
	\quad
	(1\leq e \leq N),
\end{equation*}
which depends on the mass $m_e$ of the particle $e$ and the $\Dim$-dimensional momentum $\EdgeMom_e \in \R^{\Dim}$ flowing through this edge.
Enforcing momentum conservation at each vertex fixes all $\EdgeMom_e$ in terms of $\ExtMoms$ independant external momenta $\ExtMom_1,\ldots,\ExtMom_\ExtMoms$ and $\Loops$ free loop momenta $\LoopMom_1,\ldots,\LoopMom_{\Loops}$.
Note that the actual number of external legs of $G$ is $\ExtMoms+1$, since overall momentum conservation $\sum_{i=1}^{\ExtMoms+1} \ExtMom_i=0$ imposes one relation among the external momenta. Taking only the inverse Feynman propagators $\Den_e$ as denominators, equation~\eqref{eq:FI-momentum} defines the Feynman integral associated to $G$.
\begin{example}
\label{ex:doubletri}%
	The graph in figure~\ref{fig:momentum-flows} has $V=5$ vertices, $N=6$ internal edges and $\Loops=2$ loops. It depends on two independent external momenta $p_1$ and $p_2$. A choice of loop momenta and the resulting momentum flow is depicted in figure~\ref{fig:momentum-flows}. With all masses zero, this yields
	\begin{multline*}
		\FI(\nu_1,\ldots,\nu_6)
		= \int_{\R^{\Dim}} \frac{\dd[\Dim] \LoopMom_1}{\iu \pi^{\Dim/2}}
		  \int_{\R^{\Dim}} \frac{\dd[\Dim] \LoopMom_2}{\iu \pi^{\Dim/2}}
		  \frac{1}{
			  [-\LoopMom_1^2-\iu\FeynEps]^{\nu_1}
			  [-\LoopMom_2^2-\iu\FeynEps]^{\nu_2}
			  [-(\LoopMom_1-\LoopMom_2)^2-\iu\FeynEps]^{\nu_3}
		  }\\
		\times \frac{1}{
			  [-(\LoopMom_2+p_1)^2-\iu\FeynEps]^{\nu_4}
			  [-(\LoopMom_2-p_2)^2-\iu\FeynEps]^{\nu_5}
			  [-(\LoopMom_1-p_2)^2-\iu\FeynEps]^{\nu_6}
	  	  }
		.
	\end{multline*}
\end{example}
Typically, the number $\abs{\bilabels}=\Loops(\Loops+1)/2+\Loops\ExtMoms$ of independent scalar products $s_{\set{i,j}}$ in \eqref{eq:scalar-products} is larger than the number of edges in a graph $G$.
We can then extend the initial set of denominators (given as the inverse propagators of the graph) by a suitable choice of additional quadratic (or linear) forms in the loop momenta, such that we reach a set of $\abs{\bilabels}$ denominators with the property that the matrix $\spMat$ defined by
\begin{equation*}
	\Den_{a}
	=\sum_{ \set{i,j} \in \bilabels} \spMat_{a}^{\set{i,j}} s_{\set{i,j}}+\lambda_{a}
\end{equation*}
becomes invertible. This means that all loop-momentum dependent scalar products can be written as linear combinations of the denominators, see \eqref{eq:sp-from-props}. The additional denominators introduced in this way are called \emph{irreducible scalar products}.
\begin{example}\label{ex:doubletri-ISP}
	We again consider the graph from figure~\ref{fig:momentum-flows} with $6$ internal edges labelled as in example~\ref{ex:doubletri}. In the massless case, the inverse propagators are just $\Den_e = -\EdgeMom_e^2$ and their explicit decomposition into the $\abs{\bilabels}=7$ scalar products takes the form
	\begin{equation*}
		\begin{pmatrix}
			\Den_1 \\
			\Den_2 \\
			\Den_3 \\
			\Den_4 \\
			\Den_5 \\
			\Den_6 \\
		\end{pmatrix}
		=
		\begin{pmatrix}
			- \LoopMom_1^2 \\
			- \LoopMom_2^2 \\
			- (\LoopMom_1-\LoopMom_2)^2 \\
			- (\LoopMom_2+\ExtMom_1)^2 \\
			- (\LoopMom_2-\ExtMom_2)^2 \\
			- (\LoopMom_1-\ExtMom_2)^2 \\
		\end{pmatrix}
		=
		\underbrace{\begin{pmatrix}
			-1 & 0 & 0 & 0 & 0 & 0 & 0 \\
			 0 & 0 & -1& 0 & 0 & 0 & 0 \\
			 -1& 2 & -1& 0 & 0 & 0 & 0 \\
			 0 & 0 & -1& 0 & 0 &-2 & 0 \\
			 0 & 0 & -1& 0 & 0 & 0 & 2 \\
			 -1& 0 & 0 & 0 & 2 & 0 & 0 \\
		\end{pmatrix}}_{\spMat}
		\begin{pmatrix}
			\LoopMom_1^2 \\
			\LoopMom_1 \LoopMom_2 \\
			\LoopMom_2^2 \\
			\LoopMom_1 \ExtMom_1 \\
			\LoopMom_1 \ExtMom_2 \\
			\LoopMom_2 \ExtMom_1 \\
			\LoopMom_2 \ExtMom_2 \\
		\end{pmatrix}
		+
		\underbrace{\begin{pmatrix}
			0 \\ 0 \\ 0 \\ -\ExtMom_1^2 \\ -\ExtMom_2^2 \\ -\ExtMom_2^2 \\
		\end{pmatrix}}_{\lambda}
		.
	\end{equation*}
	The matrix $\spMat$ has rank $6$ and annihilates $(0,0,0,1,0,0,0)^\Transpose$. Thus we can choose $\Den_7 = \LoopMom_1 \ExtMom_1$ as an irreducible scalar product to complete the basis of quadratic forms in the loop momenta. The matrix $\spMat$ then acquires an additional row $(0,0,0,1,0,0,0)$ and becomes invertible.
\end{example}
Note that we assume that the inverse propagators of $G$ (the initial set of denominators) are linearly independent (that is, the $N \times \abs{\bilabels}$ matrix $\spMat$ of the inverse propagators has full rank $N$) in order to be able to extend them to a basis of quadratic forms by choosing $\abs{\bilabels}-N$ irreducible scalar products.%
\footnote{%
	If there are linear dependencies between the inverse propagators, these relations
	imply that the Feynman integral can be expressed in terms of contracted graphs with
	linearly independent inverse propagators.
	For example, if $\alpha \Den_1+\beta \Den_2 = 1$, then iterated use of
		$1/(\Den_1 \Den_2)=\alpha/\Den_2 + \beta/\Den_1$
	allows one to ultimately eliminate one of $\Den_1$ or $\Den_2$. Therefore, requesting
	linear independence is no restriction.
}

For each denominator $\Den_a$ we introduce a scalar $x_{a}$, which is known as Schwinger-, Feynman- or $\alpha$-parameter. In definition \ref{def:M-Q-J} we have introduced the decomposition 
\begin{equation}
	\sum_{a=1}^{N}x_{a} \Den_{a}
	=-\sum_{i,j=1}^{\Loops} \loopMat_{ij}\LoopMom_{i}\LoopMom_{j}
	+\sum_{i=1}^{\Loops} 2 \loopLin_{i}\LoopMom_{i}
	+\loopConst
	,
	\label{eq:def:M-Q-JAppendix}%
\end{equation}
which determines a symmetric $\Loops\times\Loops$ matrix $\loopMat$, a vector $Q$ and a scalar $\loopConst$. With their help we defined the polynomials $\U$, $\F$ and $\G=\U+\F$ in \eqref{eq:def-U-F}.
Explicitly, from definition \ref{def:M-Q-J} and \eqref{eq:bilinear-decomposition} we can read off that
\begin{align}
\loopMat_{ij} 
	& = -\frac{1+\delta_{ij}}{2}\sum_{a=1}^{N}x_{a} \spMat_{a}^{\set{i,j}}
	\quad\text{for $1\leq i,j\leq \Loops$,}
	\label{eq:M-from-A}\\
\loopLin_{i}
	& = \frac{1}{2}\sum_{a=1}^{N}\sum_{j=\Loops+1}^{M}x_{a} \spMat_{a}^{\set{i,j}} q_j
	\quad\text{for $1\leq i \leq \Loops$ and}
	\label{eq:Q-from-A}\\
\loopConst
	& = \sum_{a=1}^{N}x_{a}\lambda_{a}
	.
	\label{eq:J}%
\end{align}
Since $\loopMat_{ij}$ is an $\Loops\times\Loops$ matrix with entries that are linear in the Schwinger parameters, the polynomial $\U$ is homogeneous of degree $\Loops$.
By Cramer's rule, $(\det\loopMat)\loopMat^{-1}_{ij}$ is homogeneous of degree $\Loops-1$ and the linearity of $\loopLin$ and $\loopConst$ in the $x_a$ implies
\begin{cor}
\label{cor:U-F-homogeneity}%
	$\U$ and $\F$ are homogeneous polynomials in the variables
	$x_{1},\ldots, x_{N}$ with the degrees $\deg(\U)=\Loops$
	and $\deg(\F)=\Loops+1$. Hence, for $\G=\U+\F$, we have
\begin{equation}
	\left(\sum_{a=1}^{N}x_{a}\partial_a\right) \G 
	= \Loops \U + (\Loops+1) \F
	= (\Loops+1)\G-\U
	= \Loops \G + \F.
	\label{eq:G-euler}%
\end{equation}
\end{cor}
Let us now come to the proof of proposition~\ref{prop: parametric reps} following \cite{Nak} and \cite{LeePom}.
\begin{proof}[Proof of proposition~\ref{prop: parametric reps}]
	We consider the Feynman integral defined in \eqref{eq:FI-momentum},
\begin{equation}
	\FI(\nu_{1},\ldots,\nu_{N})
	= \left(\prod_{j=1}^{\Loops}\int\frac{\dd[\Dim] \LoopMom_{j}}{\iu \pi^{\Dim/2}}\right)
	\prod_{i=1}^{N} \Den_{i}^{-\nu_{i}}
	.
	\label{eq:FI-momentumAppendix}%
\end{equation}
	Using the Schwinger trick to exponentiate each denominator,\footnote{%
		The integral~\eqref{eq:Schwinger-trick} converges only for $\Re(\nu_a)>0$ and therefore restricts the domain of convergence for the parametric integral. However, this has no consequences for algebraic relations, see remark~\ref{rem:continuation}.
	}
	\begin{equation}
		\frac{1}{\Den_a^{\nu_a}}
		= \frac{1}{\Gamma(\nu_a)} \int_0^{\infty} x_a^{\nu_a-1} e^{-x_a \Den_a} \dd x_a
		,
		\label{eq:Schwinger-trick}%
	\end{equation}
	the integral in \eqref{eq:FI-momentumAppendix} turns into
\begin{equation*}
	\FI(\nu_{1},\ldots,\nu_{N})
	=
	\left(\prod_{i=1}^{N}\int_{0}^{\infty}\frac{x_{i}^{\nu_{i}-1}\dd x_{i}}{\Gamma\left(\nu_{i}\right)}\right)
	\left(\prod_{j=1}^{\Loops}\int\frac{\dd[\Dim] \LoopMom_{j}}{\iu\pi^{\Dim/2}}\right)
	e^{-\sum_{a=1}^{N}x_{a} \Den_{a}}
	.
\end{equation*}
	According to \eqref{eq:def:M-Q-JAppendix} and \eqref{eq:def-U-F}, we can complete the square in the exponent
	\begin{equation*}
		-\sum_{a=1}^N x_a \Den_a
		= (\LoopMom - \loopMat^{-1} \loopLin)^{\Transpose} \loopMat (\LoopMom-\loopMat^{-1} \loopLin) - \F/\U
	\end{equation*}
	to perform the Gau{\ss}ian integrals over the shifted loop momenta $\LoopMom' \defas \LoopMom - \loopMat^{-1} \loopLin$ as\footnote{%
		Recall that our metric has signature $(1,-1,\ldots,-1)$, so the integrations over the $\Dim-1$ spacelike components are Euclidean and give $\sqrt{\pi^{\Loops} \U}$ each. The timelike integrations are understood as contour integrals and yield the same factor after rotating the integration contour to the imaginary axis, according to the Feynman $\iu\FeynEps$-prescription.
	}%
	\begin{equation*}
		\left(\prod_{j=1}^{\Loops}\int\frac{\dd[\Dim] \LoopMom'_{j}}{\iu\pi^{\Dim/2}}\right) e^{(\LoopMom')^{\Transpose} \loopMat \LoopMom'}
		= (\det \loopMat)^{-\Dim/2}
		= \U^{-\Dim/2}
		.
	\end{equation*}
	In summary, we therefore arrive at the integral representation \eqref{eq:parametric-exp}:
\begin{equation*}
	\FI(\nu_{1},\ldots,\nu_{N})
	=
	\left(\prod_{i=1}^{N}\int_{0}^{\infty}\frac{x_{i}^{\nu_{i}-1}\dd x_{i}}{\Gamma\left(\nu_{i}\right)}\right)
	\frac{e^{-\F/\U}}{\U^{\Dim/2}}
	.
\end{equation*}
We now multiply with $1=\int_0^{\infty} \delta(\rho-\sum_{j=1}^N x_j) \dd\rho$ and substitute $x_a \rightarrow \rho x_a$.\footnote{%
	Much more generally, we could replace $\sum_{j=1}^N x_j$ in the $\delta$-constraint with any other function as long as it is homogeneous of degree 1 and positive on $\R_+^N$.
}
The Jacobian $\rho^N$, the monomials $x_i^{\nu_i-1}$ and $\delta(\rho-\sum_j x_j) \rightarrow \delta(1-\sum_j x_j)/\rho$ contribute the power $\rho^{\abs{\nu}-1}$, whereas the homogeneity of $\F$ and $\U$ from corollary~\ref{cor:U-F-homogeneity} implies that $\U\rightarrow \rho^{\Loops}\U$ and $\F/\U\rightarrow \rho \F/\U$. 
Overall, by realising that the integral over $\rho$ is
	\begin{equation*}
		\int_0^{\infty} \rho^{\sdc-1} e^{-\rho \F/\U} \dd \rho
		= \Gamma(\sdc) \left( \frac{\U}{\F} \right)^{\sdc}
		, 
	\end{equation*}
we arrive at the first parametric formula \eqref{eq:FI-UF}.
Similarly, we multiply the integrand of \eqref{eq:Lee-Pom} with $1=\int_0^{\infty} \delta(\rho-\sum_{i} x_i) \dd\rho$ and substitute $x_i \rightarrow \rho x_i$. Using $\U \rightarrow \rho^{\Loops} \U$ and $\F\rightarrow \rho^{\Loops+1}\F$ from corollary~\ref{cor:U-F-homogeneity}, the integral over $\rho$ becomes
	\begin{equation*}
		\int_0^{\infty} \rho^{\sdc-1} (\U + \rho \F)^{-\Dim/2}
		= \U^{-\Dim/2} \left( \frac{\U}{\F} \right)^{\sdc}
		\frac{\Gamma(\sdc)\Gamma(\frac{\Dim}{2}-\sdc)}{\Gamma(\frac{\Dim}{2})}
	\end{equation*}
	and combines with the prefactors in \eqref{eq:Lee-Pom} to reproduce \eqref{eq:Lee-Pom}.
\end{proof}
We conclude the section with the proof of proposition~\ref{prop:IBP_Grozin} following Grozin~\cite{Grozin:IBP}:
\begin{proof}[Proof of proposition~\ref{prop:IBP_Grozin}]
	The action of $\momIBP{i}{j}$ on the integrand from \eqref{eq:momentum-integrand} is
\begin{equation*}
	\momIBP{i}{j}f 
	= \Dim \delta_{ij}f+f\sum_{a=1}^{N}\frac{-\nu_{a}}{\Den_{a}}q_{j}\frac{\partial \Den_{a}}{\partial q_{i}}.
\end{equation*}
	According to \eqref{eq:bilinear-decomposition}, the chain rule gives
\begin{equation*}
	q_{j}\frac{\partial}{\partial q_{i}} \Den_{a} 
	=
	q_{j}\frac{\partial}{\partial q_{i}} \sum_{\set{k,m}\in\bilabels} \spMat_a^{\set{k,m}} q_k q_m
	= \sum_{m=1}^{M} \spMat_{a}^{\set{i,m}}(1+\delta_{i,m})q_{j}q_{m}
\end{equation*}
	and we can express the scalar products $q_{j}q_{m}$ with $\set{j,m}\in\bilabels$ in terms of denominators using \eqref{eq:sp-from-props}. The remaining terms with $j,m>\Loops$ are products of external momenta, so
\begin{align*}
	\momIBP{i}{j} f 
	& = d\delta_{ij}f-f\sum_{a,b=1}^{N} C_{aj}^{bi}\frac{\nu_{a}}{\Den_{a}}\left(\Den_{b}-\lambda_{b}\right)
	\quad\text{for}\quad 1\leq j\leq \Loops \quad \text{and}
\\
	\momIBP{i}{j}f 
	& = -f\sum_{a,b=1}^{N} C_{aj}^{bi}\frac{\nu_{a}}{\Den_{a}}\left(\Den_{b}-\lambda_{b}\right)-f\sum_{a=1}^{N}\sum_{m=\Loops+1}^{M} \spMat_{a}^{\set{i,m}}q_{j}q_{m}\frac{\nu_{a}}{\Den_{a}}
	\quad\text{if}\quad \Loops <j \leq M.
\end{align*}
We conclude by noticing that multiplying the integrand $f$ with $\nu_a/\Den_a$ is equivalent to the action of the operator $\PlusD{a}$ defined in \eqref{eq:def-shift-action}, whereas multiplication with $\Den_b$ lowers the index $\nu_b$ and corresponds to $\Minus{b}$.
\end{proof}

\subsection{Algebraic proof for corollary \ref{cor:parIBP}}
\label{sec:Algebraic proof}%

With the proof of corollary~\ref{cor:parIBP} we have shown that, for every momentum space IBP relation, there is a corresponding annihilator in $\Ann\left( \G^{-\Dim/2} \right)$. The proof
rests on the inverse Mellin transform, which may be seen as a convenient but rather abstract argument. As a more direct alternative, we prove the statement in a purely algebraical way by use of properties of the graph polynomials.  
\begin{lem} \label{lem:Mom_in_Ann}
	\label{lem:parIBP-algebraic}%
	The operators $\parIBP{i}{j}$ from \eqref{eq:parIBP<=L} and \eqref{eq:parIBP>L} corresponding to the momentum space IBP relations annhilate the parametric integrand $\G^{-\Dim/2}$:
	\begin{equation*}
		\parIBP{i}{j} \in \Ann\left( \G^{-\Dim/2} \right)
		\quad\text{for all $1\leq i \leq \Loops$ and $1\leq j \leq M$.}
	\end{equation*}
\end{lem}

\begin{proof}
	Let us first consider the case $j\leq \Loops$. After acting with $\parIBP{i}{j}$ from \eqref{eq:parIBP<=L} on $\G^{-\Dim/2}$ and dividing by $(\Dim/2) \G^{-\Dim/2-1}$, we are left to prove the vanishing of
	\begin{equation}
		2\G \delta_{i,j}
		-\sum_{a,b} C^{bi}_{aj} x_a \left( \partial_b \G - \lambda_b \Big[ \Loops+1-\sum_c x_c \partial_c \Big] \G \right)
		=2\G \delta_{i,j}
		- \sum_{a,b} C^{bi}_{aj} x_a \left( \partial_b \G - \lambda_b \U \right)
		\label{eq:ibp-ann<=L}%
	\end{equation}
	where we exploited the homogeneity from \eqref{eq:G-euler}. Using \eqref{eq:def-U-F}, we note that
	\begin{equation}
		\partial_b \G - \lambda_b \U
		= \partial_b \left[ \U\left(1+ \loopConst + \loopLin^{\Transpose} \loopMat^{-1} \loopLin \right) \right] - \lambda_b \U
		= \G \frac{\partial_b \U}{\U} + \U \partial_b\left( \loopLin^\Transpose\loopMat^{-1}\loopLin \right)
		\label{eq:dG/db}%
	\end{equation}
	because $\partial_b \loopConst = \lambda_b$ according to \eqref{eq:J}.
	In order to evaluate $\partial_b \U$ with Jacobi's formula $(\partial_b \U)/\U = \sum_{r,s=1}^{\Loops} \loopMat^{-1}_{r,s} \partial_b \loopMat_{s,r}$, we use \eqref{eq:M-from-A} to compute
	\begin{equation*}
		\sum_b \spMat^b_{\set{m,j}} \partial_b \loopMat_{s,r}
		= -\frac{1+\delta_{s,r}}{2}\sum_b \spMat^b_{\set{m,j}} \spMat_b^{\set{s,r}}
		= -\frac{1+\delta_{s,r}}{2} \delta_{\set{m,j},\set{r,s}}
		= -\frac{\delta_{m,r} \delta_{j,s} + \delta_{m,s} \delta_{j,r}}{2}
	\end{equation*}
	which restricts $m$ to either $r$ or $s$. So in particular, $m\leq\Loops$ and we can use \eqref{eq:M-from-A} in
	\begin{equation}
		\sum_{a,b} C^{bi}_{aj} x_a \partial_b \loopMat_{r,s}
		=- \sum_{a,m} x_a \spMat_a^{\set{i,m}}\frac{1+\delta_{i,m}}{2}
		\left( \delta_{m,r} \delta_{j,s} + \delta_{m,s} \delta_{j,r} \right) 
		= \loopMat_{i,r} \delta_{j,s} + \loopMat_{i,s} \delta_{j,r}
		\label{eq:CxdM/db}%
	\end{equation}
	which proves that for arbitrary $j$ (independent of whether $j\leq \Loops$ or $j>\Loops$)
	\begin{equation}
		\sum_{a,b} C^{bi}_{aj} x_a \frac{\partial_b \U}{\U}
		= \sum_{r,s=1}^{\Loops} \loopMat^{-1}_{r,s}
		  (\loopMat_{i,r} \delta_{j,s} + \loopMat_{i,s} \delta_{j,r})
		= 2 \delta_{i,j}
		.
		\label{eq:CxdbU/U}%
	\end{equation}
	Via \eqref{eq:dG/db}, this identity reduces the proof of \eqref{eq:ibp-ann<=L} to showing that
	\begin{equation}
		\sum_{a,b} C^{bi}_{aj} x_a \partial_b \left( \loopLin^{\Transpose} \loopMat^{-1} \loopLin \right)
		=
		\sum_{a,b} C^{bi}_{aj} x_a 
		\Big[
			2 (\partial_b \loopLin)^{\Transpose} \loopMat^{-1} \loopLin
			- \loopLin^{\Transpose}\loopMat^{-1} (\partial_b \loopMat) \loopMat^{-1} \loopLin
		\Big]
		\label{eq:d(QMinvQ)<=L}%
	\end{equation}
	vanishes. The last term is easily evaluated with \eqref{eq:CxdM/db} and gives
	\begin{equation}
		\sum_{a,b} C^{bi}_{aj} x_a \loopLin^\Transpose \loopMat^{-1} (\partial_b \loopMat) \loopMat^{-1} \loopLin
		= 2\sum_{r,s=1}^{\Loops} (\loopMat^{-1}\loopLin)_r (\loopMat^{-1}\loopLin)_s \loopMat_{i,r} \delta_{j,s}
		= 2 \loopLin_i (\loopMat^{-1} \loopLin)_j
		,
		\label{eq:QMinvMMinvQ}%
	\end{equation}
	whereas the derivative $\partial_b \loopLin$ can be read off from \eqref{eq:Q-from-A} and the sum over $b$ yields
	\begin{equation}
		\sum_b \spMat_{\set{m,j}}^b (2\partial_b \loopLin_s)
		= \sum_b \spMat_{\set{m,j}}^b \sum_{r>\Loops} \spMat_b^{\set{s,r}} q_r
		= \sum_{r>\Loops} q_r \delta_{\set{m,j},\set{s,r}}
		= \sum_{r>\Loops} q_r \delta_{m,r} \delta_{s,j}
		\label{eq:AdQ/db}%
	\end{equation}
	because $j\leq \Loops< r$ excludes the possibility that $m=s$ and $r=j$. Thus with \eqref{eq:Q-from-A},
	\begin{equation}
		\sum_{a,b} C^{bi}_{aj} x_a (2\partial_b \loopLin)^{\Transpose}\loopMat^{-1} \loopLin
		=
		(\loopMat^{-1} \loopLin)_j \sum_a \sum_{m>\Loops} x_a \spMat_a^{\set{i,m}} q_m
		= 2 \loopLin_i (\loopMat^{-1} \loopLin)_j
		\label{eq:dQ/db<=L}%
	\end{equation}
	cancels the contribution from \eqref{eq:QMinvMMinvQ} in \eqref{eq:d(QMinvQ)<=L} and finishes the proof in the case $j\leq \Loops$.
	If instead we have $j>\Loops$, then we must replace $\delta_{\set{m,j},\set{s,r}} = \delta_{j,r} \delta_{m,s}$ in \eqref{eq:AdQ/db} such that
	\begin{equation*}
		\sum_{a,b} C^{bi}_{aj} x_a (2\partial_b \loopLin)^{\Transpose}\loopMat^{-1} \loopLin
		= \sum_{s=1}^{\Loops} \sum_a \spMat_a^{\set{i,s}} (1+\delta_{i,s})
		q_j (\loopMat^{-1}\loopLin)_s
		= -2 \sum_{s=1}^{\Loops} \loopMat_{i,s} (\loopMat^{-1}\loopLin)_s
		= -2 Q_i q_j
	\end{equation*}
	where we used \eqref{eq:M-from-A} once more.
	Now recall that \eqref{eq:CxdbU/U} remains true and becomes zero for $j>\Loops$ because $\delta_{i,j}=0$ since $i\leq \Loops$. For the same reason, $\delta_{j,s}=0$ in \eqref{eq:QMinvMMinvQ} and therefore, using \eqref{eq:dG/db},
	\begin{equation*}
		-\sum_{a,b} C^{bi}_{aj} x_a (\partial_b \G - \lambda_b \U)
		=
		-\U\sum_{a,b} C^{bi}_{aj} x_a (2\partial_b \loopLin)^{\Transpose}\loopMat^{-1} \loopLin
		= 2 \U \loopLin_i q_j.
	\end{equation*}
	This is precisely cancelled by the additional contribution to \eqref{eq:ibp-ann<=L} coming from $\parIBP{i}{j}$ in \eqref{eq:parIBP>L} in the case $j>\Loops$: The additional term acts on $\G^{-\Dim/2}$ as
	\begin{equation*}
		-\sum_a \sum_{m>\Loops} \spMat_a^{\set{i,m}} q_j q_m x_a \Big[ \Loops + 1-\sum_c x_c \partial_c \Big] \G
		= -\U q_j \sum_a \sum_{m>\Loops} \spMat_a^{\set{i,m}} x_a q_m
		= - 2\U q_j \loopLin_i
	\end{equation*}
	after dividing by $(\Dim/2) \G^{-\Dim/2-1}$. Note that here we used \eqref{eq:G-euler} and \eqref{eq:Q-from-A}.
\end{proof}

\subsection{The Baikov representation}
\label{sec:Baikov}%

In this section we discuss the representation of Feynman integrals suggested by Baikov in \cite{Baikov:ExplicitSolutions-multiloop}, whose complete form \eqref{eq:Baikov-FI} was given by Lee in \cite{Lee:LL2010,Lee:DimensionalRecurrenceAnalyticalProperties}. We will give some details on the derivation of this formula (see also \cite[section~9]{Grozin:IBP}), which was presented in \cite{Lee:LL2010} and applied in our discussion of the lowering dimension shift in section~\ref{sec:Dimension shifts}.

Assume that $q_1$, \ldots, $q_{\AllMoms}$ are vectors in a Euclidean vector space and write
\begin{equation}
	V_n \defas \begin{pmatrix}
		q_n \cdot q_n & \cdots & q_n \cdot q_{\AllMoms} \\
		\vdots & \ddots & \vdots \\
		q_{\AllMoms} \cdot q_n & \cdots & q_{\AllMoms} \cdot q_{\AllMoms} \\
	\end{pmatrix}
	= \left( q_i \cdot q_j \right)_{n \leq i,j\leq \AllMoms}
	\quad\text{and}\quad
	G_n \defas \det V_n
	\label{eq:def-Gram}%
\end{equation}
for their Gram matrices and determinants. Note that
\begin{equation*}
	V_n = \begin{pmatrix}
		q_n^2 & q_{\bullet} \cdot q_n \\
		q_n \cdot q_{\bullet} & V_{n+1} \\
	\end{pmatrix}
	\quad\text{where}\quad
	q_{\bullet} \cdot q_n
	\defas \begin{pmatrix}
		q_{n+1} \cdot q_n \\
		\vdots \\
		q_{\AllMoms} \cdot q_n \\
	\end{pmatrix},
	\quad
	q_n \cdot q_{\bullet}
	\defas \left(q_{\bullet} \cdot q_n\right)^{\Transpose}
\end{equation*}
and thus, by adding $-(p_n \cdot p_{\bullet}) V_{n+1}^{-1}$ times the lower $\AllMoms-n$ rows to the first row,
\begin{equation}
	\frac{G_n}{G_{n+1}} =
	q_n^2-\norm{\prOrth{\lin\set{q_{n+1},\ldots,q_{\AllMoms}}}(q_n)}^2
	= \norm{\prOrth{\lin\set{q_{n+1},\ldots,q_{\AllMoms}}^{\perp}}(q_n)}^2
	.
	\label{eq:Gram-recursion}%
\end{equation}
Indeed, the formula $\prOrth{\lin\set{q_{n+1},\ldots,q_{\AllMoms}}} (v) = \sum_{i,j=n+1}^{\AllMoms} q_i \left( V^{-1}_{n+1} \right) _{i,j} (q_j \cdot v)$ for the orthogonal projection of $v$ onto the space spanned by $q_{n+1},\ldots,q_{\AllMoms}$ shows that
\begin{align*}
	\norm{\prOrth{\lin\set{q_{n+1},\ldots,q_{\AllMoms}}}(q_n)}^2
	&= \sum_{i,j,k,l=n+1}^{\AllMoms} \underbrace{(q_i \cdot q_k)}_{\left(V_{n+1}\right)_{i,k}}
		\left(V^{-1}_{n+1} \right)_{i,j} (q_j \cdot q_n)
		\left(V^{-1}_{n+1} \right)_{k,l} (q_k \cdot q_n)
	\\
	&= \sum_{j,k,l=n+1}^{\AllMoms} 
		\delta_{k,j} (q_j \cdot q_n)
		\left(V^{-1}_{n+1} \right)_{k,l} (q_l \cdot q_n)
	= \sum_{k,l=n+1}^{\AllMoms} 
		(q_k \cdot q_n)
		\left(V^{-1}_{n+1} \right)_{k,l} (q_l \cdot q_n)
	\\
	&= (q_n\cdot q_{\bullet})V^{-1}_{n+1} (q_{\bullet} \cdot q_n)
	.
\end{align*}
Now assume our integrand $f$ only depends on the scalar products $s_{i,j} = q_i \cdot q_j$, and we want to integrate out the first loop momentum $q_1$. Let us decompose $q_1 = q_{\perp} + q_{\parallel}$ into the component $q_{\parallel} \in \lin \set{q_2,\ldots,q_{\AllMoms}}$ that lies in the space spanned by the other momenta, and the component $q_{\perp}$ in its orthogonal complement.
According to \eqref{eq:Gram-recursion}, $G_n^{1/2}$ is the volume of the parallelotope spanned by $q_n,\ldots,q_{\AllMoms}$.
Hence, changing coordinates from $q_{\parallel}$ to $(s_{1,2},\ldots,s_{1,\AllMoms})$ yields 
\begin{equation*}
	\int_{\R} \dd s_{1,2}\ \cdots \int_{\R} \dd s_{1,\AllMoms}
	= \sqrt{G_2}
	\int_{\R^{\AllMoms-1}} \dd[\AllMoms-1] q_{\parallel}
	.
\end{equation*}
The integral over the orthogonal component is, due to $s_{1,1}=q_1^2=q_{\perp}^2 + q_{\parallel}^2$, given by
\begin{align*}
	\int_{\R^{\Dim-\AllMoms+1}}\hspace{-9mm} \dd[\Dim-\AllMoms+1] q_{\perp}
	&= \frac{\pi^{(\Dim-\AllMoms+1)/2}}{\Gamma\left(\frac{\Dim-\AllMoms+1}{2}\right)}
	\int_{0}^{\infty} \!\!\dd q_{\perp}^2 \left(q_{\perp}^2 \right)^{(\Dim-\AllMoms-1)/2}
	= \frac{\pi^{(\Dim-\AllMoms+1)/2}}{\Gamma\left(\frac{\Dim-\AllMoms+1}{2}\right)}
	\int_{q_{\parallel}^2}^{\infty} \!\!\dd s_{1,1}
	\left( \frac{G_1}{G_2}\right)^{(\Dim-\AllMoms-1)/2}
	.
\end{align*}
Note that the lower boundary $s_{1,1}=q_{\parallel}^2$ corresponds to $0=q_{\perp}^2=G_1/G_2$. Altogether,
\begin{equation}
	\int_{\R^{\Dim}} \frac{\dd[\Dim] q_1}{\pi^{\Dim/2}}
	\ f(s)
	=
	\frac{\pi^{(1-\AllMoms)/2}}{\Gamma\left(\frac{\Dim-\AllMoms+1}{2}\right)}
	\int_{G_1/G_2>0} \!\!\dd[\AllMoms] s_{1,\bullet}
	\ \ f(s)
	\frac{G_1^{(\Dim-\AllMoms-1)/2}}{G_{2}^{(\Dim-\AllMoms)/2}}
	.
	\label{eq:measure-s1j}%
\end{equation}
Transforming the remaining loop integrations analogously, all but two of the Gram determinants cancel, and we conclude that
\begin{equation}
	\prod_{i=1}^{\Loops} \int \frac{\dd[\Dim] \LoopMom_i}{\pi^{\Dim/2}}
	f(s)
	= \frac{\pi^{-\Loops E/2-\Loops(\Loops-1)/4}}{\Gamma\left( \frac{\Dim-\Loops-E+1}{2} \right) \cdots \Gamma\left( \frac{\Dim-E}{2} \right)}
	\int \dd[N] s_{\bullet,\bullet}
	\ f(s)
	\cdot \frac{G_1^{(\Dim-\AllMoms-1)/2}}{G_{\Loops+1}^{(\Dim-E-1)/2}}
	\label{eq:measure-all-s}%
\end{equation}
where $\AllMoms=\Loops+\ExtMoms$ is the sum of the number $\Loops$ of loops and the number $\ExtMoms$ of linearly independent external momenta. Note that $G_{\Loops+1} = \det (\ExtMom_i \cdot \ExtMom_j)_{1\leq i,j\leq \ExtMoms}$ is the Gram determinant of the external momenta and independent of the integration variables.
\begin{proof}[Proof of Theorem~\ref{thm:Baikov}]
	Since every denominator $\Den_e$ is a linear combination of $s_{i,j}$ (and some loop momentum independent constant $\lambda_e$) according to \eqref{eq:bilinear-decomposition}, an affine change of variables allows us to integrate over the values of $\Den_e$'s instead of $s_{i,j}$'s. This transformation only introduces a constant Jacobian $c' = \det (\spMat^a_{\set{i,j}})$. We set $f(s)=\prod_e \Den_e^{-\nu_e}$ in the formula \eqref{eq:measure-all-s} above.
	The (Euclidean) integration domain is determined, according to \eqref{eq:Gram-recursion}, by 
	$0 < \|\prOrth{\lin\set{q_{n+1},\ldots,q_{\AllMoms}}^{\perp}}(q_n)\|^2 = G_n/G_{n+1}$
	for $1\leq n \leq \Loops$. Therefore, a point on the boundary of the integration domain is determined by $G_n=0$ for some $1\leq n \leq L$, which is equivalent to a linear dependence $q_n \in \lin\set{q_{n+1},\ldots,q_{\AllMoms}}$ and hence implies $G_1=0$.

	Note that we have to analytically continue \eqref{eq:measure-all-s} from Euclidean to Minkowski space in order to obtain the Feynman integral \eqref{eq:FI-momentum}. As Wick rotation turns $\int \dd[\Dim] \LoopMom_k/(\iu \pi^{\Dim/2})$ exactly into the measure $\int \dd[\Dim] \LoopMom_k/\pi^{\Dim/2}$ on the left-hand side of \eqref{eq:measure-all-s}, we only have to remember that, due to our mostly-minus signature $(1,-1,\cdots,-1)$ of the Minkowsi metric, the Euclidean scalar products on the right-hand side of \eqref{eq:measure-all-s} receive a factor $(-1)$. For example, the $\AllMoms \times \AllMoms$ determinant $G_1$ turns into $(-1)^{\AllMoms} \Gram_1$; similarly, $G_{\Loops+1}$ becomes $(-1)^{\ExtMoms}\Gram$.
	Overall, analytic continuation gives an additional factor of
	\begin{equation*}
		(-1)^{N}\cdot \frac{(-1)^{\AllMoms(\Dim-\AllMoms-1)/2}}{(-1)^{\ExtMoms(\Dim-\ExtMoms-1)/2}}
		= (-1)^{\Loops \Dim/2} \cdot (-1)^{N+\AllMoms(\AllMoms-1)/2-\ExtMoms(\ExtMoms-1)/2}
		.
	\end{equation*}
	We absorb the last factor, together with the Jacobian $c'$, into the constant prefactor $c$, and have thus finally arrived at \eqref{eq:Baikov-FI}.
\end{proof}

\section{The theory of Loeser-Sabbah}
\label{sec:L-S}

This section is devoted to the theorem~\ref{thm:L-S}, which was first stated in \cite{LoeserSabbah:IrredTore}. Beware that the original argument is flawed; a correct (but terse) proof was given in \cite{LoeserSabbah:IrredToreII}. Our aim here is to provide a simplified and more detailed derivation.

Throughout we will consider modules $\M$ over the algebra $\WeylT{N}_k=\Weyl{N}_k[x^{-1}]$ of differential operators \eqref{eq:WeylT} on the torus in some number $N$ of variables $x_i$, over some field $k$ of characteristic zero. To lighten the notation, let us abbreviate $\theta_i \defas x_i \partial_i$ and set
\begin{equation}
	\M(\theta_i,\ldots,\theta_N)
	\defas k(\theta_i,\ldots,\theta_N) \tp_{k[\theta_i,\ldots,\theta_N]} \M
	\label{eq:def-alg-Mellin}%
\end{equation}
for the \emph{algebraic Mellin transform} \cite[section~1.2]{LoeserSabbah:EDFdeterminants}. We begin with the finite-dimensionality, which was proven in \cite[Lemme~1.2.2]{LoeserSabbah:EDFdeterminants}:
\begin{lem}
	Let $\M$ denote a holonomic $\WeylT{N}_k$-module. Then, for any $1\leq i \leq N$, its algebraic Mellin transform $\M(\theta_i,\ldots,\theta_N)$ is a holonomic $\WeylT{i-1}_{k(\theta_i,\ldots,\theta_N)}$-module.
	\label{lem:L-S:holonomicity}%
\end{lem}
\begin{cor}
	\label{cor:L-S:findim}%
	The full Mellin transform $\M(\theta_1,\ldots,\theta_N)$ is a finite-dimensional vector space over the field $k(\theta_1,\ldots,\theta_N)$.
\end{cor}
\begin{proof}[Proof of lemma~\ref{lem:L-S:holonomicity}]
	Since $\M(\theta_i,\ldots,\theta_N) = [\M(\theta_{i+1},\ldots,\theta_N)](\theta_i)$, it suffices (by induction over $i$) to consider the case $i=N$.
	Introducing a new indeterminate $\nu$, we extend the scalars from $k$ to $k[\nu]$ to obtain a $\WeylT{N}_{k[\nu]}$-module $\M[\nu] \defas \M \tp_k k[\nu]$. It sits in an exact sequence
	\begin{equation*}
		0 \longrightarrow \M[\nu] \xrightarrow{\partial_N+\nu/x_N} \M[\nu] \xrightarrow{\sum_j \nu^j m_j \mapsto \sum_j (- x_N \partial_N)^j m_j} \mathfrak{M} \longrightarrow 0
	\end{equation*}
	of $\WeylT{N-1}_{k[\nu]}$-modules, where $\mathfrak{M} = \M$ denotes the initial module $\M$ with the action of $\nu$ defined as $\nu m \defas -x_N \partial_N m$.
	Since $k(\nu)$ is flat, this sequence remains exact after tensoring with $k(\nu)$ over $k[\nu]$. Through identification of $\nu$ with $-\theta_N$, we conclude that
	\begin{equation*}
		\frac{\M\tp_k k(\nu)}{(\partial_N + \nu/x_N) (\M \tp_k k(\nu))}
		\isomorph \mathfrak{M} \tp_{k[\nu]} k(\nu)
		\isomorph 
		\M \tp_{k[\theta_N]} k(\theta_N)
		= \M(\theta_N)
	\end{equation*}
	are isomorphic as $\WeylT{N-1}_{k(\nu)}$-modules. The left hand side is the quotient $\M x_N^{\nu}/\partial_N \M x_N^{\nu}$ of the $\WeylT{N}_{k(\nu)}$-module $\M x_N^{\nu} \defas \M \tp_k k(\nu)$ defined by the original action of $\WeylT{N-1}_{k}$ and $x_N^{\pm 1}$ on $\M$, but twisting the operator $\partial_N$ to act like $\partial_N+ \nu/x_N$.\footnote{%
		This just encodes the natural action $\partial_N x_N^{\nu} m = x_N^{\nu}(\partial_N + \nu/x_N) m$ on products of elements $m$ of $\M$ with the function $x_N^{\nu}$---hence the suggestive notation $\M x_N^{\nu}$.
}
	The holonomicity of $\M$ implies that $\M x_N^{\nu}$ is also holonomic,\footnote{%
		Given a good filtration $\Gamma_{\bullet}$ of $\M$, $\Gamma_{j}' (\M x_N^{\nu}) \defas (x_N^{-j}\Gamma_{2j} \M) \tp_k k(\nu)$ defines a filtration of $\M x_N^{\nu}$ with $\dim_{k(\nu)} \Gamma_j' (\M x_N^{\nu}) \leq \dim_k \Gamma_{2j} \M \leq c \cdot (2j)^N$ for some $c<\infty$.
	}
	and hence its push-forward $\pi_{\ast} (\M x_N^{\nu}) = \M x_N^{\nu} / \partial_N (\M x_N^{\nu}) \isomorph \M(\theta_N)$, with respect the projection $\pi\colon \Gm[k(\nu)]^N \surjects \Gm[k(\nu)]^{N-1}$ that forgets the last coordinate, is also holonomic.
\end{proof}
Now we want to relate this dimension to the de Rham complex
$
	\DR(\M)
	=
	\Koszul{\M}{\partial}
$,
which is a special case of the Koszul complex:
\begin{defn}
	For commuting $k$-linear endomorphisms $s=(s_1,\ldots,s_N)$ of a $k$-vector space $\M$, let
\begin{equation}
	\Koszul{\M}{s_1,\ldots,s_N}
	\defas \big( \Lambda^{\bullet} k \tp_k \M[N], \dd \big)
	\label{eq:Koszul}%
\end{equation}
	denote the Koszul complex with $r$-forms sitting in degree $r-N$. The $r-N$ cochains
\begin{equation*}
	\Koszul[r-N]{\M}{s} 
	= \Lambda^{r} \M 
	= \bigoplus _{\abs{I}=r} \uv{I} \tp \M
\end{equation*}
have natural coordinates with respect to the basis $\uv{I} \defas \uv{i_1} \wedge \cdots \wedge \uv{i_r} \in \Lambda^{r} k$ indexed by $r$-sets $I=\set{i_1<\cdots<i_r}$. The cochain map reads $\dd (\uv{I} \tp m) = \sum_{i \notin I} (\uv{i} \wedge \uv{I}) \tp s_i m$.
\end{defn}
\begin{rem}\label{rem:theta-qis}
Since $\theta_j \left(\prod_{i\in I} x_i \right)m = \left(\prod_{i \in I\cup \set{j}} x_i\right) \partial_j m$ (for $j \notin I$), the rule
\begin{equation}
	\Koszul{\M}{\partial} \longrightarrow \Koszul{\M}{\theta},
	\quad
	\uv{I} \tp m \mapsto \uv{I} \tp x^I m
	\label{eq:theta-qis}%
\end{equation}
defines a cochain map. It has an inverse, defined by $\uv{I} \tp m \mapsto \uv{I} \tp x^{-I}m$. We thus conclude that $\DR(\M)=\Koszul{\M}{\partial}$ and $\Koszul{\M}{\theta}$ are quasi-isomorphic and therefore share the same Euler characteristic.
\end{rem}

We prove theorem~\ref{thm:L-S} by an induction over the number of variables. The base case is
\begin{thm}[{\cite[{Th\'{e}or\`{e}me~1}]{LoeserSabbah:IrredTore}}]
	\label{thm:L-S:1}%
	If $\M$ denotes a holonomic $\WeylT{1}_k$-module, then 
	\begin{equation}
		\dim_{k(\theta_1)} \M(\theta_1) 
		= \chi(\M)
		\defas \dim_k \frac{\M}{\partial_1 \M} - \dim_k \ker(\partial_1)
		.
	\end{equation}
\end{thm}
\begin{proof}
	We can pick a generator of $\M$ (by holonomicity, $\M$ is cyclic as a $\WeylT{1}_k$-module) and extend it to a (finite) basis of $\M(\theta_1)$ as a vector space over $k(\theta_1)$, due to corollary~\ref{cor:L-S:findim}. Let $\Nmod \subset \M$ denote the $k[\theta_1]$-module generated by such a basis, hence $\Nmod(\theta_1)=\M(\theta_1)$. By construction, $\M = \Weyl{1}_k \Nmod = \sum_{j\in \Z} x_1^j \Nmod$ is exhaustively filtered by the finitely generated $k[\theta_1]$-modules $\Nmod_j \defas \sum_{i=-j}^j x_1^i \Nmod$.\footnote{%
		Due to $\theta_1 x_1^{i} \Nmod = x_1^i (\theta_1+i) \Nmod \subseteq x_1^i\Nmod \subseteq \Nmod_j$, indeed $\Nmod_j$ is a $k[\theta_1]$-module.
	}

	Since $\M(\theta_1)=\Nmod(\theta_1) = \Nmod_1(\theta_1)$ is finitely generated, there is a non-zero polynomial $b(\theta_1)\in k[\theta_1]$ such that $b(\theta_1) \Nmod_1 \subseteq \Nmod$. Therefore, using $(\theta_1 -1)x_1 = x_1 \theta_1$,
	\begin{equation*}
		b(\theta_1 \mp j) x_1^{\pm (j+1)} \Nmod
		= x_1^{\pm j} b(\theta_1) x_1^{\pm 1} \Nmod
		\subseteq
		x_1^{\pm j} b(\theta_1) \Nmod_1
		\subseteq
		x_1^{\pm j} \Nmod
		\subseteq \Nmod_j
	\end{equation*}
	shows that the polynomials $b_{j+1}(\theta_1) \defas b(\theta_1+j)b(\theta_1-j) \in k[\theta_1]$ have the property $b_{j+1}(\theta_1)\Nmod_{j+1} \subseteq \Nmod_j$. Let $Z=b^{-1}(0)$ denote the zeroes of $b_1=b^2$, then note that the zeroes of $b_{j+1}$ are $(Z+j) \cup (Z-j)$ and get pushed away from zero for increasing $j$. In particular, there exists some $j_0 \in \N$ such that  $b_{j}(0) \neq 0$ for all $j > j_0$. For each such value of $j$, we can find $u_j,v_j \in k[\theta_1]$ such that $1 = u_j(\theta_1) b_j(\theta_1) + v_j(\theta_1) \theta_1$; then
	\begin{equation*}
		m = 1\cdot m
		= u_j(\theta_1) b_j(\theta_1) m + v_j(\theta_1) \theta_1 m
		\in \Nmod_{j-1} + v_j(\theta_1) \theta_1 m
	\end{equation*}
	holds for every $m \in \Nmod_j$. This proves $\ker(\theta_1) \cap \Nmod_j \subseteq \Nmod_{j-1}$ for all $j>j_0$, and therefore $\ker(\theta_1) \subseteq \Nmod_{j_0}$. Similarly, we conclude $\M/ (\theta_1 \M) \isomorph \Nmod_{j_0} / (\Nmod_{j_0} \cap \theta_1 \M$). But given some $m =\theta_1 x \in \Nmod_{j_0}$ with $x \in \Nmod_j$,
	$
		x = u_j(\theta_1) b_j(\theta_1) x + v_j(\theta_1) m
		\in \Nmod_{j-1} + \Nmod_{j_0}
	$
	proves that $\Nmod_{j_0} \cap \theta_1(\Nmod_j) = \Nmod_{j_0} \cap \theta_1(\Nmod_{j-1})$ for all $j>j_0$. In consequence, we have proven that
	\begin{equation*}
		\ker(\partial_1)
		= 
		\ker (\theta_1) = \ker\left(\restrict{\theta_1}{\Nmod_{j_0}}\right)
		\quad\text{and}\quad
		\frac{\M}{\partial_1(\M)}
		\isomorph
		\frac{\M}{\theta_1(\M)}
		\isomorph
		\frac{\Nmod_{j_0}}{\theta_1(\Nmod_{j_0})}
		;
	\end{equation*}
	in other words, the Koszul complexes $\DR(\M)=\Koszul{\M}{\partial_1}$ and $\Koszul{\Nmod_{j_0}}{\theta_1}$ are quasi-isomorphic (see remark~\ref{rem:theta-qis}). The statement of the theorem thus reduces to the identity
	\begin{equation*}
		\dim_{k(\theta_1)} \Nmod_{j_0}(\theta_1) 
		= \chi\left(\Koszul{\Nmod_{j_0}}{\theta_1}\right)
	\end{equation*}
	for a finitely generated $k[\theta_1]$-module $\Nmod_{j_0}$. Since both sides are additive under short exact sequences, this claim reduces (via a finite free resolution) to the case of a free rank one $k[\theta_1]$-module, i.e.\ $k[\theta_1]$, which is clear: $k[\theta_1](\theta_1) = k(\theta_1)$ is of dimension one over $k(\theta_1)$, while $\ker(\theta_1) = \set{0}$ is trivial and $k[\theta_1]/(\theta_1 k[\theta_1]) = k$ is one-dimensional.
\end{proof}
With this starting point, we can now prove theorem~\ref{thm:L-S} by induction. In fact, the higher dimensional case can be seen as a straightforward corollary of the univariate case above. In contrast to \cite{LoeserSabbah:IrredToreII}, our demonstration avoids any reference to higher-dimensional lattices.
\begin{proof}[Proof of theorem~\ref{thm:L-S}]
	Let $\M$ denote a holonomic $\WeylT{N}_k$-module, and suppose we have proven theorem~\ref{thm:L-S} for all holonomic modules in less than $N$ variables. In particular, we may invoke the claim for the $\WeylT{N-1}_k$-modules
	$\ker \partial_N$ and $\M/ \partial_N \M$,
	as these are holonomic because they are the cohomologies of the complex
	\begin{equation*}
		0 \longrightarrow \M \xrightarrow{\partial_N} \M \longrightarrow 0
	\end{equation*}
	which computes the push-forward of $\M$ along the projection $\pi\colon \Gm[k]^N \longrightarrow \Gm[k]^{N-1}$ that forgets the last coordinate. So we already know that 
	\begin{equation*}
		\chi\left(\frac{\M}{\partial_N \M} \right)
		=\dim_{k'} \left( \frac{\M}{\partial_N \M}(\theta') \right)
		=\dim_{k'} \frac{\M'}{\partial_N \M'}
	\end{equation*}
	where $k'\defas k(\theta')$ and $\M' \defas \M(\theta')$ with $\theta'\defas(\theta_1,\ldots,\theta_{N-1})$. 
	Analogously, $\chi(\ker \partial_N) = \dim_{k'} \ker(\partial_N')$, where $\partial_N'$ denotes the action of $\partial_N$ on $\M'$. In conclusion, we know that
	\begin{align*}
		\chi\left( \M/(\partial_N \M) \right)
		- \chi\left( \ker\partial_N \right)
		&=
		\dim_{k'} \M'/(\partial_N \M')
		- \dim_{k'} \ker(\partial_N')
		= \chi(\M')
		\\
		&= \dim_{k'(\theta_N)} \M'(\theta_N)
		=\dim_{k(\theta)} \M(\theta),
	\end{align*}
	where we recognized the first line as the Euler characteristic of the de Rham complex of the $\WeylT{1}_{k'}$-module $\M'$ and applied theorem~\ref{thm:L-S:1} to get to the last line ($\M'=\M(\theta')$ is holonomic by lemma~\ref{lem:L-S:holonomicity}). So we only need to show that the left hand side is equal to $\chi(\M)$.
	
	This is well-known and follows from the Grothendieck spectral sequence.\footnote{%
		Let $\pi^{N}\colon \Gm[k]^N\longrightarrow \set{\text{pt}}$ denote the projection to a point, such that $\pi^N = \pi^{N-1} \circ \pi$.
		The identity $\pi^{N}_+ = \pi^{N-1}_+ \circ \pi_{+}$ of the corresponding push-forwards in the derived category of $\WeylT{N}_k$-modules implies that $\chi(\DR(\M))=\chi(\pi^N_+(\M)) = \sum_i (-1)^i \chi( H^i (\pi_+ \M))$, where $H^{-1}(\pi_+ \M) = \ker \partial_N$ and $H^0(\pi_+ \M) = \M/(\partial_N \M)$.
	}
	Alternatively, an elementary way to obtain the identity $\chi(\M/(\partial_N\M))-\chi(\ker\partial_N)=\chi(\M)$ is given by the long exact sequence 
	\begin{equation}
		\cdots \rightarrow
		H^{i+1}(\DR(\ker \partial_N)) \rightarrow
		H^i(\DR(\M)) \rightarrow
		H^i(\DR(\M/\partial_N \M)) \rightarrow
		H^{i+2}(\DR(\ker \partial_N))
		\rightarrow\cdots
		\label{eq:les-relative-dR}%
	\end{equation}
	in de Rham cohomology \cite[Ch.~2, Proposition~4.13]{Bjork:RingsofDop}.
\end{proof}
For completeness, let us demonstrate \eqref{eq:les-relative-dR}, by following the standard construction in the proof of Kashiwara's theorem: First note that
	\begin{equation*}
		\Nmod \defas \setexp{m\in\M}{\partial_N^k m = 0 \quad\text{for some}\quad k>0}
		\subseteq \M
	\end{equation*}
	defines a $\WeylT{N}_k$-submodule of $\M$.\footnote{%
		One only needs to check that $x_N^{\pm 1} \Nmod \subseteq \Nmod$, which follows from $\partial_N^{k+1} x_N^{\pm 1} m = \big[(k+1)(\partial_N x_N^{\pm 1}) + x_N^{\pm 1} \partial_N) \partial_N^k m = 0$ whenever $\partial_N^k m =0$.
	}
	The exact sequence
$
		0 \rightarrow \Nmod \rightarrow \M \rightarrow\! \M/\Nmod \rightarrow 0
$
	of holonomic $\WeylT{N}_{k}$-modules induces a sequence of de Rham complexes,
$
		0 \rightarrow \DR(\Nmod) \rightarrow \DR(\M) \rightarrow \DR(\M/\Nmod) \rightarrow 0
$,
which is also exact ($\M \mapsto \DR(\M)=\Lambda^{\bullet} k \tp_k \M$ is exact by flatness of $\Lambda^{\bullet}k$).
	Hence we get a long exact sequence in cohomology:
	\begin{equation}
		\cdots \rightarrow
		H^i(\DR(\Nmod))\rightarrow
		H^i(\DR(\M)) \rightarrow
		H^i(\DR(\M/\Nmod)) \rightarrow
		H^{i+1}(\DR(\Nmod))
		\rightarrow \cdots
		\label{eq:les-relative-proof}%
	\end{equation}
	The key observation now is that $\partial_N$ is injective on $\M/\Nmod$ and surjective on $\Nmod$.\footnote{%
		The first statement is clear since $\ker\partial_N \subseteq \Nmod$. The second claim follows from the identity $x^k\partial^k = (\partial x -k)x^{k-1}\partial^{k-1} = \cdots = \prod_{i=1}^k (\partial x -i)$, which implies $0=x_N^k \partial_N^k m = (-1)^k (k!) m \mod \partial_N \Nmod$ whenever $\partial_N^k m = 0$.
	}
	We thus get short exact sequences
	\begin{align*}
		&
		0 \longrightarrow \ker\left( \partial_N \right) \injects \Nmod \xrightarrow{\partial_N} \Nmod \longrightarrow 0
		\quad\text{and}
		\\
		&
		0 \longrightarrow \M/\Nmod \xrightarrow{\partial_N} \M/\Nmod \surjects \M/\partial_N \M \longrightarrow 0
	\end{align*}
	of $\WeylT{N-1}_k$-modules. The induced short exact sequences of de Rham complexes provide quasi-isomorphisms $\DR(\ker \partial_N) \qiso \DR(\Nmod)[1]$ and $\DR(\M/\partial_N\M) \qiso \DR(\M/\Nmod)$, because the de Rham complex of a $\WeylT{N}_k$-module $\Nmod$ is the mapping cone of the map $\partial_N\colon \Koszul{\Nmod}{\partial'}\longrightarrow \Koszul{\Nmod}{\partial'}$.
	Hence we obtain \eqref{eq:les-relative-dR} from \eqref{eq:les-relative-proof} due to
	\begin{equation*}
		H^i(\DR(\Nmod)) \isomorph H^{i+1}(\DR(\ker \partial_N))
		\quad\text{and}\quad
		H^i(\DR(\M/\Nmod)) \isomorph H^i(\DR(\M/(\partial_N \M)))
		.
	\end{equation*}
To clarify this final step, first note that separating $\partial_N$ from $\partial'\defas (\partial_1,\ldots,\partial_{N-1})$ yields an isomorphism of $k$-vector spaces
\begin{align*}
	\Phi\colon
	\Koszul[\bullet-N]{\Nmod}{\partial'}
	\oplus \Koszul[\bullet-1-N]{\Nmod}{\partial'}
	&\stackrel{\isomorph}{\longrightarrow}
	\Koszul[\bullet-N]{\Nmod}{\partial}
	= \bigoplus_{ N \notin I} \uv{I} \tp \Nmod
	\oplus \bigoplus_{N \in I} \uv{I} \tp \Nmod
	\\
	x \oplus y
	\phantom{\Lambda^{r-1} \Koszul{\Nmod}{\partial'}}
	&\ \mapsto\ 
	x \oplus (\uv{N} \wedge y)
	.
\end{align*}
In this representation, the differential is given by 
\begin{equation*}
	\Phi^{-1} \left( \dd\Phi(x \oplus y) \right)
	= {\dd}'(x) \oplus (\partial_N x -{\dd}'(y))
	,
\end{equation*}
which is known as the mapping cone of $\partial_N\colon \Koszul{\Nmod}{\partial'} \longrightarrow \Koszul{\Nmod}{\partial'}$. Here we denote by ${\dd}'$ the differential of $\Koszul{\Nmod}{\partial'}$. 

If $\partial_N$ is surjective, we can find $x$ with $\partial_N x = y$ and hence $\dd \Phi(x\oplus 0)=\Phi({\dd}'x \oplus y)$ for every $y$. Therefore, every element of $\Koszul{\Nmod}{\partial}$ has a representative of the form $\Phi(x \oplus 0)$, modulo exact forms. But such a form is closed, $\dd \Phi(x\oplus 0)=0$, if and only if $x \in \ker \partial_N \cap \ker {\dd}'$.
\begin{cor}\label{cor:Koszul-qis}%
	If $\partial_N$ is surjective, then $\Koszul{\Nmod}{\partial}$ and $\Koszul{\ker\partial_N}{\partial'}[1]$ are quasi-isomorphic.
	If $\partial_N$ is injective, then $\Koszul{\Nmod}{\partial}$ and $\Koszul{\Nmod/\partial_N \Nmod}{\partial'}$ are quasi-isomorphic.
\end{cor}
The proof of the second statement is very similar to the surjective case and left as a straightforward exercise.

\section{A two-loop example}
\label{sec:two-loop example}

We demonstrate some main points of this article by a pedagogical example. Consider the massless two-loop two-point graph with five propagators, graph $\WS{3}'$ in figure~\ref{fig:wheels}. To this graph we associate the family of integrals
\begin{equation}
	\FI(\nu_1,\ldots,\nu_5)
	= \int\frac{\dd[\Dim] l_{1}}{\iu\pi^{\Dim/2}}\int\frac{\dd[\Dim] l_{2}}{\iu\pi^{\Dim/2}}
	\frac{1}{[-l_{1}^{2}]^{\nu_{1}} [-l_{2}^{2}]^{\nu_{2}} [-(l_{2}-p)^{2}]^{\nu_{3}} [-(l_{1}-p)^{2}]^{\nu_{4}} [-(l_{1}-l_{2})^{2}]^{\nu_{5}}} 
	\label{eq:Example momentum rep}%
\end{equation}
with two loop-momenta $q_1=l_1,\,q_2=l_2$ and one external momentum $q_3=p$. We normalize to $-p^2=1$. The graph polynomial $\G=\U+\F$ is given by the Symanzik polynomials 
\begin{align*}
	\U
	& = (x_1+x_4)(x_2+x_3)+x_5(x_1+x_2+x_3+x_4) \quad\text{and}
    \\
	\F
	    & = x_1 x_2(x_3+x_4)+x_3x_4(x_1+x_2)+x_5(x_1+x_2)(x_3+x_4).
\end{align*}
By definition \ref{def:Mellin transform} the modified Feynman integral
\begin{equation}
	\FImel(\nu_1,\ldots,\nu_5)
	=  \Mellin{\G^{s}}
	\label{eq:Example modified Feyn int}%
\end{equation}
is related to the Feynman integral by 
\begin{equation}
	\FI(\nu_1,\ldots,\nu_5)
	=\frac{\Gamma(-s)}{\Gamma(-s-\omega)}\FImel (\nu_1,\ldots,\nu_5)
\end{equation}
with $\omega = \nu_1+\nu_2+\nu_3+\nu_4+\nu_5+2s$ and $s=-\Dim/2$.

\subsection{From annihilators to integral relations}

A set of generators of the annihilator ideal $\Ann_{A_5[s]}\left(\G^s\right)$ can be derived in {\Singular} \cite{Singular} with algorithms introduced in \cite{Andetal}. Using the command \texttt{SannfsBM} we obtain a set of 13 generators. To give an impression, the first five of them read

\begin{align*}
	P_1
	& = \partial_1-\partial_2+(x_3+x_4+1)(\partial_3-\partial_4)+(x_3-x_4)\partial_5,
\\
	P_2
    & = (x_2-1)\partial_1+(-2x_1-3x_2-1)\partial_2+(-x_4+1)\partial_3+(2x_3+3x_4+1)\partial_4 +(2x_1-x_2 \\
    & -2x_3+x_4)\partial_5,
\\
P_3
& = (x_1+2)\partial_1+(x_1+2 x_2)\partial_2+x_4 \partial_3 +(-2 x_3-3 x_4-2)\partial_4+(-x_1+2 x_3-x_4)\partial_5,
\\
P_4
& = \partial_1 x_1+(x_2+1)\partial_2+(-x_3-x_4-1)\partial_3+(x_4+x_5)\partial_5-s,
\\
P_5
& = (x_1 x_3+x_1 x_4+x_2 x_3+x_2 x_4+x_1+x_2+x_3+x_4)\partial_4 +(-x_1 x_2-x_1 x_3-x_1 x_5\\
& -x_2 x_3-x_2 x_5-x_2-x_3-x_5)\partial_5.
\end{align*}
We notice that this set includes one generator which is quadratic in the differential operators, reading
\begin{align*}
P_{13}
& = (x_4+x_5-2)\partial_1^2+(2 x_2+x_5)\partial_2^2 +(-2 x_3^2+2 x_3 x_4-x_3 x_5+x_4 x_5)\partial_5^2-x_3 x_4\partial_3\partial_5\\
&+(2 x_3^2+x_3 x_4-2 x_4^2+2 x_3-2 x_4)\partial_4\partial_5+(x_3-x_4+(-2 x_2-x_4-2 x_5+2)\partial_2\\
&+(x_4^2-x_4)\partial_3+(-x_4^2+2 x_3+3 x_4+2)\partial_4+(x_3 x_4+x_3 x_5-x_4^2-x_4 x_5-4 x_3+2 x_4\\
&-x_5)\partial_5)\partial_1+(x_3 x_4\partial_3-x_3+x_4+(-x_3 x_4-2 x_3-2 x_4-2)\partial_4+(-2 x_2 x_3+2 x_2 x_4\\
&-x_3 x_5+x_4 x_5+2 x_3+x_5)\partial_5+2)\partial_2+(-x_3+x_4)\partial_5.
\end{align*}

Every operator in ${\Ann_{A_5[s]}\left(\G^s\right)}$ gives rise to an integral relation. According to lemma \ref{lem:Mellin-properties}, we just need to replace each $x_i$ by $\PlusD{i}$ and each $\partial_i$ by $-\Minus{i}$ to obtain a shift relation between modified Feynman integrals. For example, for the generator $P_1$, we obtain the shift operator
\begin{equation*}
	\Mellin{P_1}
	=-\Minus{1}+\Minus{2}-(\PlusD{3}+\PlusD{4}+1)(\Minus{3}-\Minus{4})-(\PlusD{3}-\PlusD{4})\Minus{5}
\end{equation*}
satisfying 
\begin{equation*}
	\Mellin{P_1}
	\FImel(\nu_1,\ldots,\nu_5)=0.
\end{equation*}
For the (unmodified) Feynman integral $\FI(\nu_1,\ldots,\nu_5)$ we obtain the corresponding shift relation via \eqref{eq:FImel-relation-to-FI}. From $P_1$ we obtain the operator 
\begin{equation}
	\boldsymbol{O_1}
	= \left(\sum_{i=1}^5\nOp{i}+3s\right)(\Minus{1}-\Minus{2}+\Minus{3}-\Minus{4})-(\PlusD{3}+\PlusD{4})(\Minus{3}-\Minus{4})-(\PlusD{3}-\PlusD{4})\Minus{5}
\end{equation}
which satisfies
\begin{equation} 
	\boldsymbol{O_1}\FI(\nu_1,\ldots,\nu_5)=0.
	\label{eq:Example-relO1}%
\end{equation}

\subsection{Linear annihilators}
\label{sec:Linear annihilators}%

It is useful to consider the linear annihilators $\Ann^1_{\Weyl{N}[s]}(\G^s) \subseteq \Ann_{\Weyl{N}[s]}(\G^s).$ Recall that they are the annihilators which are linear differential operators, being of the form $P=q + \sum_{i=1}^N p_i \partial_i$ with $p,q_1,\ldots,q_N \in \Q[s,x_1,\ldots,x_N]$. We compute the generators of $\Ann^1(\G^s)$ as generators of the Syzygy-module of $(\G,\partial_1 \G,\ldots,\partial_N \G)$, using the {\Singular} command \texttt{syz}. For our example we obtain the following 8 generators:
\begin{align*}
L_1
&= \partial_1-\partial_2+(x_3+x_4+1)\partial_3+(-x_3-x_4-1)\partial_4+(x_3-x_4)\partial_5,\\
L_2
&= (-x_2-1)\partial_1+(x_2+1)\partial_2+(-x_4-1)\partial_3+(x_4+1)\partial_4+(x_2+x_4+2 x_5)\partial_5-2s,\\
L_3
&= (2 x_1+x_2+1)\partial_1+(x_2+1)\partial_2+(-2 x_3-x_4-1)\partial_3+(-x_4-1)\partial_4+(-x_2+x_4)\partial_5,\\
L_4
&= (x_1+x_2+1)\partial_1+(-x_1-x_2-1)\partial_2+\partial_3-\partial_4+(x_1-x_2)\partial_5,\\
L_5
&= -2s x_4+(x_4+x_5)\partial_1+(-2 x_2-x_4-x_5-2)\partial_2+(2 x_3 x_4+x_4^2+2 x_3+3 x_4+2)\partial_3\\
& +(x_4^2+x_4)\partial_4+(-x_4^2-2 x_4-x_5)\partial_5,\\
L_6
&= (-x_2 x_3+x_2 x_4-x_2-x_3+2 x_4+2 x_5)\partial_1+(-2 x_1 x_3+2 x_1 x_4-x_2 x_3+x_2 x_4+2 x_1\\
& -9 x_2-x_3-2 x_5-6)\partial_2+(-2 x_1 x_3-2 x_1 x_4-x_2 x_3-x_2 x_4-2 x_3^2-4 x_3 x_5+2 x_4^2\\
& -4 x_4 x_5-2 x_1-x_2+3 x_3+6 x_4-4 x_5+6)\partial_3+(2 x_1 x_3+2 x_1 x_4+x_2 x_3+x_2 x_4\\
& +4 x_3 x_4+4 x_3 x_5+4 x_4^2+4 x_4 x_5+2 x_1+x_2+x_3+6 x_4+4 x_5)\partial_4+(-2 x_4^2-6 x_4\\
& -4 x_5)\partial_5+s(-4 x_4+2),\\
L_7
&= (2 x_1 x_2+2 x_2^2+2 x_1+4 x_2+2)\partial_1+(x_2 x_4+2 x_2+x_5)\partial_3+(-x_2 x_4-2 x_2-2 x_4\\
& -x_5-2)\partial_4+(-2 x_2^2-x_2 x_4-2 x_2 x_5-2 x_2-x_5)\partial_5,\\
L_8
&= (x_4+x_5)\partial_1+(2 x_1 x_2+2 x_1 x_4+2 x_1 x_5+2 x_2^2+2 x_2 x_4+2 x_2 x_5+2 x_1+2 x_2+x_4\\
& +x_5)\partial_2+(-2 x_1 x_3-2 x_1 x_4-2 x_3 x_5+x_4^2-2 x_1+2 x_3-x_4)\partial_3+(-x_4^2-2 x_4 x_5-x_4\\
& -2 x_5)\partial_4+(-x_4^2-2 x_4 x_5-2 x_5^2+x_5)\partial_5+s(-2 x_2+2 x_5-2),\\
\end{align*}
Using {\Singular} we find that for our two-loop example every generator ${P_i}$ can be expressed as a linear combination of the linear generators ${L_i}$ over $\Ann_{\Weyl{5}[s]}(\G^s)$. For instance, the first five annihilators satisfy
\begin{align*}
	P_1 &= L_1, \quad
	P_2 = -2L_1-L_3+2L_4, \quad
	P_3 = 2L_1+L_3-L_4, \quad
	2P_4 = L_2+L_3, 
	\\
	2P_5 &= -2(x_1+x_2+1)L_1+(1+x_2-x_4-x_5)(L_2-L_3)-2(x_2+x_4+x_5)L_4+L_5+2L_7-L_8.
\end{align*}
We emphasize that such a relation exists for all the $P_i$. In particular such a relation also exists for the quadratic $P_{13}$, which however is too long to be shown here. As a consequence we can view $\Ann_{\Weyl{N}[s]}(\G^s)$ as generated by the linear $L_i$, which will simplify the discussion in section \ref{sec:Comparing}.

\subsection{From IBP relations to annihilators}
Going in the other direction, we can derive annihilators from momentum-space IBP relations. In the usual way, inserting the differential operators 
\begin{equation}
	\momIBP{i}{j}
	= \frac{\partial}{\partial q_i}q_j
	\quad \text{for}\quad i\in \set{1,2} \quad\text{and}\quad j\in \set{1,2,3}
\end{equation}
we obtain six IBP relations $\shiftIBP{i}{j} \FI=0$ with the shift operators 
\begin{align*}
	\shiftIBP{1}{1} 
	&=-\PlusD{4}\Minus{1}-\PlusD{5}\Minus{1}+\PlusD{5}\Minus{2}+\PlusD{4}-2\nOp{1}-\nOp{4}-\nOp{5},
	\\
	\shiftIBP{1}{2} 
	&=-\PlusD{1}\Minus{2}+\PlusD{1}\Minus{5}-\PlusD{4}\Minus{1}-\PlusD{4}\Minus{3}+\PlusD{4}\Minus{5}-\PlusD{5}\Minus{1}+\PlusD{5}\Minus{2}+\PlusD{4}-\nOp{1}+\nOp{5},
	\\
	\shiftIBP{1}{3} 
	&=-s\PlusD{1}+s\PlusD{4}-\PlusD{4}\Minus{1}-\PlusD{5}\Minus{1}+\PlusD{5}\Minus{2}-\PlusD{5}\Minus{3}+\PlusD{1}\Minus{4}+\PlusD{5}\Minus{4}-\nOp{1}+\nOp{4},
	\\
	\shiftIBP{2}{1} 
	&=-\PlusD{2}\Minus{1}+\PlusD{2}\Minus{5}-\PlusD{3}\Minus{2}-\PlusD{3}\Minus{4}+\PlusD{3}\Minus{5}+\PlusD{5}\Minus{1}-\PlusD{5}\Minus{2}+\PlusD{3}-\nOp{2}+\nOp{5},
	\\
	\shiftIBP{2}{2} 
	&=-\PlusD{3}\Minus{2}+\PlusD{5}\Minus{1}-\PlusD{5}\Minus{2}+\PlusD{3}-2s-2\nOp{2}-\nOp{3}-\nOp{5}
	\\
	\shiftIBP{2}{3} 
	&=\PlusD{2}\Minus{3}-\PlusD{3}\Minus{2}+\PlusD{5}\Minus{1}-\PlusD{5}\Minus{2}+\PlusD{5}\Minus{3}-\PlusD{5}\Minus{4}-\PlusD{2}-\PlusD{3}-\nOp{2}+\nOp{3}
	.
\end{align*}

Following the steps in the proof of corollary \ref{cor:parIBP} we derive for each shift opeartor $\shiftIBP{i}{j}$ a parametric annihilator $\parIBP{i}{j}$. We obtain
\begin{align*}
\parIBP{1}{1}
&= (x_1 x_4+2 x_1+x_4+x_5)\partial_1+(x_2 x_4-x_5)\partial_2+x_3 x_4\partial_3-3s x_4+(x_4^2+x_4)\partial_4+(x_4 x_5\\
& +x_5)\partial_5-2s,\\
\parIBP{1}{2}
&= -3s x_4+(x_1 x_4+x_1+x_4+x_5)\partial_1+(x_2 x_4+x_1-x_5)\partial_2+(x_3 x_4+x_4)\partial_3+x_4^2\partial_4\\
& +(x_4 x_5-x_1-x_4-x_5)\partial_5,\\
\parIBP{1}{3}
&= 3s x_1-3s x_4+(-x_1^2+x_1 x_4+x_1+x_4+x_5)\partial_1+(-x_1 x_2+x_2 x_4-x_5)\partial_2+(-x_1 x_3\\
& +x_3 x_4+x_5)\partial_3+(-x_1 x_4+x_4^2-x_1-x_4-x_5)\partial_4+(-x_1 x_5+x_4 x_5)\partial_5,\\
\parIBP{2}{1}
&= -3s x_3+(x_1 x_3+x_2-x_5)\partial_1+(x_2 x_3+x_2+x_3+x_5)\partial_2+x_3^2\partial_3+(x_3 x_4+x_3)\partial_4\\
& +(x_3 x_5-x_2-x_3-x_5)\partial_5,\\
\parIBP{2}{2}
&= -3s x_3+(x_1 x_3-x_5)\partial_1+(x_2 x_3+2 x_2+x_3+x_5)\partial_2+(x_3^2+x_3)\partial_3+x_3 x_4\partial_4+(x_3 x_5\\
& +x_5)\partial_5-2s,\\
\parIBP{2}{3}
&= 3s x_2-3s x_3+(-x_1 x_2+x_1 x_3-x_5)\partial_1+(-x_2^2+x_2 x_3+x_2+x_3+x_5)\partial_2+(-x_2 x_3\\
& +x_3^2-x_2-x_3-x_5)\partial_3+(-x_2 x_4+x_3 x_4+x_5)\partial_4+(-x_2 x_5+x_3 x_5)\partial_5.
\end{align*}
These operators are useful to compare both approaches as discussed next. 

\subsection{Comparing annihilators and IBP operators} \label{sec:Comparing}

According to corollary \ref{cor:parIBP}, every momentum-space IBP relation corresponds to a parametric annihilator. For our two-loop example, this is given by the fact that 
\begin{equation*}
	\parIBP{i}{j}
	\in \Ann_{\Weyl{5}[s]}\left(\G^s\right)
	\quad\text{for}\quad
	i\in \set{1,2}
	\quad\text{and}\quad
	j\in \set{1,2,3}.
\end{equation*}

We may furthermore ask if the reverse is true: Can every annihilator of $\G$ be derived from IBP relations? If the answer would be no, the approach via parametric annihilators would provide new integral identities. While this question remains open for the general case, we can test it for simple Feynman graphs such as the present two-loop example.

In a first attempt, we could consider the shift relations obtained from the generators ${P_1,\ldots,P_{13}}$ and try to confirm that they are combined IBP relations. If we use one of the well-known implementations of Laporta's algorithm to reproduce e.g.\ equation~\eqref{eq:Example-relO1}, we have to fix the values of $\nu_1,\ldots,\nu_5$ and do not answer the question for arbitrary values of the $\nu_i$. We therefore approach the problem on the level of parametric differential operators instead. 

We find that actually not all parametric annihilators are contained in $\Mom$; however, they turn out to still be consequences of the momentum space IBP relations in the following sense:
While we checked that $P_1 \notin \Mom$, we can find a polynomial $q_1\in\Q[s,x_1\partial_1,\ldots,x_5\partial_5]$ such that $q_1 P_1\in \Mom$. Recall that, under the Mellin transform, such a $q_1$ corresponds to a polynomial in the dimension and in the $\nu_e$. 
The interesting question then is if we can find a polynomial $q\in\Q[s,x_1\partial_1,\ldots,x_N\partial_N]$ for every $P\in\Ann_{\Weyl{N}[s]}(\G^s)$ such that $q P \in \Mom$. If we can find such a $q_i$ for every generator $P_i$, we can express every annihilator in terms of the $\parIBP{i}{j}$. The $q_i$ are the denominators of the coefficients in such a linear combination.

In section \ref{sec:Linear annihilators} we have seen for our example that $\Ann_{\Weyl{N}[s]}(\G^s)$ is generated by the linear annihilators $L_i$. As a consequence, it is sufficient to show that for each $L_i$ there is a $\tilde{q}_i$ such that $\tilde{q}_i L_i \in \Mom$ for $i=1,\ldots,8$. Indeed we can construct such ${q_1,\ldots,q_8}$ by an explicit Ansatz. For example we obtain the identity
\begin{equation*}
	\left(2\sum_{i=1}^5x_i \partial_i -6s\right)L_1= c_1 \parIBP{1}{1} +c_2 \parIBP{2}{1}+c_3 \parIBP{1}{2} + c_4 \parIBP{2}{2} + c_5 \parIBP{1}{3} + c_6 \parIBP{2}{3}
\end{equation*}
with
\begin{align*}
c_1 
&= -(-2 \partial_1+ \partial_2- \partial_3+2 \partial_4+x_5 ( \partial_1 (s+1)+ \partial_2 (-s-1)+ \partial_5s) +x_4 \partial_4+x_3 ( \partial_2 (-s-1)\\
& - \partial_4s
+ \partial_5s+3s^2+3s+x_5 \partial_5 (-s-1)+x_4 \partial_4 (-s-1))+x_3^2 \partial_3 (-s-1)+x_2 (- \partial_1s\\
&+ \partial_2 (
-s-1)+ \partial_5s+x_3 \partial_2 (-s-1))+x_1 ( \partial_1 +x_3 \partial_1 (-s-1))),\\
c_2
&= -(- \partial_2+ \partial_3-2s^2-2s+
x_5 \partial_5 (2s+1)+x_4 (- \partial_3s+ \partial_4 (s+1) + \partial_5s)+x_3 \partial_3+x_2 \partial_2\\
&+x_1 ( \partial_1 (s+1)- \partial_2s+
 \partial_5s)),\\
c_3
 &=-( \partial_1+1- \partial_4+s +x_5 ( \partial_1 (-s-1)+ \partial_2 (s+1)- \partial_5s)-x_4 \partial_4+x_3 ( \partial_2 (s+
1)+ \partial_4s\\
&- \partial_5s-3s^2-3s+x_5 \partial_5 (s+1)+x_4 \partial_4 (s+1))+x_3^2 \partial_3 (s+1)+x_2 ( \partial_1s+
 \partial_2 (s+1)\\
 &- \partial_5s+x_3 \partial_2 (s+1))+x_1 (- \partial_1+x_3 \partial_1 (s+1)))\\
c_4 
&= -(- \partial_1+2 \partial_2-2 \partial_3
+ \partial_4+x_5 \partial_5+x_4 (- \partial_3+ \partial_5)-x_3 \partial_3-x_2 \partial_2+x_1 (- \partial_2+ \partial_5))\\
c_5
&=-( \partial_1- \partial_2+ \partial_5)\\
c_6
&=
-( \partial_1- \partial_2- \partial_5).
\end{align*}
Using these results and the expressions for the generators of $\Ann_{\Weyl{5}[s]}(\G^s)$ in terms of the $L_i$, we can derive every annihilator from the IBP operators. We have done this computation with the same conclusion for several further graphs of low loop-order. These computations support our conjectures phrased in questions \ref{con:Ann=Ann1} and \ref{con:Ann=Mom}.

\subsection{The number of master integrals}

The result of the previous subsection implies that the parametric approach and momentum-space IBP lead to the same number of master integrals for this example. Indeed, as mentioned in section \ref{sec:lin-red}, we compute the Euler characteristic
\begin{equation*}
	\NoM{\G} = 3
\end{equation*}
and obtain the same number of master integrals with {\Azurite}. The three master integrals suggested by {\Azurite} are
\begin{equation*}
	I_1=\FI(1,1,1,1,0),\quad
	I_2=\FI(0,1,0,1,1) \quad \text{and} \quad 
	I_3=\FI(1,0,1,0,1).
\end{equation*}
Notice that symmetries of the graph were not taken into account here, which in {\Azurite} is assured by setting \texttt{Symmetry -> False} and \texttt{GlobalSymmetry -> False}. For an integral reduction in practice, one would of course make use of the symmetry  
\begin{equation*}
	\FI(\nu_1,\nu_2,\nu_3,\nu_4,\nu_5)=\FI(\nu_2,\nu_1,\nu_4,\nu_3,\nu_5) 
\end{equation*}
and compute with one of the sets $\set{I_1,I_2}$, $\set{I_1,I_3}$.

\bibliographystyle{JHEPsortdoi}
\bibliography{refs}

\end{document}